%% file: main_arxiv.tex
\documentclass[manuscript, nonacm]{acmart}

\PassOptionsToPackage{prologue,dvipsnames}{xcolor}
\usepackage{xcolor}
\usepackage{graphicx}
\usepackage{bm}
\usepackage{amsmath}
\usepackage{amsthm}
\usepackage{tikz}
\usepackage{hyperref}
\usepackage{multirow}
\usepackage{makecell}
\usepackage{booktabs}
\usepackage{algorithm}
\usepackage{algorithmicx}
\usepackage[noend]{algpseudocode}
\usepackage{algpseudocodex}
\usepackage{subcaption}
\usepackage{listings}
\usepackage{hhline}
\usepackage{xspace}
\usepackage{tabularx}
\usepackage{arydshln}
\input{macros}

\begin{document}
\title{Advancing Fact Attribution for Query Answering: Aggregate Queries and Novel Algorithms}

\pagestyle{empty}        
\begin{titlepage}

\vfill
\end{titlepage}

\clearpage
\pagestyle{fancy}

\author{Omer Abramovich}
\affiliation{%
  \institution{Tel Aviv University}
  \city{Tel Aviv}
  \country{Israel}
}
\author{Daniel Deutch}
\affiliation{%
  \institution{Tel Aviv University}
  \city{Tel Aviv}
  \country{Israel}
}
\author{Nave Frost}
\affiliation{%
  \institution{eBay Research}
  \city{Netanya}
  \country{Israel}
}
\author{Ahmet Kara}
\affiliation{%
  \institution{OTH Regensburg}
  \city{Regensburg}
  \country{Germany}
}
\author{Dan Olteanu}
\affiliation{%
  \institution{University of Zurich}
  \city{Zurich}
  \country{Switzerland}
}

\begin{abstract}
In this paper, we introduce a novel approach to computing the contribution of input tuples to the result of the query, quantified by the Banzhaf and Shapley values. In contrast to prior algorithmic work that focuses on Select-Project-Join-Union queries, ours is the first practical approach for queries with aggregates. It relies on two novel optimizations that are essential for its practicality and  significantly improve the runtime performance already for queries without aggregates. The first optimization exploits the observation that many input tuples have the same contribution to the query result, so it is enough to compute the contribution of one of them. The second optimization uses the gradient of the query lineage to compute the contributions of all tuples with the same complexity as for one of them. 
Experiments with a million instances over 3 databases show that our approach achieves up to 3 orders of magnitude runtime improvements over the state-of-the-art for queries without aggregates, and that it is practical for aggregate queries.

\end{abstract}

\maketitle
 
\setcounter{page}{1}
\input{Sections_Arxiv_New/Introduction}
\input{Sections_Arxiv_New/Preliminaries}

\input{Sections_Arxiv_New/Boolean}

\input{Sections_Arxiv_New/Aggregates}

\input{Sections_Arxiv_New/Experiments}
\input{Sections_Arxiv_New/Related}
\input{Sections_Arxiv_New/Conclusions}

\begin{acks}
    Ahmet Kara and Dan Olteanu would like to acknowledge the UZH Global Strategy and Partnerships Funding Scheme. The work of  Omer Abramovich and Daniel Deutch was partially funded by the European Research Council (ERC) under the European Union's Horizon 2020 research and innovation programme (Grant agreement No. 804302). 
\end{acks}

\bibliographystyle{ACM-Reference-Format}
\bibliography{bibliography}
 \pagebreak

 \appendix
\input{Sections_Arxiv_New/App_Algorithm_Arxiv}
\end{document}

%% file: macros.tex
\newcommand{\mods}{\mathsf{models}}
\newcommand{\vars}{\mathsf{vars}}
\newcommand{\facts}{\mathsf{facts}}

\newcommand{\lin}{\mathsf{lin}}
\newcommand{\Const}{\mathsf{Const}}

\newcommand{\maxagg}{\mathrm{MAX}}
\newcommand{\minagg}{\mathrm{MIN}}
\newcommand{\sumagg}{\mathrm{SUM}}
\newcommand{\countagg}{\mathrm{COUNT}}
\newcommand{\groupbyagg}{\mathrm{GROUP\ BY}}
\newcommand{\posbool}{\mathsf{PosBool}}
\newcommand{\bool}{\mathsf{Bool}}
\newcommand{\banz}{\mathsf{Banzhaf}}
\newcommand{\shap}{\mathsf{Shapley}}
\newcommand{\dislift}{\mathsf{lift}\text{-}\mathsf{or}}
\newcommand{\conlift}{\mathsf{lift}\text{-}\mathsf{and}}
\newcommand{\lift}{\mathsf{Lift}}

\newcommand{\liftedcompile}{\mathsf{LiftedCompile}}
\newcommand{\inline}{\mathsf{inline}}
\newcommand{\compile}{\_\mathsf{Compile}}
\newcommand{\gradientBanzhaf}{\mathsf{GradientBanzhaf}}
\newcommand{\pv}{\mathsf{PV}}

\definecolor{goodgreen}{rgb}{0.1, 0.5, 0.1}
\newcommand{\STAB}{\makebox[1ex][r]{}}%
\newcommand{\TAB}{\makebox[2.5ex][r]{}}%
\newcommand{\SWITCH}{\textbf{switch}\xspace}%
\newcommand{\CASE}{\textbf{case}\xspace}%
\newcommand{\DEFAULT}{\textbf{default}\xspace}%
\algrenewcommand\algorithmicrequire{\textbf{Input:}}
\algrenewcommand\algorithmicensure{\textbf{Output:}}

\newcommand{\adaban}{\textsc{AdaBan}\xspace}

\newcommand{\exaban}{\textsc{ExaBan}\xspace}
\newcommand{\lexaban}{\textsc{LExaBan}\xspace}
\newcommand{\exashap}{\textsc{ExaShap}\xspace}
\newcommand{\lexashap}{\textsc{LExaShap}\xspace}

\newcommand{\calB}{\mathcal B}

\theoremstyle{definition}
\newtheorem{theorem}{Theorem}[section]
\newtheorem{lemma}[theorem]{Lemma}
\newtheorem{proposition}[theorem]{Proposition}

\newtheorem{corollary}[theorem]{Corollary}         
\newtheorem{definition}[theorem]{Definition}
\newtheorem{example}[theorem]{Example}
\newtheorem{remark}[theorem]{Remark}

\newcommand{\omer}[1]{\todo[inline,color=cyan]{\textsf{#1} \hfill \textsc{--Omer.}}}

\definecolor{verylightgray}{gray}{0.90}
\newcommand{\ahmet}[1]{\todo[inline,color=verylightgray]{\textsf{#1} \hfill \textsc{--Ahmet.}}}

\newcommand{\daniel}[1]{\todo[inline,color=orange]{\textsf{#1} \hfill \textsc{--Daniel.}}}

\newcommand{\nop}[1]{}

\newcommand{\tensorset}{BNP\xspace}
\newcommand{\boolnumpairs}{Boolean$\otimes$Number Pairs\xspace}
\newcommand{\boolnumpair}{Boolean$\otimes$Number Pair\xspace}

%% file: Sections_Arxiv_New/Introduction.tex
\section{Introduction}
\label{sec:Introduction}

Recent years have witnessed a surge of development that uses game theoretic measures to quantify the contribution of each database tuple to the query answer, also referred to as {\em attribution}. The Banzhaf ~\cite{Penrose:Banzhaf:1946, Banzhaf:1965} and Shapley ~\cite{shapley1953value} values are two prominent examples of such measures. They originate in cooperative game theory, where they are used to quantify the contribution of a player to the game outcome by summing up the marginal contributions of the player over the subsets (or permutations, for Shapley) of players. 
Prior work used Banzhaf and Shapley values for attribution in query answering: the contribution of a database tuple to the query result is quantified by defining a game where the players are database tuples and the objective function is given by each tuple in the query result~\cite{TheShapleyValueofTuplesinQueryAnswering, ComputingTheShapleyValueOfFactsInQueryAnswering, Sig24:BanzhafValuesForFactsInQueryAnswering}. Banzhaf and Shapley values are  different, yet closely related measurements. Previous works have demonstrated the usefulness of both values for explanations in query answering ~\cite{TheShapleyValueofTuplesinQueryAnswering, Sig24:BanzhafValuesForFactsInQueryAnswering, FastShapleyValueComputationinDataAssemblage, ShapleyValueInDataManagement, ShapleyvaluecomputationinontologyMediatedqueryanswering, ThecomplexityoftheShapleyvalueforregularpathqueries, FromShapleyValueToModelCountingAndBack, Sig24:ExpectedShapleyLikeScores, TheimpactofnegationonthecomplexityoftheShapleyvalueinconjunctivequeries}
and for feature importance, data valuation, feature selection, and explanations in machine learning~~\cite{sun2018banzhaf, fryer2021shapley,rozemberczki2022shapley, wang2023dataBanzhaf}.
The main difference is that the sum of Shapley values over all input tuples is equal to the query result (this is 1 in case of a Boolean query and the numeric result for an aggregate query), whereas this is not the case for Banzhaf. This may be considered an advantage of Shapley values, in that they better capture relative contributions. Nevertheless, prior work \cite{Sig24:BanzhafValuesForFactsInQueryAnswering} showed experimentally that they tend to agree on the relative order of fact contribution in query answering. 

Yet computing Banzhaf values is typically more efficient than computing Shapley values \cite{Sig24:BanzhafValuesForFactsInQueryAnswering}. Our work observes this for aggregate queries, too.




\begin{example}\label{ex:intro-example1}
Figure~\ref{fig:running_example} depicts a fragment of an IMDB-like database about movies directed by Tarantino, award nominations, and cast. We consider two queries.
Query $Q_1$ asks whether Tarantino was nominated for a “best director” award for movies with actors from the Cast relation. The answer is true: there are two such movies in the fragment (and three in the full database). We would like to understand which actor contributed most to this answer, using Banzhaf/Shapley values as the  measure. In our fragment, both Brad Pitt and Zoë Bell contributed most as they appear in the two movies. If we consider the full database but restrict the MovieCast relation to only include  the top-5 actors (first 5 listed in IMDB) for each movie, the most influential actor is Brad Pitt, who had a leading role in two of the three movies.
Query $Q_2$ retrieves the maximum revenue of a movie directed by Tarantino. Again, we seek to identify the actor who contributed most to the query answer. In our fragment, Brad Pitt and Zoë Bell are the top contributors according to the Banzhaf/Shapley values. For the full database (again with only the top-5 actors in MovieCast), the top contributor is Samuel L. Jackson.
\end{example}

\begin{figure*}[!t]
    \scriptsize
    \centering
    \begin{minipage}{0.20\linewidth}
    \centering
    \caption*{\texttt{DirectingAwards} (\normalfont{endo})}
    \begin{tabular}{c|c|}
    \hhline{~-}
    & \textbf{Award} \\
    \hline
    $d_1$ & Academy \\
    $d_2$ & BAFTA \\
    \hline
    \end{tabular}
    \end{minipage}%
    \hfill
    \begin{minipage}{0.40\linewidth}
    \centering 
    \caption*{\texttt{Movies} (\normalfont{endo})}
    \begin{tabular}{c|r|r|r|}
    \hhline{~---}
    & \textbf{Title} & \textbf{Director} & \textbf{Gross revenue(M\$)}\\
    \hline
    $m_1$ & Kill Bill: Vol. 1 &  Tarantino& 176 \\
    $m_2$ & Inglourious Basterds &  Tarantino & 322 \\
    $m_3$ & Once Upon a Time in Hollywood &  Tarantino & 377 \\
    \hline
    \end{tabular}
    \end{minipage}%
    \hfill
    \begin{minipage}{0.3\linewidth}
    \centering
    \caption*{\texttt{Cast} (\normalfont{endo})}
    \begin{tabular}{c|r|r|}
    \hhline{~-}
    & \textbf{Name} \\
    \hline
    $a_1$ & Brad Pitt \\
    $a_2$ & Leonardo DiCaprio \\
    $a_3$ & Zoë Bell \\
    $a_4$ & Uma Thurman\\

    \hline
    \end{tabular}
    \end{minipage}

    \begin{minipage}{0.30\linewidth}
    \centering
    \caption*{\texttt{MovieAwards} (\normalfont{exo})}
    \begin{tabular}{|r|r|}
    \hhline{--}
    \textbf{Movie} & \textbf{Award} \\
    \hline
    Inglourious Basterds & Academy\\
    Inglourious Basterds & BAFTA \\
    Once Upon a Time in Hollywood & Academy \\
    Once Upon a Time in Hollywood & BAFTA \\
    \hline
    \end{tabular}
    \end{minipage}%
    \hfill
    \begin{minipage}{0.19\linewidth}
    \centering
    \caption*{\texttt{Query $Q_1$}}
    $\begin{aligned} 
    Q() &=\, \texttt{Movies}(X, \text{Tarantino}, R),\\ 
              &  \texttt{Cast}(Y), \\
              &  \texttt{DirectingAwards}(Z), \\
              &  \texttt{MovieCast}(X,Y), \\
              & \texttt{MovieAwards}(X,Z)
    \end{aligned}$
    \end{minipage}%
    \hfill
    \begin{minipage}{0.20\linewidth}
    \centering
    \caption*{\texttt{Query $Q_2$}}
    $\begin{aligned}
    \big\langle MA&X, R,\\ Q&(X,Y,R)=\\
    & \texttt{Movies}(X, \text{Tarantino}, R),\\ 
  &  \texttt{Cast}(Y), \\
  &  \texttt{MovieCast}(X,Y)\big\rangle
    \end{aligned}$
    \end{minipage}%
    \hfill
    \begin{minipage}{0.30\linewidth}
    \centering
    \caption*{\texttt{MovieCast} (\normalfont{exo})}
    \begin{tabular}{|r|r|}
    \hhline{--}
     \textbf{Movie} & \textbf{Actor} \\
    \hline
    Kill Bill: Vol. 1 & Uma Thurman\\
     Inglourious Basterds & Brad Pitt \\
    Once Upon a Time in Hollywood & Brad Pitt \\
    Inglourious Basterds & Zoë Bell \\
    Once Upon a Time in Hollywood & Zoë Bell \\
    Once Upon a Time in Hollywood & Leonardo DiCaprio \\
    \hline
    \end{tabular}
    \end{minipage}
    \caption{Running example: a fragment of an IMDB-like database and two queries.}
    \label{fig:running_example}
\end{figure*}

\nop{
\ahmet{@Omer:\\
- In Figure 1, TopMovieCast(y,x) should be TopMovieCast(x,y).\\ 
- Maybe  we should also use uppercase variables $R$, $Y$, $X$, $Z$ in the figure, since in the definition of queries we use uppercase variables.\\
- In Example 1.1, TopCast should be changed to TopMovieCast}}

In this paper, we study the problem of {\em computing the contribution of database tuples to the results of aggregate queries, quantified by the Banzhaf and Shapley values}. The computational problem is highly intractable in general: already for Boolean non-hierarchical queries, exact computation of Shapley and Banzhaf values is \#P-hard~\cite{TheShapleyValueofTuplesinQueryAnswering}.
{\em All prior practical approaches were designed for Select-Project-Join-Union (SPJU) queries and do not work for queries with aggregates. Their runtime is far from negligible, especially for computing the Banzhaf/Shapley values for all input-output tuple pairs.}

In this paper, we go beyond prior work by adding support for the COUNT, SUM, MIN, MAX aggregates in the select clause of the query. Aggregation is key to practical OLAP and decision support workloads. All 22 TPC-H queries involve aggregation.

The computation of Shapley and Banzhaf values resorts to the compilation of the query lineage, which is a polynomial over variables representing the input tuples~\cite{VLDB2014:Causality_and_explanations_in_databases, TrendsInExplanations,ProvenanceInDatabases:WhyHowandWhere} into tractable circuits, such as d-trees~\cite{AggregationInProbabalisticDatabasesViaKnowlegeCompilation}. Yet aggregate queries require lineage formalisms beyond the Boolean semiring for queries without aggregates. {\em The first main contribution of this paper is an extension of this compilation approach to the lineage of queries with aggregates.}


Computing Banzhaf/Shapley values for aggregate queries remains expensive for most benchmarks used in the literature~\cite{ComputingTheShapleyValueOfFactsInQueryAnswering,Sig24:BanzhafValuesForFactsInQueryAnswering}. {\em The second main contribution of this paper is  the interplay of two novel techniques that we incorporate in our compilation algorithm: lifted compilation and gradient-based Banzhaf/Shapley value computation.} These techniques significantly reduce the computation time for fact attribution for {\em queries with and even without aggregates}.

In more detail, our contributions are as follows.

We introduce \lexaban and \lexashap, two novel algorithms to compute Banzhaf and respectively Shapley values for SPJUA (Select-Project-Join-Union-Aggregate) queries. We conduct experiments with these algorithms and the state-of-the-art algorithms \exaban and \exashap for SPJU queries using over a million instances~\cite{Sig24:BanzhafValuesForFactsInQueryAnswering}. For queries without aggregates, \lexaban achieves 2-3 orders of magnitude (OOM) improvement over \exaban, while \lexashap achieves more than 2 OOM improvement over \exashap.
For queries with aggregates, \lexaban and \lexashap are, to our knowledge, the first practical algorithms for Banzhaf and Shapley value computation. 
The performance gains of \lexaban and \lexashap over the state-of-the-art are primarily due to two novel techniques.

{\bf Lifted compilation.} Lifting exploits the observation that some variables (input tuples) in Boolean functions (query lineage) are interchangeable in the sense that they behave the same; in particular, the Banzhaf/Shapley values of interchangeable variables are the same. For instance, the pairs of variables $\{x_1,x_2\}$ and $\{y_1,z_1\}$ are interchangeable in $(x_1\wedge y_1\wedge z_1)\vee(x_2\wedge y_1\wedge z_1)$. We accommodate lifting into lineage compilation by extending the d-trees with a new gate type, which replaces conjunctions and disjunctions of interchangeable variables with a fresh variable. Lifting preserves equi-satisfiability. In our experiments, lifting leads to more than 2 OOM speedup and more than 1 OOM smaller d-trees.

{\bf Gradient computation.} Computing Banzhaf/Shapley values separately for each input fact, as done in all prior work, is expensive. We show how to compute them for all the facts in the same complexity as for one fact. For this, we expose a fundamental connection between the vector of Banzhaf/Shapley values, one per input fact, and the gradient of a function that computes the probability of the query lineage with respect to each of its variables. In particular, we show that the Banzhaf value vector is precisely this gradient where the probability of each variable is 1/2. The vector of Shapley values can be recovered similarly. We devise a back-propagation algorithm to compute these values efficiently over any d-tree. In our experiments, this gradient technique speeds up the computation by a factor that grows with the size of the instance and up to 2 OOM speedup for large instances. This enables the computation for instances larger than previously possible.

\nop{We exhibit a novel relationship between the Banzhaf/Shapley values and the gradient of a function defined by the query lineage. Interestingly, in the case of Banzhaf values, this function aligns with a fundamental property of the lineage—its model count. This allows us to devise a back-propagation algorithm to compute these values efficiently over the d-tree. Our gradient technique allows to compute the Banzhaf/Shapley values for all input tuples with the same complexity as for one tuple. In our experiments, this gradient technique speeds up the computation by a factor that grows with the size of the instance and up to more than 2 OOM speedup for large instances. This enables the computation of Banzhaf values for instances larger than previously possible.}

We first introduce these two techniques for SPJU queries in Section~\ref{sec: Algorithm} and then extend them to SPJUA queries in Section~\ref{sec: aggregates}. In Section ~\ref{sec:experiments} we show that the interplay of both techniques can lead to 3 OOM speedup.





\nop{To the best of our knowledge no prior work used semimodule d-trees for Banzhaf/Shapley value computation.}

\nop{
\paragraph*{Related Work} Banzhaf and Shapley values were originally introduced in the context of game theory to quantify the contribution of individual players to the value of a cooperative game \cite{shapley1953value,Banzhaf:1965}.
Recent work laid the theoretical foundation of Banzhaf and Shapley values for query answering~\cite{TheShapleyValueofTuplesinQueryAnswering,ComputingTheShapleyValueOfFactsInQueryAnswering, TheimpactofnegationonthecomplexityoftheShapleyvalueinconjunctivequeries,Sig24:ExpectedShapleyLikeScores}. The state-of-the-art approach for computing such values for SPJU queries compiles the query lineage into d-tree circuits~\cite{Sig24:BanzhafValuesForFactsInQueryAnswering}. This builds on earlier work that compiles lineage into d-DNNF circuits~\cite{ComputingTheShapleyValueOfFactsInQueryAnswering}. Compared to the state-of-the-art, our work
achieves {\em significant performance gains}: up to 3 OOM speedup and pushing the boundaries of tractability for the problem; it also supports queries with aggregates, rather than SPJU in \cite{Sig24:BanzhafValuesForFactsInQueryAnswering, ComputingTheShapleyValueOfFactsInQueryAnswering}. It does so through the novel techniques of lifted compilation and gradient-based computation, mentioned above.
Further notions have been used to quantify fact contribution in query answering: causality \cite{VLDB2014:Causality_and_explanations_in_databases}, responsibility \cite{Thecomplexityofcausalityandresponsibilityforqueryanswersandnon-answers}, and counterfactuals \cite{BringingProvenancetoItsFullPotentialUsingCausalReasoning}. The SHAP score \cite{lundberg2017unified} leverages and adapts the definition of Shapley values to explain model predictions.  
}

Implementation of the algorithms describes in this paper is available online ~\cite{full-version-git}.

%% file: Sections_Arxiv_New/Preliminaries.tex
\section{Preliminaries}
\label{sec:Preliminaries}

Table~\ref{tab:notation_overview} summarizes the notation used in the paper. We use $\mathbb{N}$ to denote the set of natural numbers including $0$. 
\nop{
For $n \in \mathbb{N}$, we define
$[n] = \{1, \ldots, n\}$. In case $n = 0$, we have 
$[n]=\emptyset$.
}

\paragraph{Boolean Formulas}  
A \emph{(Boolean) formula} $\varphi$ over a set \(\bm{X}\) of Boolean variables is either a constant \(0\) or \(1\), a variable \(x \in \bm{X}\), 
a negation \(\neg \varphi_1\), a conjunction
\(\varphi_1 \wedge \varphi_2\), or a disjunction
 \(\varphi_1 \vee \varphi_2\) of formulas $\varphi_1$ and 
 $\varphi_2$. \(\bool(\bm{X})\) ( \(\posbool(\bm{X})\) ) is the set of all (positive) Boolean formulas over \(\bm{X}\).  
A \emph{literal} is a constant, a variable, or its negation. The set of variables in \(\varphi\) is denoted \(\vars(\varphi)\). A formula  is \emph{read-once} if each variable appears at most once. \(\varphi[x := b]\) denotes the substitution of \(x \in \bm{X}\) with \(b \in \{0,1\}\) in \(\varphi\).  
A \emph{valuation} \(\theta : \vars(\varphi) \to \{0,1\}\) 
maps each variable in $\varphi$ to \(0\) or \(1\); we identify \(\theta\) with \(\{x \in \vars(\varphi) \mid \theta(x) = 1\}\), where \(|\theta|\) is the number of variables mapped to \(1\). The Boolean value of \(\varphi\) under \(\theta\) is \(\varphi[\theta]\), and \(\theta\) is a \emph{model} of \(\varphi\) if \(\varphi[\theta] = 1\). The \emph{model set} \(\mods(\varphi)\) consists of all models of \(\varphi\), with \emph{model count} \(\#\varphi = |\mods(\varphi)|\) and \emph{\(k\)-model count} \(\#_k \varphi = |\mods(\varphi) \cap \binom{\bm{X}}{k}|\), where \(\binom{\bm{X}}{k}\) is the set of \(k\)-element subsets of \(\bm{X}\).  
Two formulas \(\varphi_1\) and \(\varphi_2\) are \emph{independent} if \(\vars(\varphi_1) \cap \vars(\varphi_2) = \emptyset\) and \emph{exclusive} if \(\mods(\varphi_1) \cap \mods(\varphi_2) = \emptyset\). 
\nop{The set \(\posbool(\bm{X})\) of \emph{positive formulas} (formulas without negation) forms a commutative semiring under \(\vee, \wedge, 0, 1\)~\cite{Golan1999SemiringsAT}.}  
\nop{I have commented out that \(\posbool(\bm{X})\) is a semiring.}
\begin{example}
Consider the formula $\varphi = x\vee y$.
The models of $\varphi$ are $\{x\}$, $\{y\}$, and $\{x, y\}$. 
Hence, $\#\varphi = 3$, $\#_0\varphi = 0$, $\#_1\varphi = 2$,
and $\#_2\varphi = 1$.
\end{example}
\begin{table}[t]
    \caption{Overview of notation introduced in Section~\ref{sec:Preliminaries}. $\varphi$
is a Boolean formula; $\Phi$ is a \boolnumpairs expression.}
\begin{tabular}{@{}l| l@{}}
\hline
$[n]$ & $\{1, \ldots , n\}$; $[n]=\emptyset$ if $n = 0$\\
$\bool(\bm{X})$ \hspace{-0.1em}($\posbool(\bm{X})$) \hspace{-0.5em} & Set of all (positive) Boolean formulas over $\bm X$ \\
$\varphi[x:=b]$ & Substitution of variable $x$ in $\varphi$ by $b \in \{0,1\}$ \\
$\varphi[\theta]$ & Boolean value of $\varphi$ under  valuation  $\theta$ \\
$\#\varphi$ ($\#_k\varphi$) & Number of models of $\varphi$ (of size $k$)\\
$\sum^M_{i} \varphi_i \otimes m_i$ & \boolnumpairs (\tensorset) using monoid $M$
\\
$\#^p\Phi$ ($\#_k^p\Phi$) & Number of valuations $\theta$ (of size $k$) s.t. $\Phi[\theta]= p$ \\
$D_x, D_n \subseteq D$ & Exogenous/endogenous facts in database $D$ \\
$\lin(Q,D,\bm t)$ & Lineage of $\bm t$ with respect to $Q$ over $D$ \\
$\lin(\langle \alpha,\gamma, Q\rangle,D)$ & Lineage of aggregate query $\langle \alpha,\gamma, Q\rangle$ over $D$\\
\hline
\end{tabular}

\label{tab:notation_overview}
\end{table}


\paragraph{Databases}  
We assume a countably infinite set \(\Const\) of constants used as database values. 
A {\em $k$-tuple} \(\bm{t}\) for \(k \in \mathbb{N}\) is an element from \(\Const^k\). 
A database consists of \emph{facts} \(R(\bm{t})\), where \(R\) is a relation name with some arity $k \in \mathbb{N}$ and \(\bm{t}\) is a $k$-tuple. 
As in ~\cite{TheShapleyValueofTuplesinQueryAnswering}, each database \(D\) is partitioned into \emph{endogenous} facts \(D_n\) and \emph{exogenous} facts \(D_x\).  
We focus on the contribution of endogenous facts to the query result. Intuitively, we separate the facts in the database that are under analysis/control from those that are fixed/externally imposed.

\begin{example}
In Figure~\ref{fig:running_example}, all facts labeled ``endo" (``exo") are endogenous (respectively exogenous), meaning that we (do not) wish to analyze their contributions, e.g., Cast is labeled ``endo" and we want to find out which actors contributed the most to the query result. 
\end{example}

\nop{

\begin{example}

As indicated by the labels next to the relation names in the IMDB-like database shown in Figure~\ref{fig:running_example}, all facts in the relations MovieAwards and TopMovieCast are considered exogenous, whereas all other facts are considered endogenous. Intuitively, this choice is made 
because one of the questions we are interested in is to find out which actors contributed the most to a director receiving an award for one of their movies. We therefore focus on the actor tuples in Cast, declaring them endogenous.

\end{example}
}

\paragraph{Conjunctive Queries}  
A \emph{conjunctive query (CQ)} \cite{abiteboul1995foundations} has the form $Q(\bm{X}) = R_1(\bm{X}_1), \dots, R_m(\bm{X}_m)$, where each  $R_i(\bm{X}_i)$ is an \emph{atom}, $R_i$ is a relation symbol, and $\bm{X}_i$ is a tuple of variables and constants. 
The set \(\bm{X}\subseteq \bm{X}_1\cup \dots\cup \bm{X}_m\) consists of the free (head) variables.
A CQ is \emph{$k$-ary} if \(|\bm{X}| = k\) and \emph{Boolean} if $k = 0$.
A \emph{union of conjunctive queries (UCQ)} consists of a set of $k$-ary CQs for some $k \in \mathbb{N}$. 

\nop{I have moved the definition of query grounding here, since it also entails the definition of $Q(D)$.}
\paragraph{Query Grounding} A {\em grounding} of a CQ $Q$ w.r.t. a database $D$ is a mapping $G$ of the variables of $Q$ to constants such that replacing every variable $X$ by $G(X)$ turns every atom in $Q$ to a fact in $D$. We denote this set of facts by $\facts(Q,D,G)$. A grounding for a UCQ is a grounding for one of its CQs.
Each grounding $G$ yields an {\em output tuple}, as the restriction of $G$ to the head variables of $Q$. 
Multiple groundings may yield the same output tuple. 
We denote the set of groundings yielding $\bm t$ by $G(Q,D,\bm t)$ and  the set of all output tuples, which defines the result  of $Q$ over $D$, by $Q(D)$. 
For an output tuple $\bm t\in Q(D)$, a head variable $X \in \bm X$,
and a set $\bm Y \subseteq \bm X$ of head variables, we denote by $\bm t[X]$
the X-value in $\bm t$ and by $\bm t[\bm Y]$ the projection of $\bm t$ onto $\bm Y$. A Boolean query is satisfied if it has a grounding.

\begin{example}
 For $Q_1$ in Figure~\ref{fig:running_example} and 
  $X =\ \text{'Inglourious Basterds'}$, $R = 322$, $Y = \ \text{'Brad Pitt'}$, and $Z = \ \text{'Academy'}$,  we obtain a grounding and thus the query is satisfied. 
\end{example}

\nop{
\paragraph{Aggregate Queries Semantics} 
The semantics that we give here follows that of \cite{TheShapleyValueofTuplesinQueryAnswering}, combining set semantics for the underlying relational query with bag semantics for the aggregate function. Given a database $D$, 
let $Q(D) = \{\bm t_1, \ldots, \bm t_k\}$ be the result of evaluating $Q$ over $D$ under set semantics.
The result of 
$\langle \alpha, \gamma, Q\rangle$ over $D$ is defined as
$\langle\alpha, \gamma , Q\rangle (D):= \alpha((\gamma(\bm t_i))_{i\in k})$.
If $\gamma$ returns the $i$-th value of the tuple $\bm t$, 
we just write $i$ instead of $\alpha$. Note that although $Q(D)$ is a set, the object $\gamma(\bm t_i)_{i\in k}$ on which $\alpha$ is applied is generally a bag, since the same value may occur multiple times in the aggregated column of different tuples. Some of the aggregation functions that we will consider are sensitive to such multiplicities.
We focus on the following aggregation functions: MAX,MIN, SUM and COUNT, defined as follows. We consider SUM and MAX aggregate queries 
of the following form:
{   \begin{align*}
    \langle \sumagg, \gamma, Q\rangle(D) \overset{def}{=}& \sum_{\vec{c}\in Q(D)}{\gamma{(\vec{c})}}\\
    \langle \maxagg, \gamma, Q\rangle(D) \overset{def}{=}& \begin{cases} 
\max{\{\gamma(\vec{c})\, |\, \vec{c}\in Q(D)\}} & \text{if } Q(D) \neq \emptyset \\ 
0 & \text{if } Q(D) = \emptyset
\end{cases}
\end{align*}    }

The query $\langle \minagg, \gamma, Q\rangle$ is defined analogously to $\langle \maxagg, \gamma, Q\rangle$. Last, we define 
$\langle \countagg, Q\rangle \overset{def}{=}
\langle \sumagg, \bm 1 , Q\rangle$, where 
$\bm 1$ maps each output tuple to $1$, yielding the number of output tuples of $Q$. To support $\groupbyagg$, we extend the notation to $\langle \alpha, \gamma, Q, \vec{g} \rangle$, where $\vec{g}$ is a projection defining the grouping attributes. Let $Q(D) = \{\bm t_1, \ldots, \bm t_k\}$. For each group value $\vec{v}$ in $\{\vec{g}(\bm t_i) \mid i \in [k]\}$, let $G_{\vec{v}} = \{\bm t_i \in Q(D) \mid \vec{g}(\bm t_i) = \vec{v}\}$. We define:
\begin{align*}
    \langle \alpha, \gamma, Q, \vec{g} \rangle(D) :=& \{ (\vec{v}, \alpha(\gamma(\bm t) \mid \bm t \in G_{\vec{v}})) \}
\end{align*}

That is, the aggregation is applied separately to each group, with multiplicities preserved as in the ungrouped case.
}
\paragraph{Aggregate Queries} 
An aggregate query is a triple $\langle\alpha, \gamma, Q\rangle$, where $Q$
is a $k$-ary UCQ for some $k \in \mathbb{N}$, $\gamma: \Const^k \rightarrow \mathbb{R}$ \nop{is a function that} maps 
each k-tuple over $\Const$ to a numeric value and $\alpha: \mathbb{R}^* \to \mathbb{R}$ aggregates tuples of numbers to single numbers~\cite{TheShapleyValueofTuplesinQueryAnswering}. 
When evaluating an aggregate query $Q$ over a database $D$, $\gamma$ is applied to each tuple in $Q(D)$ to yield a number, and $\alpha$ is applied on the set of these numbers. 
Let $Q(D) = \{\bm t_1, \ldots, \bm t_n\}$. 
The result of 
$\langle \alpha, \gamma, Q\rangle$ is defined as
$\langle\alpha, \gamma , Q\rangle (D):= \alpha\big((\gamma(\bm t_i))_{i\in [n]}\big)$.
If $\gamma(\bm t)$ returns the $X$-value of $\bm t$ for a head variable $X$, we just write $X$ instead of $\gamma(\bm t)$. 
We focus on the \nop{aggregation functions}aggregations $\maxagg$, $\minagg$, $\sumagg$, and $\countagg$ defined as follows:
   \begin{align*}
    \langle \sumagg, \gamma, Q\rangle(D) \overset{def}{=}& \sum_{\bm t\in Q(D)}{\gamma{(\bm t)}}\\
    \langle \maxagg, \gamma, Q\rangle(D) \overset{def}{=}& \begin{cases} 
\max{\{\gamma(\bm t)\, |\, \bm t\in Q(D)\}} & \text{if } Q(D) \neq \emptyset \\ 
0 & \text{if } Q(D) = \emptyset
\end{cases}
\end{align*}    
The query $\langle \minagg, \gamma, Q\rangle$ is defined analogously to $\langle \maxagg, \gamma, Q\rangle$. We define 
$\langle \countagg, Q\bm \rangle \overset{def}{=}
\langle \sumagg, \bm 1 , Q\rangle$, where 
$\bm 1$ maps each output tuple in $Q(D)$ to $1$, yielding the number of output tuples of $Q$. 
If $\alpha$ and $\gamma$ are clear from context, we refer to an aggregate query by $Q$.

\nop{

To support $\groupbyagg$ aggregate queries, we apply the functions $\alpha$
and $\gamma$ to each group of tuples in the query result
that agree on the values of the $\groupbyagg$ variables.
Given a $k$-ary UCQ $Q$ with head variables $\bm X$,
we express a $\groupbyagg$ aggregate query as 
$\langle \alpha, \gamma, Q, \bm Y\rangle$, where $\bm Y \subseteq \bm X$
are the $\groupbyagg$ variables of the query. Let 
$Q(D) = \{\bm t_1, \ldots , \bm t_n\}$ be the result tuples of $Q$, 
$P = \{\bm t_i[\bm Y] \mid i \in [n]\}$ the projections of the result tuples
onto $\bm Y$, and $I_{\bm t} = \{i \in {n} \mid \bm t_i[\bm Y] = \bm t\}$ the indices of the
group of result tuples that agree with $\bm t \in P$ on $\bm Y$. We define the result of $\langle \alpha, \gamma, Q, \bm Y \rangle$ over $D$ as:
\begin{align*}
    \langle \alpha, \gamma, Q, \bm Y \rangle(D)\overset{def}{=}& 
    \big\{ \big(\bm t, \alpha\big((\gamma(\bm t_i))_{i \in I_{\bm t}}\big)\big) \mid \bm t \in P \big\}
\end{align*}
}

\begin{remark}\label{rem:max-min}
The choice for the $\maxagg/\minagg$ queries to evaluate to 0 if $Q(D) = \emptyset$, as in \cite{TheShapleyValueofTuplesinQueryAnswering},  is not arbitrary. In the definition of Banzhaf and Shapley values below, we sum over marginal contributions w.r.t. the query result of individual tuples for different sub-databases, potentially including ones for which the query result is empty. The alternative of defining the result of aggregation over the empty set to be $\pm\infty$ would lead to unintuitive Banzhaf/Shapley values of $\pm\infty$. 
\end{remark}

\begin{example}
    Query $Q_2$ of Figure~\ref{fig:running_example} is an aggregate query that asks for the maximal revenue of a movie directed by Tarantino.
    By setting $X =\ \text{'Once Upon a Time in Hollywood'}$, $R = 377$, and  $Y = \ \text{'Brad Pitt'}$, we obtain the maximal revenue 377.
\end{example}


\paragraph{Monoids}
We denote by $\overline{\mathbb{R}}$ the real numbers $\mathbb{R}$ including 
$\{\infty, -\infty\}$.
 \begin{definition} A (numeric) monoid $M=(\overline{\mathbb{R}},+_M, 0_M)$ consists of a binary operation $+_M: \overline{\mathbb{R}} \times \overline{\mathbb{R}} \rightarrow \overline{\mathbb{R}}$, and a neutral element $0_M \in \overline{\mathbb{R}}$ s.t. 
for all $m_1, m_2, m_3 \in \overline{\mathbb{R}}$, the following holds:
{   \begin{align*}
(m_1 +_M m_2) +_M m_3 & =  m_1 +_M (m_2 +_M m_3) \\
0_M +_M m_1 &  = m_1 +_M 0_M = m_1
\end{align*}    }
A monoid is \emph{commutative} if 
$m_1 + m_2 = m_2 + m_1$ for all  $m_1,m_2 \in \overline{\mathbb{R}}$. It is idempotent if $m+m=m$ for all $m \in \overline{\mathbb{R}}$.
\end{definition}
%
Monoids naturally model aggregate functions, as illustrated next:
\begin{example}
To model the aggregate functions $\maxagg$ and $\minagg$, we use the monoids $(\overline{\mathbb{R}}, \max, -\infty)$ and respectively
$(\overline{\mathbb{R}}, \min, \infty)$. We use $\pm\infty$ as neutral values as usual for $\minagg$ and $\maxagg$ despite of Remark \ref{rem:max-min}, because $0$ is not a neutral value for $\minagg$ ($\maxagg$) in presence of positive (resp. negative) numbers; adjustments in light of the remark will be made in our notion of valuations below. 
For $\sumagg$ and $\countagg$, we use $(\overline{\mathbb{R}}, +, 0)$. The monoids for $\minagg$ and $\maxagg$ are idempotent and the monoid for $\sumagg$ and $\countagg$ is not.
\end{example}
\nop{For $\countagg$, I changed $(\{0,1\}, \sum, 0)$ to $(\mathbb{N}, +, 0)$.}

\paragraph{\boolnumpairs} 
A \boolnumpair (\tensorset) 
represents a collection of numeric values conditioned on Boolean formulas. 
Consider a variable set $\bm X$ and a bag 
$\calB = \{\!\!\{ (\varphi_i, m_i) \mid i \in [n]\}\!\!\}$, where each $\varphi_i$ is a positive Boolean formula from $\posbool(\bm X)$ and each $m_i$ is an element from 
$\overline{\mathbb{R}}$ for a commutative numeric monoid 
$M = (\overline{\mathbb{R}},+_M, 0_M)$. 
We call $\Phi = \sum^M_{i} \varphi_i \otimes m_i$ a  \tensorset over $\bm X$ and $M$.
The {\em Boolean part} of $\Phi$ is defined as
$\varphi = \bigvee_i \varphi_i$ and its variables are $\vars(\Phi) = \vars(\varphi) = \bigcup_i \vars(\varphi_i)$.
Given a valuation $\theta$ of $\varphi$, we define $VB$ as the bag of values $m_i$ whose formulas $\varphi_i$ evaluate to $1$ under $\theta$:  $VB = \{\!\!\{  m_i \mid (\varphi_i, m_i) \in \calB \text{ and } \varphi_i[\theta] = 1 \}\!\!\}$. We define the value of $\Phi$ under $\theta$ as $\Phi[\theta] = \sum\nolimits^M_{m \in VB} m$ if $VB \neq \emptyset$ and $\Phi[\theta] = 0_{\mathbb{R}}$ otherwise, 
where $0_{\mathbb{R}}$ is the real number 0 and the summation $\sum^M$ uses the $+_M$ operator of the monoid. The explicit account for the case of valuations that yield an empty set is due to Remark \ref{rem:max-min}. 
A valuation $\theta$ is a {\em model} of $\Phi$
if it is a model of $\varphi$. 
The {\em model count} $\#\Phi$ of $\Phi$ is defined as the model count $\#\varphi$ of $\varphi$. 
For a value $p \in \overline{\mathbb{R}}$, we denote by $\#^p\Phi$ the number of valuations $\theta$ such that $\Phi[\theta]= p$ and by $\#_k^p\Phi$ the number of such valuations of size $k$.
We say that two \tensorset expressions 
$\Phi_1$ and
$\Phi_2$ are {\em independent} (resp. {\em exclusive}) if their Boolean parts are independent (resp. exclusive).
%

\begin{example}
Consider the  \tensorset $\Phi = (x \vee y) \otimes 3 +  (x \vee y) \otimes 3 + x \otimes 2 +  y \otimes 4$ using the monoid $(\overline{\mathbb{R}}, +, 0)$. Its Boolean part is equivalent to $\varphi = x\vee y$ and its set of variables is $\vars(\Phi) = \{x,y\}$.
The valuation $\theta =\{x \mapsto 1, y \mapsto 0\}$ only satisfies the first three 
 formulas in $\Phi$. Hence, $\Phi[\theta] = 3 + 3 + 2 = 8$. 
We have  
$\#\Phi = \#\varphi = 3$ and $\#^8 \Phi = 1$.
It holds that $\#_2^8\Phi = 0$, since 
the only valuation of size 2 is $\theta' = \{x \mapsto 1, y \mapsto 1\}$ and $\Phi[\theta'] = 12$.
\end{example}

\paragraph{Semimodules}
If the monoid $M = (\overline{\mathbb{R}},+_M, 0_M)$ used by a \tensorset is  idempotent, then the \tensorset can be interpreted as a semimodule expression\footnote{Previous work \cite{ProvenanceForAggregateQueries}, which proposed using semimodules to capture lineage, also noted that PosBool[X] is only ``compatible" in the context of this construction with idempotent monoids. This is why for the non-idempotent case we use \tensorset instead.}, where the operator $\otimes: \posbool(\bm{X}) \times \overline{\mathbb{R}} \to \overline{\mathbb{R}}$ obeys the following axioms:
{    \begin{enumerate}
    \item $\varphi_1 \otimes (m_1 +_M m_2) = \varphi_1 \otimes m_1 +_M \varphi_1 \otimes m_2$
    \item $(\varphi_1 \vee \varphi_2) \otimes m_1 = \varphi_1 \otimes m_1 +_M \varphi_2 \otimes m_1$
    \item $(\varphi_1 \wedge \varphi_2) \otimes m_1 = 
            \varphi_1 \otimes (\varphi_2 \otimes m_1)$
    \item $\varphi_1 \otimes 0_M = 0 \otimes m_1 = 0_M$
    \item $1 \otimes m_1 = m_1$
\end{enumerate}     }

Our notion of valuation for BNPs also extends to semimodules.


\nop{

\paragraph{Monoids and Semimodules}
The development in this paper uses the algebraic structures of numeric monoids and semimodules constructed using such monoids and the $\posbool(\bm X)$ semiring. We denote by $\overline{\mathbb{R}}$ the set of real numbers $\mathbb{R}$ including 
$\{\infty, -\infty\}$.
 \begin{definition} A numeric monoid (or simply, a monoid) $(\overline{\mathbb{R}},+_M, 0_M)$ consists of a binary operation $+_M: \overline{\mathbb{R}} \times \overline{\mathbb{R}} \rightarrow \overline{\mathbb{R}}$, and a neutral element $0_M \in \overline{\mathbb{R}}$ s.t. 
for all $m_1, m_2, m_3 \in \overline{\mathbb{R}}$, the following holds:
{   \begin{align*}
(m_1 +_M m_2) +_M m_3 & =  m_1 +_M (m_2 +_M m_3) \\
0 +_M m_1 &  = m_1 +_M 0 = m_1
\end{align*}    }
A monoid is \emph{commutative} if 
$m_1 + m_2 = m_2 + m_1$ for all  $m_1,m_2 \in M$.
\end{definition}

\paragraph{Monoids for Aggregate Queries} Monoids naturally model many aggregate functions, and specifically $\minagg,\maxagg,\sumagg, \countagg$. We denote by $\overline{\mathbb{R}}$ the set of real numbers $\mathbb{R}$ including
$\{\infty, -\infty\}$. For $\maxagg$ we then use the monoid
$(\overline{\mathbb{R}}, \max, -\infty)$ and for $\minagg$ we use the monoid  
$(\overline{\mathbb{R}}, \min, \infty)$. For $\sumagg$ we use $(\mathbb{R}, \sum, 0)$ and for $\countagg$ we use $(\{0,1\}, \sum, 0)$ (as we sum only values of $1$).

Next, we introduce $\posbool$-semimodules, which combine monoid values with positive Boolean formulas to represent numerical values conditioned by Boolean formulas. 

\begin{definition}
\label{def: semimodule}
Given a monoid $M=(\overline{\mathbb{R}},+_M,0_M)$ and the semiring $(\posbool(\bm X), \vee,\wedge,0,1)$, the triple $(\posbool(\bm X),M,\otimes)$ is a $\posbool$-semimodule where $\otimes : \posbool(\bm X) \times \overline{\mathbb{R}} \to \overline{\mathbb{R}}$, such that for all $\varphi_1, \varphi_2 \in \posbool(\bm X)$  and $m_1, m_2 \in \overline{\mathbb{R}}$, the following axioms hold:
{    \begin{enumerate}
    \item $\varphi_1 \otimes (m_1 +_M m_2) = \varphi_1 \otimes m_1 +_M \varphi_1 \otimes m_2$
    \item $(\varphi_1 \vee \varphi_2) \otimes m_1 = \varphi_1 \otimes m_1 +_M \varphi_2 \otimes m_1$
    \item $(\varphi_1 \wedge \varphi_2) \otimes m_1 = 
            \varphi_1 \otimes (\varphi_2 \otimes m_1)$
    \item $\varphi_1 \otimes 0_M = 0 \otimes m_1 = 0_M$
    \item $1 \otimes m_1 = m_1$
\end{enumerate}     }
\end{definition}

 A {\em semimodule expression} is either an element from $M$ or of one of the forms $\Phi_1 +_M \Phi_2$ 
and  $\varphi \otimes  \Phi_1$ for  
semimodule expressions  $\Phi_1$ and $\Phi_2$ and a 
Boolean formula $\varphi \in \posbool(\bm X)$.
It easily follows from the definition of 
semimodules that every semimodule expression can be 
equivalently written as
$\varphi_1 \otimes m_1 +_M \cdots +_M \varphi_k \otimes m_k$,
where each $\varphi_i$ is a Boolean formula and each 
$m_i$ is an element from $M$.
We abbreviate the latter expression by $\sum_M \varphi_i \otimes m_i$.

Consider a semimodule expression 
$\Phi = \sum_M \varphi_i \otimes m_i$. 
The {\em Boolean part} of $\Phi$ is  
$\bigvee \varphi_i$.
Given a valuation $\theta$ over $\bm X$, we denote $S_\theta = \{i | \varphi_i[\theta]=1\}$. By  
$\Phi[\theta]$, we denote
the monoid value computed by  
$\sum_M \big[\varphi_j[\theta] \otimes m_j\big]_{j\in S_\theta}$. 
Note that this is equal to $\sum_M m_j$ where the sum ranges over all $j$ for which $\varphi_j$ evaluates to $1$ over $\theta$.
\begin{remark}
Our semantics for semimodule expressions considers the summation only over those terms $\varphi_i \otimes m_i$ in $\Phi$ for which $\varphi_i$ is evaluated to 1. This is to avoid $\Phi$ taking the default value of the monoid for the valuation that maps all $\varphi_i$ to 0  (see Remark~\ref{rem:max-min}).
\nop{
 The choice to ignore the terms $\varphi_i \otimes m_i$ in the semimodule expression $\Phi$ for which $\varphi_i[\theta]=0$ is so that we do not need to sum over the default values in the monoid. For the MAX/MIN monoids, the default values are not useful for computing the fact attribution, as also pointed out in Remark~\ref{rem:max-min}.
 }
 \nop{
 The choice to evaluate the semimodule expression over formulas $\varphi_i$ s.t. $\varphi_i[\theta] = 1$ instead of evaluating it over all formulas is not standard. This is to accommodate for the definition of the $\maxagg$ and $\minagg$ queries that evaluate to $0$ in case the $Q(D)= \emptyset$. The standard evaluation of a semimodule would be the $0_M$ values $-+\infty$, in contrast with the query definition.
 }
\end{remark}
For a semimodule expression $\Phi$, we denote by $\#^p\Phi$ the number of valuations $\theta$ such that $\Phi[\theta]= p$ and by $\#_k^p\Phi$ the number of such valuations of size $k$.
We define the variables $\vars(\Phi)$ of $\Phi$ to be the variables $\vars(\varphi)$ of $\varphi$ and the model count $\#\Phi$ of $\Phi$
to be the  model count $\#\varphi$ of $\varphi$, where $\varphi$
is the Boolean part of $\Phi$. We say that two semimodule expressions 
$T_1$ and
$T_2$ are independent (resp. exclusive) if their Boolean parts are independent (resp. exclusive).

}

\paragraph{Query Lineage}
Consider a database 
$D = D_n \cup D_x$. To construct the lineage of a
query over $D$, we first associate each endogenous fact (or tuple) $f$
in $D_n$ with a Boolean variable $v(f)$.
 Given a UCQ $Q$ and an output tuple $\bm t$, the lineage of 
 $\bm t$ with respect to $Q$ over $D$, denoted by $\lin(Q,D,\bm t)$, is a positive Boolean formula in Disjunctive Normal Form over the variables $v(f)$ of facts $f$ in $D_n$:
 {   \begin{align*}
 \lin(Q,D,\bm t) = \bigvee_{G \in G(Q,D,\bm t)}\ \ \bigwedge_{f \in facts(Q,D,G) \cap D_n}v(f)
 \end{align*}    }
That is, the lineage is a disjunction over all groundings 
yielding $\bm t$. For each such grounding $G$, 
the lineage has a conjunctive clause consisting of all endogenous
facts from $\facts(Q,D,G)$. If $Q$ is a Boolean query, we abbreviate $\lin(Q,D,())$ by $\lin(Q,D)$.

\begin{table*}[h]
    \caption{Lineages and lifted lineages for queries $Q_1$ and $Q_2$ from Figure~\ref{fig:running_example}. 
    Colored variables are mapped to formulas.}
\scriptsize
    \centering
    \begin{tabular}{l c| c}
        \toprule
        \textbf{Query} & \textbf{Lineage} & \textbf{Lifted Lineage} \\
        \midrule
        $Q_1$ & $
        \begin{aligned}
            &(d_1 \wedge a_1 \wedge m_3) \vee (d_1 \wedge a_1 \wedge m_2) \vee (d_1 \wedge a_3 \wedge m_3)
            \vee (d_1 \wedge a_2 \wedge m_3) \vee (d_1 \wedge a_3 \wedge m_2)\\ &\vee (d_2 \wedge a_1 \wedge m_3)
            \vee (d_2 \wedge a_1 \wedge m_2) \vee (d_2 \wedge a_3 \wedge m_3) \vee (d_2 \wedge a_2 \wedge m_3)
            \vee (d_2 \wedge a_3 \wedge m_2)
        \end{aligned}$ &         $\begin{aligned}
&(\textcolor{red}{d_{1,2}} \wedge \textcolor{goodgreen}{a_{1,3}} \wedge m_3) \vee (\textcolor{red}{d_{1,2}} \wedge \textcolor{goodgreen}{a_{1,3}} \wedge m_2) \vee(\textcolor{red}{d_{1,2}} \wedge a_2 \wedge m_3)\\
& \ell = \{{\color{red}d_{1,2}} \mapsto (d_1 \vee d_2),{\color{goodgreen}a_{1,3}} \mapsto (a_1 \vee a_3), a_2 \mapsto a_2, m_2\mapsto m_2, m_3\mapsto m_3\}
        \end{aligned}$ \\
        \midrule
        $Q_2$ & $
        \begin{aligned}
            &(a_1\wedge m_3) \otimes 377 +_{\max}
            (a_2\wedge m_3) \otimes 377 +_{\max}
            (a_3\wedge m_3) \otimes 377 \\ &+_{\max}
            (a_1 \wedge m_2) \otimes 322 +_{\max}
            (a_3 \wedge m_2) \otimes 322
            +_{\max} (a_4 \wedge m_1) \otimes 176
        \end{aligned}$ & $
        \begin{aligned}
            &(\textcolor{red}{a_{1,3}} \wedge \textcolor{blue}{w_{3}}) \vee (a_2 \wedge \textcolor{blue}{w_{3}}) \vee ( \textcolor{red}{a_{1,3}} \wedge  \textcolor{goodgreen}{w_2}) \vee (\textcolor{orange}{w_1}) \qquad \ell = \{\textcolor{red}{a_{1,3}} \mapsto (a_1 \vee a_3),\\& \textcolor{orange}{w_{1}} \mapsto ((a_4\wedge m_1)\otimes 176), \textcolor{goodgreen}{w_2} \mapsto (m_2 \otimes 322),  \textcolor{blue}{w_{3}} \mapsto (m_3 \otimes 377), a_2\mapsto a_2\}
        \end{aligned}$ \\
        \bottomrule
    \end{tabular}
\label{tab:lineages_lifted}
\end{table*}

For an aggregate query $\langle \alpha,\gamma, Q\rangle$, let 
 $Q(D) = \{\bm t_1, \ldots, \bm t_n\}$. The lineage 
 of the query over $D$ is a \tensorset\ expression of the form:
{   \begin{align*}
\lin(\langle \alpha,\gamma, Q\rangle,D) = 
\sum\nolimits^M_i \lin(Q,D,\bm t_i) \otimes \gamma(\bm t_i),  
\end{align*}    }
where the summation uses the $+_M$ operator of the monoid $M$ associated with $\alpha$. For  $\minagg$ and $\maxagg$, this lineage is a semimodule expression since their corresponding monoids are idempotent. 
\nop{For convenience simplicity, we consider any fact $f$ such that $v(f)$ is not in the lineage to be exogenous. }

\begin{example}
    Table~\ref{tab:lineages_lifted} shows the lineage of the queries from Figure~\ref{fig:running_example}: 
    This is a positive DNF formula for $Q_1$ and a semimodule expression with the Max monoid for $Q_2$.
    The Boolean variables correspond to database facts, e.g., $a_1$ corresponds to  Cast('Brad Pitt').
\end{example}

\paragraph{Banzhaf and Shapley Values}
We use such values to measure the contribution of database facts to query results~\cite{TheShapleyValueofTuplesinQueryAnswering,Sig24:BanzhafValuesForFactsInQueryAnswering}.

\begin{definition}[Banzhaf Value]
Given a Boolean or aggregate query $Q$ and a database $D = D_x \cup D_n$, the Banzhaf value of an endogenous fact $f \in D_n$ is defined as:
{    \begin{align*}
    \banz(Q,f,D) \overset{def}{=} \sum_{D'\subseteq D_n\setminus \{f\}} Q(D' \cup \{f\} \cup D_x) - Q(D'\cup D_x)
\end{align*}     }

\end{definition}

\begin{definition}[Shapley Value]
    Let $C_{\bm Y}^{\bm X} = \frac{|\bm Y|!(|\bm X|-|\bm Y|-1)!}{|\bm X|!}$ for variable sets $\bm Y, \bm X$ (when $\bm X$ is clear from context, we just use  $ C_{\bm Y}$).
    Given a Boolean or aggregate query $Q$ and a database $D = D_x \cup D_n$, the Shapley value of an endogenous fact $f \in D_n$ is defined as:
{   \begin{align*}
    \shap(Q,f,D) \overset{def}{=}& \sum_{Y\subseteq D_n \setminus \{f\}} C^{\bm D_n}_{Y}\cdot \big(Q(Y \cup \{f\} \cup D_x) - Q(Y \cup D_x)\big)
\end{align*}    }

\end{definition} 

Banzhaf and Shapley values sum up marginal contributions of facts to query results over sub-databases. They differ in that the former weigh this contribution based on sub-database sizes. These weights guarantee that the sum of Shapley values over all facts equals the query result. This is not the case for Banzhaf. 

Given a Boolean formula or a \nop{semimodule expression} \tensorset\ $\Psi$ over a variable set $\bm X$, \nop{the Banzhaf value of $x \in \bm X$ is:}the Banzhaf and Shapley values of $x \in \bm X$ are:
{   \begin{align}
\label{eq: banzhaf value formula and semimodule}
\banz(\Psi,x) &= \sum_{\bm Y \subseteq \bm X \setminus \{x\}}
\Psi[\bm Y \cup \{x\}]- \Psi[\bm Y] \\
\label{eq: shapley value formula and semimodule} 
\shap(\Psi,x) &= \sum_{\bm Y \subseteq \bm X \setminus \{x\}}
C^{\bm X}_{\bm Y} \cdot (\Psi[\bm Y \cup \{x\}]- \Psi[\bm Y])
\end{align}    }

In case of a Boolean formula $\varphi$, the Banzhaf value of a variable $x$ can be computed based on the model counts of the formulas obtained from $\varphi$ by substituting $x$ by the constants $1$ and $0$~\cite{Sig24:BanzhafValuesForFactsInQueryAnswering}:
{   \begin{align}
\banz(\varphi,x) = \#\varphi[x:=1] - \#\varphi[x:=0]
\end{align}    }

Prior work connected the Banzhaf and Shapley values for database facts and those of Boolean variables ~\cite{TheShapleyValueofTuplesinQueryAnswering}. Given a Boolean or aggregate query $Q$, a database $D = D_x \cup D_n$ and a fact $f\in D_n$:
{   \begin{align*}
\banz(Q,D,f) =& \banz(\lin(Q,D), v(f))\\
\shap(Q,D,f) =& \shap(\lin(Q,D), v(f))
\end{align*}    }

For non-Boolean queries, the Banzhaf and Shapley values are defined with respect to each output tuple. 

\begin{example}
    Consider the lineage $\varphi$ of $Q_1$ in Table~\ref{tab:lineages_lifted}. We sketch the computation of the Banzhaf and Shapley values of the variable $a_1$ representing the fact Cast('Brad Pitt'). Consider the valuation $\theta = \{d_1, m_3\}$. The marginal contribution of $a_1$ on $\theta$ is $\varphi[\theta\cup \{a_1\}] - \varphi[\theta] = 1$ because $\theta$ is not a model for $\varphi$ but $\theta\cup \{a_1\}$ is. The Banzhaf value is achieved by summing up the marginal contributions across all valuations. For Shapley value, the marginal contributions are multiplied by the corresponding coefficients; e.g.,
    the contribution of $a_1$ to $\theta$ is multiplied by  $\frac{2!(7-2-1)!}{7!}$, since $|\theta| = 2$ and $|vars(\varphi)| = 7$.
    Consider now the lineage $\Phi$  of $Q_2$ in Table~\ref{tab:lineages_lifted} and the valuation $\theta' = \{m_3\}$. We have $\Phi[\theta']=0$, since $\theta'$ is not a model of any Boolean formula in the \tensorset $\Phi$. This is an example of the effect of our special treatment of aggregation over an empty set of values. Further, $\Phi[\theta' \cup \{a_1\}]=377$, so the marginal contribution of $a_1$ to the sub-database consisting of only $m_3$ is $377$.
\end{example}


\paragraph{Decomposition Trees}
We next introduce decomposition trees for Boole\-an and semimodule expressions.

\begin{definition}[~\cite{AggregationInProbabalisticDatabasesViaKnowlegeCompilation}] 
 A {\em decomposition tree} (d-tree), is defined recursively as follows:

\begin{itemize}

\item If $v$ is a $0$/$1$ constant or a variable then $v$ is a d-tree.

\item Let $T_{\varphi_1}, \ldots ,T_{\varphi_n}$ be d-trees for pairwise independent formulas $\varphi_1 , \ldots ,\varphi_n$, respectively,
let $T_\varphi$ and $T_\psi$ be d-trees for formulas $\varphi$ and $\psi$, respectively and let $v$ be a variable not in $vars(\varphi)\cup vars(\psi)$, Then,

\begin{center}
\begin{minipage}{0.2\linewidth}
\tikz {
 \node at (3.6,-1)  (n4) {$\oplus$};
\node at (3,-1.6)  (n3) {$T_{\varphi_1}$} edge[-] (n4);
\node at (3.6,-1.6)  (n3) {$...$};
\node at (4.2,-1.6)  (n3) {$T_{\varphi_n}$} edge[-] (n4);

}
\end{minipage}
\hspace{0.5em}
,
\hspace{0.5em}
\begin{minipage}{0.2\linewidth}
\tikz {
 \node at (3.6,-1)  (n4) {$\odot$};
\node at (3,-1.6)  (n3) {$T_{\varphi_1}$} edge[-] (n4);
\node at (3.6,-1.6)  (n3) {$...$};
\node at (4.2,-1.6)  (n3) {$T_{\varphi_n}$} edge[-] (n4);
}
\end{minipage}
\hspace{0.5em}
, and
\hspace{0.5em}
\begin{minipage}{0.2\linewidth}
\tikz {
 \node at (3.6,-1)  (n4) {$\sqcup_v$};
\node at (3.2,-1.6)  (n3) {$T_\varphi$} edge[-] (n4);
\node at (4.0,-1.6)  (n3) {$T_\psi$} edge[-] (n4);
}
\end{minipage}
\end{center}

are d-trees for $\bigvee_{i \in [n]} \varphi_i$, $\bigwedge_{i \in [n]} \varphi_i$, 
and $(v\wedge \varphi) \vee (\neg v\wedge \psi)$. The latter is the Shannon expansion. $\varphi$ and $\psi$ are the $1$ and respectively $0$ branches and $v$ is the condition variable.

\end{itemize}

We next extend d-trees to represent semimodule expressions~\cite{AggregationInProbabalisticDatabasesViaKnowlegeCompilation}:
\begin{itemize}
    \item Every monoid value $m\in \overline{\mathbb{R}}$ is a d-tree.
\item Let $T_{\varphi_1}, \ldots , T_{\varphi_n}$ be d-trees for
pairwise independent Boolean formulas $\varphi_1,\ldots ,\varphi_{n}$,
and $T_{\Phi}$ be a d-tree for a semimodule expression $\Phi$ whose Boolean part is independent from each of $\varphi_1,\ldots ,\varphi_{n}$. Let $T_{\Phi_1}, \ldots , T_{\Phi_n}$ be d-trees for
pairwise independent semimodule expressions $\Phi_1,\ldots ,\Phi_{n}$. Then, 
\begin{center}
\hfill
\begin{minipage}{0.2\linewidth}
\tikz {
 \node at (3.6,-1)  (n4) {$\otimes$};
\node at (2.9,-1.6)  (n3) {$T_{\varphi_1}$} edge[-] (n4);
\node at (3.4,-1.6)  (n3) {$\ldots$};
\node at (3.9,-1.6)  (n3) {$T_{\varphi_n}$} edge[-] (n4);
\node at (4.5,-1.6)  (n3) {$T_{\Phi}$} edge[-] (n4);
}
\end{minipage}
\hfill
and
\begin{minipage}{0.45\linewidth}
\begin{centering}
\tikz {
 \node at (3.6,-1)  (n4) {$\oplus$};
\node at (3.2,-1.6)  (n3) {$T_{\Phi_1}$} edge[-] (n4);
\node at (3.7,-1.6)  (n3) {$\ldots$};
\node at (4.2,-1.6)  (n3) {$T_{\Phi_n}$} edge[-] (n4);

}
\end{centering}
\end{minipage}
\end{center}
are d-trees for $(\bigwedge_{i \in [n]} \varphi_i) \otimes \Phi$ and $\sum_M \Phi_i$ respectively.
\end{itemize}
\end{definition}
The size of a d-tree is its number of nodes.
D-trees can be constructed by iteratively decomposing a Boolean formula or a semimodule expression into simpler parts.

\begin{example} Consider the semimodule expression $(a_1\wedge b_1) \otimes 3 +_M (a_1 \wedge b_2)\otimes 7$.
Using Axiom (3) of semimodules,
we rewrite it as
$(a_1 \otimes (b_1 \otimes 3)) +_M (a_1 \otimes (b_2 \otimes 7))$,
and applying Axiom (1) gives
$a_1 \otimes ((b_1 \otimes 3) +_M (b_2 \otimes 7))$.
This expression corresponds to a d-tree $\otimes(a_1, \oplus(\otimes(b_1, 3), \otimes(b_2, 7)))$.
\end{example}

The notions of variables, valuations and Banzhaf and Shapley values extend naturally to d-trees as they may be interpreted with respect to the Boolean formula represented by the d-tree. We thus use them for formulas and d-trees interchangeably.

%% file: Sections_Arxiv_New/Boolean.tex
\section{Attribution for Union of Conjunctive Queries}
\label{sec: Algorithm}

We present our algorithms \lexaban and \lexashap for computing Banzhaf and Shapley values for UCQs. Starting with the query lineage in DNF (computed by tools such as ProvSQL \cite{senellart2018provsql} or GProM \cite{AF18}), we first compile it into a d-tree using a novel lifting optimization (Section~\ref{subsubsect: lifting}), then introduce gradient-based algorithms (Section~\ref{subsubsect: arithmetic}) to compute Banzhaf/Shapley based on the lifted d-tree.    
\nop{
They are extended to aggregate queries in Section~\ref{sec: aggregates}.
The starting point for our algorithms is the state-of-the-art approach~\cite{Sig24:BanzhafValuesForFactsInQueryAnswering}, which first compiles the query lineage into a d-tree and then computes Banzhaf/Shapley values in one pass over the d-tree. Our algorithms employ two novel optimizations: lifting (Section \ref{subsubsect: lifting}) which yields a more efficient d-tree compilation, and gradient-based computation (Section \ref{subsubsect: arithmetic})  of Banzhaf and Shapley values by evaluating and differentiating arithmetic circuits based on the d-tree. 
}
\nop{Section \ref{subsect: LExaBan} unifies these techniques into a Banzhaf value computation algorithm. Section \ref{subsect: Expected_Banzhaf} extends our algorithm for probabilistic databases and 
Section \ref{sec: Aggregate_extension} extends it to handle aggregate queries.}



\begin{algorithm}[t]
\caption{$\lift$}
\label{alg:Lift}
\begin{algorithmic}[1]
\Require Lifted DNF formula $(\varphi, \ell)$
\Ensure Saturated lifted DNF formula
\While{$(\varphi, \ell)$ is not saturated}
\If{$\varphi$ has a maximal set $\bm V$ of cofactor-equivalent variables with $|\bm V| >1$}
    \State  $(\varphi , \ell) \gets \dislift(\varphi, \ell, \bm V)$
\EndIf
\If{$\varphi$ has a maximal set $\bm V$ of interchangeable variables with $|\bm V| >1$}
  \State $(\varphi, \ell) \gets \conlift(\varphi,\ell, \bm V)$
\EndIf
\EndWhile
\State   \Return $(\varphi, \ell)$
\end{algorithmic}
\end{algorithm}

\subsection{Lifted Compilation of Lineage Into D-tree}
\label{subsubsect: lifting}

\nop{
A Boolean formula may admit multiple equivalent representations, each balancing the size and the computational tractability of model counting and  Banzhaf and Shapley value computation. Knowledge compilation aims to transform succinct yet intractable representations into tractable and reasonably succinct representations. In this work, we consider query lineages as the input representation and d-trees as the target compiled form. Our key insight is that leveraging structural symmetries within the query lineage during compilation reduces both the compilation time and the size of the resulting d-tree. To achieve this, we introduce \emph{lifting}, a technique that identifies conjunctions or disjunctions of variables with identical Banzhaf or Shapley values and replaces them with fresh variables.}

\nop{
\footnote{Unlike query-based lineage factorization~\cite{OlteanuZ12}, which applies uniformly across all lineages of a query, lifting is a fine-grained factorization of the query lineage specific to the data instance. Whereas achieving minimal-size instance-specific factorization of query lineage is computationally expensive~\cite{Neha:MinFactorization:2024}, lifting can be computed efficiently.}
}

    

Consider two disjoint sets $\bm X$ and $\bm Y$ of Boolean variables. 
A {\em lifted formula} over $\bm X \cup \bm Y$ is a pair 
$(\varphi, \ell)$, where $\varphi \in \bool(\bm Y \cup \bm X)$ is a formula over $\bm Y \cup \bm X$ and $\ell: \vars(\varphi) \rightarrow \bool(\bm X)$
maps each variable in $\varphi$ to a formula over $\bm X$.  
The {\em inlining} $\inline(\varphi, \ell)$ of a lifted formula 
is obtained by replacing each variable $y$ in $\varphi$ by 
the formula $\ell(y)$.
A {\em lifted DNF formula} is a lifted formula $(\varphi, \ell)$,
where $\varphi$ is in DNF.
In the rest of this section, we refer only to positive DNF formulas.

\begin{example}
\label{ex:inlining}
Consider the DNF formulas
$\varphi_0 = (x_1 \wedge x_2 \wedge x_3) \vee (x_4 \wedge x_5 \wedge x_3)$, 
$\varphi_1 = (y_1 \wedge y_2) \vee (y_3 \wedge y_2)$, and 
$\varphi_2 = (y_4 \wedge y_2)$ and the functions 
$\ell_0 = \{x$ $\mapsto x | x \in \vars(\varphi_0)\}$,
$\ell_1 = \{y_1$ $\mapsto (x_1\wedge x_2),$ $y_2$ $\mapsto x_3,$ 
$y_3\mapsto$ $(x_4\wedge x_5)\}$, and
$\ell_2 = \{y_2\mapsto x_3, y_4 \mapsto  \big((x_1\wedge x_2) \vee 
(x_4 \wedge x_5)\big) \}$.    
We have $\inline(\varphi_0, \ell_0) = \varphi_0$, 
$\inline(\varphi_1, \ell_1)$ $=$ $\big((x_1 \wedge x_2) \wedge x_3 \big)  \vee \big((x_4 \wedge x_5) \wedge x_3 \big)$, and
$\inline(\varphi_2, \ell_2)$ $=$ $\big(x_1 \wedge x_2) \vee (x_4 \wedge x_5)\big) \wedge x_3$.
The latter two inlinings are equivalent to the formula $\varphi_0$.
\end{example}

Consider a DNF formula $\varphi = \bigvee_{i \in n} C_i$, where each 
clause $C_i$ is a conjunction of variables.
For a clause $C_i$ and variable $x$, we denote by $C_i \setminus \{x\}$
the clause that results from $C_i$ by omitting $x$.
We call the set $\{C_i \setminus \{x\}\mid i \in [n], x \in C_i\}$ the {\em cofactor} of $x$.
Two variables in $\varphi$ are called {\em cofactor-equivalent} if they have the same cofactor. 
Two variables $x$ and $y$ are called {\em interchangeable} if for each 
clause $C_i$, it holds that $x$ appears in $C_i$ if and only if 
$y$ appears in $C_i$.
A variable set $\bm V$ is called a {\em maximal set of cofactor-equivalent (interchangeable)} variables 
if every pair of variables in $\bm V$ is 
cofactor-equivalent (interchangeable)
and this does not hold for any superset of $\bm V$. 
A lifted DNF formula $(\varphi, \ell)$ is called {\em saturated} if it contains neither a set of cofactor-equivalent variables nor a set of interchangeable variables of size greater than one.

\begin{example}
\label{ex:cofactor_interchange}
Consider again the formulas $\varphi_0$, $\varphi_1$, and $\varphi_2$
from Example~\ref{ex:inlining}.
The sets $\{x_1, x_2\}$ and $\{x_4, x_5\}$ are both maximal sets 
of interchangeable variables in $\varphi_0$.
The set $\{y_1, y_3\}$ is a maximal set of cofactor-equivalent variables in $\varphi_1$ with cofactor $\{y_2\}$. Whereas the formulas $\varphi_0$ and $\varphi_1$ are not saturated, the formula $\varphi_2$ is.
\end{example}

\nop{
\begin{example}
Consider the formulas $\varphi_1 = a_1$, $\varphi_2 = (a_2\vee b_2)$, $\varphi_3 = (a_3\wedge b_3)\vee (a_4 \wedge b_3) \vee (a_3 \wedge b_5)$, $\varphi_4 = (a_3 \wedge (c_1 \vee c_2))$, and $\varphi_5 = a_2$.
    The formula $\varphi = (\varphi_1 \wedge \varphi_2) \vee (\varphi_1 \wedge \varphi_4)$ is not in DNF but in lifted DNF. 
    The formula $\varphi_2 = (\varphi_3 \wedge \varphi_1)$ is not in lifted DNF, since $\varphi_3$ is not read-once.
    Likewise, $\varphi_3 = (\varphi_1 \wedge \varphi_2) \vee (\varphi_1 \wedge \varphi_5)$ is not in lifted DNF, since $\varphi_2$ and $\varphi_5$ share the variable $a_2$.
\end{example}
}

\paragraph{Lifting DNF Formulas}
We describe how a DNF formula can be
translated into a saturated lifted DNF formula. 
We first introduce some further notation. 
Consider a lifted DNF formula 
$(\varphi, \ell)$ over $\bm X \cup \bm Y$, where
$\varphi = \bigwedge_{i \in [n]} C_i$.
Let $\bm V = \{
y_1 , \ldots , y_k\}$ be a set of
cofactor-equivalent variables in $\varphi$, where each variable has 
cofactor $\{N_{1}, \ldots ,N_{p}\}$. 
Let $C_{i_1}, \ldots , C_{i_m}$ be all clauses in $\varphi$
containing variables from $\bm V$.
We define  
$\dislift(\varphi,\ell, \bm V)$ to be the lifted DNF formula
$(\varphi',\ell')$, where: (1) $\varphi'$ results from $\varphi$
by omitting the clauses $C_{i_1}, \ldots , C_{i_m}$ and adding the new clauses $y \wedge N_1, \ldots, y \wedge N_p$ for a fresh variable 
$y \in \bm Y \setminus \vars(\varphi)$;
(2) the function $\ell'$ is defined by $\ell'(y) = \bigvee_{i \in [k]} \ell(y_i)\}$ and 
$\ell'(y') = \ell(y')$ for all $y' \in \vars(\varphi')$ with $y' \neq y$. 
Now, assume that $\bm V = \{y_1 , \ldots , y_k\}$ consists of interchangeable variables in $\varphi$. 
We denote by $\conlift(\varphi,\ell, \bm V)$ the lifted DNF formula 
$(\varphi',\ell')$, where: (1) $\varphi' = C_1' \wedge \cdots \wedge 
C_n'$ such that, given a fresh variable $y \in \bm Y \setminus \vars(\varphi)$, each $C_i'$ results from $C_i$  by replacing the variables $y_1, \ldots ,y_k$ by $y$; (2) $\ell'$ is defined by $\ell'(y) = \bigwedge_{i \in [k]} \ell(y_i)\}$ and 
$\ell'(y') = \ell(y')$ for all $y' \in \vars(\varphi')$ with $y' \neq y$.


\nop{
\begin{example}
Consider the DNF formula $\varphi$ $=$ 
$(x \wedge x_1 \wedge x_2) \vee  (x \wedge  x_1 \wedge x_3) \vee
(y \wedge x_1 \wedge x_2) \vee (y \wedge x_1 \wedge x_3) \vee
(z \wedge x_1)$, where $x$ and $y$ have the same neighborhood 
$\{\{x_1,$ $x_2\},$ $\{x_1,$ $x_3\}\}$. It holds 
$\dislift(\varphi,$ $\{x,y\})$ $=$ 
$\big((x\vee y) \wedge x_1 \wedge x_2\big) \vee 
\big((x\vee y) \wedge  x_1 \wedge  x_3\big) \vee \big(z \wedge x_1 \big)$.

Now, consider the DNF formula $\varphi$ $=$ 
$(x \wedge y \wedge x_1) \vee (x \wedge y \wedge x_2) \vee (x \wedge y \wedge x_3)$, where the variables $x$
and $y$ are interchangeable. We have 
$\conlift(\varphi, \{x,y\})$ $=$ $\big((x\wedge y) , x_1\big) \vee 
\big((x\wedge y) \wedge  x_2\big) \vee \big((x\wedge y) \wedge  x_3\big)$.
\end{example}
}

The function $\lift$ in Algorithm~\ref{alg:Lift} transforms a lifted DNF formula $\varphi$ into a saturated one.
It repeatedly invokes $\dislift$ or $\conlift$ while $\varphi$ contains a variable set $\bm V$ with $|\bm V| >1$, such that  all variables in $\bm V$ are either cofactor-equivalent or interchangeable, respectively.
Intuitively, both cofactor-equivalent and interchangeable variables exhibit a symmetry in the structure of the formula. Consequently, we can simplify the formula without loss of information.

\begin{proposition}
\label{prop:lifting_properties}
Let $\varphi$ be a DNF formula, $\ell$ the identity function on $\vars(\varphi)$, and $(\varphi', \ell')$ the output of $\lift(\varphi, \ell)$ in Algorithm~\ref{alg:Lift}. Then, $(\varphi', \ell')$ is a saturated lifted formula such that: (1) $\inline(\varphi', \ell')$ is equivalent to $\varphi$. (2) $\ell'(x)$ and $\ell'(y)$ are independent formulas for each distinct variables $x,y \in \vars(\varphi')$. (3) $\ell'(x)$ is a read-once formula for each 
$x \in \vars(\varphi')$.
\end{proposition}

\begin{proof}
Let $\varphi$ be a DNF formula, $\ell$ the identity function on $\vars(\varphi)$, 
and $(\varphi', \ell')$ the output of $\lift(\varphi, \ell)$ in Algorithm~\ref{alg:Lift}.
Consider the sequence 
$(\varphi, \ell)$ $=$ $(\varphi_1, \ell_1),$ $\ldots,$ $(\varphi_n, \ell_n)$ $=$ 
$(\varphi', \ell')$, where $(\varphi_2, \ell_2),$ $\ldots,$ $(\varphi_{n-1}, \ell_{n-1})$
are the intermediate lifted formulas constructed by the algorithm.
It follows from the stopping condition at Line~1 of the algorithm 
that $(\varphi', \ell')$ is saturated. 

In the following, we  show by induction on $i$ that Properties (1)--(3)
from Proposition~\ref{prop:lifting_properties}
hold for each $(\varphi_i, \ell_i)$ with $i \in [n]$, i.e., (1) $\inline(\varphi_i, \ell_i)$ is equivalent to $\varphi$, 
(2) $\ell_i(x)$ and $\ell_i(y)$ are independent formulas for each distinct variables $x,y \in \vars(\varphi_i)$, and
(3) $\ell_i(x)$ is a read-once formula for each 
$x \in \vars(\varphi_i)$.

\smallskip

\noindent {\it Base Case}:
Since $\ell_1$ is the identity function on $\vars(\varphi)$
and $\varphi_1 = \varphi$, Properties (1)--(3)
obviously hold for  $(\varphi_1, \ell_1)$.

\smallskip

\noindent {\it Induction Hypothesis}:
We assume that Properties (1)--(3)
hold for $(\varphi_i, \ell_i)$ for some $i \in [n-1]$.

\smallskip

\noindent {\it Induction Step}: 
We show that Properties (1)--(3) hold for 
$(\varphi_{i+1}, \ell_{i+1})$. Without loss of generality, we assume that $(\varphi_{i+1}, \ell_{i+1})$ results from $(\varphi_{i}, \ell_{i})$  by the application of the $\dislift$ operator. 
The case of the $\conlift$ operator is handled analogously. 
Let $\bm V = \{y_1 , \ldots , y_k\}$ be a maximal set of cofactor-equivalent variables 
in $\varphi_i$ such that each variable $y_i$ has cofactor $\{N_{1}, \ldots ,N_{p}\}$.
Let $C_{i_1}, \ldots , C_{i_m}$ be all clauses in $\varphi_{i}$
containing a variable from $\bm V$. By assumption, $\varphi_{i+1}$ results from $\varphi_i$ by omitting the clauses $C_{i_1}, \ldots , C_{i_m}$ and adding the clauses $y \wedge N_1, \ldots , y \wedge N_p$ for a fresh variable $y$; 
the function $\ell_{i+1}$ is defined by $\ell_{i+1}(y) = \bigvee_{j \in [k]} \ell_i(y_j)$ and 
$\ell_{i+1}(y') = \ell_i(y')$ for all $y' \in \vars(\varphi_{i+1})$ with $y' \neq y$.

We show that Property (1) holds for $(\varphi_{i+1},$ $\ell_{i+1})$, i.e., $\inline(\varphi_{i+1},$ $\ell_{i+1})$ is equivalent to $\varphi$. Let $\varphi_{i}'$ be the formula obtained from $\varphi_i$ by replacing $y$ with $\bigvee_{j \in [k]} y_j$.
It follows from the distributivity of conjunction over disjunction that 
$\varphi_{i}'$ is equivalent to $\varphi_{i}$. Hence, 
$\inline(\varphi_i', \ell_i)$ is equivalent to $\inline(\varphi_i, \ell_i)$.
By induction hypothesis, $\inline(\varphi_i, \ell_i)$, hence  $\inline(\varphi_i', \ell_i)$, is equivalent to $\varphi$.
By construction, $\inline(\varphi_i', \ell_i)$ $=$ $\inline(\varphi_{i+1}, \ell_{i+1})$.
It follows that $\inline(\varphi_{i+1}, \ell_{i+1})$ is equivalent to $\varphi$.

We prove that Property (2) holds for $(\varphi_{i+1}, \ell_{i+1})$,
i.e., $\ell_{i+1}(x)$ and $\ell_{i+1}(y)$ are independent formulas for each distinct variables $x,y \in \vars(\varphi_{i+1})$. For the sake of contradiction,
assume that Property (2) does not hold for $(\varphi_{i+1}, \ell_{i+1})$. Since Property (2) holds for 
$(\varphi_{i}, \ell_{i})$ (induction hypothesis), there must be a variable $x \in \vars(\varphi_{i+1})$ 
such that $z$ appears in 
$\ell_{i+1}(x)$ and $\ell_{i+1}(y)$, where $y$ is the fresh variable in $\varphi_{i+1}$.
From the construction of $\ell_{i+1}$, it follows that 
there must be a variables $z' \in \vars(\varphi_i)$ such that $z$ appears in $\ell_{i}(x)$ and $\ell_{i}(z')$. This, however, contradicts our assumption that Property (2) holds for $(\varphi_{i}, \ell_{i})$.

Now, we prove that Property (3) holds for $(\varphi_{i+1}, \ell_{i+1})$, i.e., $\ell_{i+1}(x)$ is a read-once formula for each $x \in \vars(\varphi_{i+1})$. Consider a variable $x \in \vars(\varphi_{i+1})$. If $x$ is not the fresh variable $y$, then  $x \in \vars(\varphi_{i})$
and $\ell_i(x) =\ell_{i+1}(x)$. By induction hypothesis, $\ell_{i}(x)$ is a read-once formula.
If $x$ is equal to the fresh variable $y$, then 
$\ell_{i+1}(x) = \bigvee_{j \in [k]} \ell_{i}(y_j)$. 
By induction hypothesis, the formulas 
$\ell_{i}(y_1), \ldots , \ell_{i}(y_k)$ are independent and read-once.
Hence, $\ell_{i+1}(x)$ is a read-once formula.  
\end{proof}




\begin{example}
\label{example:lifted lineage}
We explain how  the procedure $\lift$  transforms the lineage 
$\varphi_0$ of the query $Q_1$ depicted in Table \ref{tab:lineages_lifted} (top left).
Let $\ell_0$ be the identity function that maps 
every variable in $\varphi_0$ to itself.
The variables $d_1$ and $d_2$ have the same cofactor 
$\big\{\{a_1,m_3\},$ $\{a_1,m_2\},$ $\{a_3,m_3\},$
$\{a_2,m_3\},$ $\{a_3,m_2\}\big\}$.
The procedure executes $\dislift(\varphi_0, \ell_0, \{d_1,d_2\})$,
which returns the lifted formula $(\varphi_1, \ell_1)$, where
$
\varphi_1 =\  (d_{1,2} \wedge a_1 \wedge  m_3) \vee 
(d_{1,2} \wedge a_1 \wedge m_2)
\vee (d_{1,2} \wedge a_3 \wedge m_3) \vee
(d_{1,2} \wedge a_2 \wedge  m_3) \vee 
(d_{1,2} \wedge a_3 \wedge m_2),
$
and $\ell_1 = \{d_{1,2} \mapsto (d_1 \vee d_2)\} \cup \{x \mapsto x| x \in \vars(\varphi) \text{ and } x \neq d_{1,2}\}$.
In $\varphi_1$, the variables $a_1$ and $a_3$ have the same 
cofactor $\{\{d_{1,2},$ $m_3\},$ $\{d_{1,2},$ $m_2\}\}$.
The procedure executes $\dislift(\varphi_1,\ell_1, \{a_1,a_3\})$ and obtains $(\varphi_2, \ell_2)$
shown in Table~\ref{tab:lineages_lifted} (top right).
\nop{, where 
\begin{align*}
\varphi_2 =  ({\color{red}d_{1,2}} \wedge {\color{green}(a_{1,3}} \wedge m_3) \vee 
({\color{red}d_{1,2}} \wedge {\color{green}(a_{1,3}}
\wedge m_2) \vee
(\textcolor{red}{d_{1,2}} \wedge a_2 \wedge m_3), 
\end{align*}
and $\ell_2 = \{d_{1,2} \mapsto d_1 \vee d_2, a_{1,3}\mapsto (a_1 \vee a_3)\}$.}
The formula $\varphi_2$ does not have two or more variables that are 
cofactor-equivalent or interchangeable. Hence, $(\varphi_2, \ell_2)$ is saturated. 
\end{example}

\nop{
\begin{definition}
A formula $\varphi$ is in {\em lifted DNF} if it is of the form
 \begin{align}
\varphi =& \bigvee_{i=1}^m \left( \bigwedge_{j=1}^{n_i} \varphi_{i,j} \right)
\label{eq: Lifted_formula}
\end{align}
where each $\varphi_{i,j}$ is read-once and any two sub-formulas $\varphi_{i_1,j_1}$ and $\varphi_{i_2,j_2}$ are independent or syntactically the same formula.
\end{definition}

We call each sub-formula 
$\bigwedge_{j=1}^{n_i} \varphi_{i,j}$ of a lifted
DNF formula a {\em clause}.
By definition, every formula in DNF is also in lifted DNF. However, there are lifted DNF formulas that are not in DNF, as illustrated next:
\begin{example}
    Consider the formulas $\varphi_1 = a_1$, $\varphi_2 = (a_2\vee b_2)$, $\varphi_3 = (a_3\wedge b_3)\vee (a_4 \wedge b_3) \vee (a_3 \wedge b_5)$, $\varphi_4 = (a_3 \wedge (c_1 \vee c_2))$, and $\varphi_5 = a_2$.
    The formula $\varphi = (\varphi_1 \wedge \varphi_2) \vee (\varphi_1 \wedge \varphi_4)$ is not in DNF but in lifted DNF. 
    The formula $\varphi_2 = (\varphi_3 \wedge \varphi_1)$ is not in lifted DNF, since $\varphi_3$ is not read-once.
    Likewise, $\varphi_3 = (\varphi_1 \wedge \varphi_2) \vee (\varphi_1 \wedge \varphi_5)$ is not in lifted DNF, since $\varphi_2$ and $\varphi_5$ share the variable $a_2$.
\end{example}

\paragraph{Translation from DNF to Succinct Lifted DNF}
We describe how a DNF formula can be systematically 
transformed into a lifted DNF formula, where the number of clauses and the number of conjuncts in each clause are reduced. First, we introduce some notation. For ease of presentation, we represent each clause in a lifted DNF formula by the set of its conjuncts and a lifted 
DNF formula by the set of its clauses. 
Consider a lifted DNF formula $\varphi =\{C_1, \ldots, C_n\}$ consisting of the clauses $C_1, \ldots ,C_n$. 
For a variable $x$, we call the set $\{C_i - \{x\}\mid i \in [n], x \in C_i\}$ the {\em neighborhood} of $x$.
Two variables in $\varphi$ are called {\em neighborhood-equivalent} if they have the same neighborhood. 
We say that two variables $x$ and $y$ are {\em interchangeable} if for each 
clause $C_i$, it holds $x \in C_i$ if and only if $y \in C_i$.
A variable set $\bm V$ is called a {\em maximal set of neighborhood-equivalent (interchangeable)} variables 
if every pair of variables in $\bm V$ is 
neighborhood-equivalent (interchangeable)
and this does not hold for any superset of $\bm V$. 
A lifted DNF formula is called {\em saturated} if it contains neither a set of neighborhood-equivalent variables nor a set of interchangeable variables of size greater than one.
Consider a set $\bm V = \{x_1 , \ldots , x_\ell\}$
of neighborhood-equivalent variables in $\varphi$, where each variable has 
neighborhood $\{C_{j_1}, \ldots ,C_{j_m}\}$. 
Let $C_{i_1}, \ldots , C_{i_k}$ be the clauses in $\varphi$
containing variables from $\bm V$.
We define  
$\dislift(\varphi,\bm V)$ to be the lifted DNF formula 
$\{C_1, \ldots, C_n\}$ $\setminus$ $\{C_{i_1}, \ldots , C_{i_k}\}$ $\cup$ $\{C_{j_1} \cup \{\bigvee_{i\in [\ell]} x_i\}, \ldots ,C_{j_m} \cup \{\bigvee_{i\in [\ell]} x_i\} \}$.
Now, assume that $\bm V = \{x_1 , \ldots , x_\ell\}$ consists 
of interchangeable variables in $\varphi$. 
We denote by $\conlift(\varphi,\bm V)$ the lifted DNF formula $\{C_1', \ldots, C_n'\}$, where each 
$C_i'$ is defined as follows: if $C_i$ does not contain the variables in $\bm V$, then $C_i' = C_i$, otherwise $C_i' = C_i\setminus\{x_1, \ldots ,x_\ell\} \cup \{\bigwedge_{i \in [\ell]} x_i\}$. 

\begin{example}
Consider the DNF formula $\varphi$ $=$ 
$(x \wedge x_1 \wedge x_2) \vee  (x \wedge  x_1 \wedge x_3) \vee
(y \wedge x_1 \wedge x_2) \vee (y \wedge x_1 \wedge x_3) \vee
(z \wedge x_1)$, where $x$ and $y$ have the same neighborhood 
$\{\{x_1,$ $x_2\},$ $\{x_1,$ $x_3\}\}$. It holds 
$\dislift(\varphi,$ $\{x,y\})$ $=$ 
$\big((x\vee y) \wedge x_1 \wedge x_2\big) \vee 
\big((x\vee y) \wedge  x_1 \wedge  x_3\big) \vee \big(z \wedge x_1 \big)$.

Now, consider the DNF formula $\varphi$ $=$ 
$(x \wedge y \wedge x_1) \vee (x \wedge y \wedge x_2) \vee (x \wedge y \wedge x_3)$, where the variables $x$
and $y$ are interchangeable. We have 
$\conlift(\varphi, \{x,y\})$ $=$ $\big((x\wedge y) , x_1\big) \vee 
\big((x\wedge y) \wedge  x_2\big) \vee \big((x\wedge y) \wedge  x_3\big)$.
\end{example}

\begin{algorithm}[t]
\caption{$\lift$}
\label{alg:Lift}
\begin{algorithmic}[1]
\Require Lifted DNF formula $\varphi$
\Ensure Saturated lifted DNF formula
\While{$\varphi$ is not saturated}
\If{$\varphi$ has a maximal set $\bm V$ of neighborhood-equivalent variables with $|\bm V| >1$}
    \State  $\varphi \gets \dislift(\varphi, \bm V)$
\EndIf
\If{$\varphi$ has a maximal set $\bm V$ of interchangeable variables with $|\bm V| >1$}
  \State $\varphi \gets \conlift(\varphi, \bm V)$
\EndIf
\EndWhile
\State   \Return $\varphi$
\end{algorithmic}

\end{algorithm}

The function $\lift$ described in Algorithm~\ref{alg:Lift} transforms any given lifted DNF formula into a saturated lifted DNF formula by repeatedly calling the function $\dislift$ or $\conlift$ as long as the formula contains a set 
$\bm V$ of variables with $|\bm V| >1$, such that  all variables in $\bm V$ are either neighborhood-equivalent or interchangeable, respectively.

\begin{example}
\label{example:lifted lineage}
We explain how  the function $\lift$  transforms the lineage $\varphi$ of the query $Q_1$ from Table \ref{tab:lineages_lifted} (top left).
The variables $d_1$ and $d_2$ have the same neighborhood 
$\big\{\{a_1,m_3\},$ $\{a_1,m_2\},$ $\{a_3,m_3\},$
$\{a_2,m_3\},$ $\{a_3,m_2\}\big\}$.
We execute $\dislift(\varphi, \{d_1,d_2\})$ and obtain
\begin{align*}
\varphi_1 =\ & ({\color{red}d_{1,2}} \wedge a_1 \wedge  m_3) \vee 
(\textcolor{red}{d_{1,2}} \wedge a_1 \wedge m_2)
\vee (\textcolor{red}{d_{1,2}} \wedge a_3 \wedge m_3) \vee \\
& ({\color{red}d_{1,2}} \wedge a_2 \wedge  m_3) \vee 
(\textcolor{red}{d_{1,2}} \wedge a_3 \wedge m_2),
\end{align*}
where $d_{1,2}$ stands for $(d_1 \vee d_2)$.
In $\varphi_1$, the variables $a_1$ and $a_3$ have both the 
neighborhood $\{\{d_{1,2}, m_3\}, \{d_{1,2}, m_2\}\}$.
We execute $\dislift(\varphi_1, \{a_1,a_3\})$ and obtain 
the formula 
\begin{align*}
\varphi_2 =  ({\color{red}d_{1,2}} \wedge {\color{green}(a_{1,3}} \wedge m_3) \vee 
({\color{red}d_{1,2}} \wedge {\color{green}(a_{1,3}}
\wedge m_2) \vee
(\textcolor{red}{d_{1,2}} \wedge a_2 \wedge m_3), 
\end{align*}
where $a_{1,3}$ abbreviates $(a_1 \vee a_3)$.
The formula $\varphi_2$ does not have two variables that are neighborhood-equivalent or interchangeable. Hence, it is saturated. 
\end{example}
}

\nop{
\paragraph{Translation from DNF to Succinct Lifted DNF}
Algorithm \ref{alg:Lift} outlines the procedure 
Lift that transforms a DNF formula into an equivalent lifted DNF formula with a small number of variables. 

Given a DNF $\varphi$, lines \ref{alg_line: Lift: recursive_start}-\ref{alg_line: Lift: recursive_end} recursively call the procedure of lifting until we reach a final lifted form $\varphi^L$. 
The lifting procedure begins by identifying sets of variables that can be grouped to conjunction or disjunction. In the following, we will refer to a DNF clause as the set of variables the appear in it, and to a DNF formula as a set of its clauses. For a variable $x$ and a DNF formula $\varphi$ we define the neighbor set of x as the $N_x = \{clause\setminus\{x\} \mid x\in clause, clause\in \varphi\}$. Lines \ref{alg_line: Lift: collect_neighbors_start}-\ref{alg_line: Lift: collect_neighbors_end} generate the neighbor set for each variable.
Line \ref{alg_line: Lift: group_disjunction} generates disjunction of sets of variables if their neighbors sets are identical. 
Line \ref{alg_line: Lift: group_conjunction} generates conjunction of sets of variables if they appear in the same clauses.

The algorithm then proceeds to update the formula: For each identified disjunction $d_i = \bigvee v_i$ s.t. $|d_i| > 1$ the procedure in line \ref{alg_line: replace_disjunction_group} selects an arbitrary variable $v_a\in d_i$, removes any clauses holding variables from $v_i\in \vars(d_i)\setminus \{v_a\}$, and replaces $v_a$ with $d_i$.
For each identified conjunction $c_i = \bigwedge v_i$ s.t. $|c_i| > 1$, the procedure in line \ref{alg_line: replace_conjunction_group} selects an arbitrary variable $v_a\in \vars(c_i)$, removes all other variables $v_i\in \vars(c_i)\setminus \{v_a\}$ from the formula, and replaces $v_a$ with $c_i$.

\begin{example}
\label{example:lifted lineage}
Let us examine the algorithm on the lineage of $Q_1$ from Table \ref{tab:lineages_lifted}.
To achieve the lifted form from the non lifted lineage, we begin by applying the lifted procedure. The procedure collects the neighbor sets for each variable. In this example, the neighbor set of $d_1$ is: \begin{align*}
    N_{d_1} = \big\{\{a_1,m_2\}, \{a_1,m_3\}, \{a_2,m_2\}, \{a_2,m_3\}, \{a_3,m_3\}\big\}
\end{align*}
The procedure continues for the rest of the variables, and detects that the neighbor set for $d_2$ is similar to that of $d_1$, and the neighbor set of $a_1$ is similar to that of $a_3$, thus lifting them with disjunction. The procedure proceeds to update the formula to account for the new lifted variables. We mark single lifted variables with unique colors. For $\textcolor{red}{(d_1\vee d_2)}$ we arbitrarily select $d_1$, remove all clauses with $d_2$ and then replace every occurrence of $d_1$ with $\textcolor{red}{(d_1\vee d_2)}$. The final updated form is presented in Table \ref{tab:lineages_lifted}.
\nop{\begin{align*}
\varphi' = & (\textcolor{red}{(d_1\vee d_2)} \wedge \textcolor{green}{(a_1\vee a_3)} \wedge m_3) \vee (\textcolor{red}{(d_1\vee d_2)} \wedge \textcolor{green}{(a_1\vee a_3)} \wedge m_2)\\ \vee&(\textcolor{red}{(d_1\vee d_2)} \wedge a_2 \wedge m_3)
\end{align*}}
Since we managed to create a non trivial disjunction, The algorithm calls the lifting procedure again and tries again to detect new lifting opportunities. This time, the neighbor set of the new lifted variable $\textcolor{red}{(d_1\vee d_2)}$ is $N_{\textcolor{red}{(d_1\vee d_2)}} = \big\{\{\textcolor{green}{(a_1\vee a_3)},m_2\}, \{\textcolor{green}{(a_1\vee a_3)},m_3\}, \{a_2,m_3\}\big\}$.
This time, there are no lifting opportunities, and the algorithm stops, returning the resulted formula as the lifted lineage.

\end{example}

\begin{algorithm}[t]
\caption{Lift}
\label{alg:Lift}
\begin{algorithmic}[1]
\Require DNF lineage $\varphi$ for query $Q$ on database $D$
\Ensure $success, \varphi^L$
\While{True}
    \State $success, \varphi \gets$ \textsc{SingleLift}($\varphi$) \label{alg_line: Lift: recursive_start}
    \If{Not $success$}
        \State \textbf{break}
    \EndIf
\EndWhile \label{alg_line: Lift: recursive_end}
\item[]
\Procedure{SingleLift}{$\varphi$} \label{alg_line: Lift: collect_neighbors_start}
    \For{each variable $x$ in $\varphi$} \label{alg_line: Lift: collect_neighbors_start}
        \State $N_x \gets$ neighbors set of $x$ \label{alg_line: Lift: collect_neighbors_end}
    \EndFor
    \State Generate a disjunction $\bigvee_{k\in S} v_k$ if $\forall_{i,j\in S}, N_{v_i}=N_{v_j}$ \label{alg_line: Lift: group_disjunction}
    \State Generate a conjunction $\bigwedge_{k\in S} v_k$ if $\forall_{i,j\in S}, \forall c\in N_{v_i}, v_j\in c$ \label{alg_line: Lift: group_conjunction}
    \State $lifted \gets$ found a non-trivial conjunction or disjunction
    \For{each generated disjunction $d_i$ s.t. $|d_i| > 1$}
        \State UpdateDisjunction($\varphi$, $d_i$) \label{alg_line: replace_disjunction_group}
    \EndFor
    \For{each generated conjunction $c_i$ s.t. $|c_i| > 1$}
        \State UpdateConjunction($\varphi$, $c_i$) \label{alg_line: replace_conjunction_group}
    \EndFor
    \Return $<lifted,\varphi>$
\EndProcedure
\end{algorithmic}
\end{algorithm}

}

\begin{algorithm}[h]
\caption{$\liftedcompile$}
\label{alg:lifted_compile}
\begin{algorithmic}[1]
\Require DNF lineage $\varphi$ for query $Q$ on database $D$
\Ensure d-tree for $\varphi$
\State $(\psi, \ell) \gets \lift(\varphi, \ell')$ where $\ell'$
is identity function on $\vars(\varphi)$ 
\label{alg_line: liftedCompile: preprocessing}
\State \Return $\compile(\psi, \ell)$
\item[]
\Procedure{$\compile$}{$\psi, \ell$}
\If{$\psi$ is a variable $x$}
    \Return $\ell(x)$ \label{alg_line: liftedCompile: trivial}
\EndIf
\SWITCH{ $\psi$}
     \State \TAB \CASE{ $\psi_1 \vee \dots \vee \psi_n$ for independent $\psi_1, \dots, \psi_n$} \label{alg_line: liftedCompile: ind_or}
         \State \TAB \TAB $T_\psi\ \gets \bigoplus_{i \in [n]} \compile(\psi_i, \ell|_{\psi_i})$ 
    \State \TAB \CASE{ $\psi_1 \wedge \dots \wedge \psi_n$ for independent $\psi_1, \dots, \psi_n$} \label{alg_line: liftedCompile: ind_and}
    \State \TAB \TAB $T_\psi \gets \bigodot_{i \in [n]}\compile(\psi_i, \ell|_{\psi_i})$
    \State \TAB \DEFAULT
        \State \TAB \TAB Pick a most frequent variable $y$ in $\psi$ \label{alg_line: liftedCompile: shannon_start}
        \State \TAB \TAB $T_\psi \gets (y \odot\compile( \lift(\psi[y:=1], \ell|_{\psi[y:=1]}))) \oplus\\ \TAB\TAB\TAB\TAB\STAB (\neg y \odot \compile(\lift(\psi[y:=0], \ell|_{\psi[y:=0]}))$ \label{alg_line: liftedCompile: shannon_end}

\Return $T_{\psi}$
\EndProcedure
\end{algorithmic}
\end{algorithm}
\paragraph{D-Tree Compilation Using Lifting}
Algorithm \ref{alg:lifted_compile} gives our lifted  
compilation procedure for the query lineage (Boolean formula) into a d-tree. The input lineage $\varphi$ is first translated into a saturated lifted formula $(\psi,\ell)$ (Line~\ref{alg_line: liftedCompile: preprocessing}). The lifted formula is then passed to the sub-procedure 
$\compile$, which traverses recursively over the structure of 
$\psi$. If $\psi$ is a single variable $x$, the procedure returns 
$\ell(x)$, which is a read-once formula (Line~\ref{alg_line: liftedCompile: trivial}).
If $\psi$ is a disjunction of independent sub-formulas $\psi_1, \ldots , \psi_n$, it first constructs  the d-trees for $(\psi_1,\ell|_{\psi_1}), \ldots , (\psi_n,\ell|_{\psi_n})$, where $\ell|_{\psi_i}$ is the restriction of the domain of  $\ell$ onto the variables of $\psi_i$. Then, it combines the d-trees using $\oplus$ (Line~\ref{alg_line: liftedCompile: ind_or}). 
For conjunctions of independent subformulas, the procedure similarly combines the compiled trees using $\odot$ (Line~\ref{alg_line: liftedCompile: ind_and}). Otherwise, it performs Shannon expansion on a variable $y$: it compiles $\lift(\psi[y := 0], \ell|{\psi[y := 0]})$ and $\lift(\psi[y := 1], \ell|{\psi[y := 1]})$ into d-trees $T_0$ and $T_1$, then combines them with a $\sqcup_{\ell(y)}$ gate, where $T_0$ and $T_1$ become the 0- and 1-branches, respectively (Lines~\ref{alg_line: liftedCompile: shannon_start}--\ref{alg_line: liftedCompile: shannon_end}). The lifting procedure is applied to $\psi[y:=1]$
and $\psi[y:=0]$, since the substitution of $y$ by a constant can create new symmetries.

\nop{\begin{example}
\label{example: lifted_compile}
We apply Algorithm~\ref{alg:lifted_compile} to the lineage of $Q_1$ given in Table~\ref{tab:lineages_lifted} (top left). The saturated lifted lineage $(\varphi, \ell)$ constructed in Line~\ref{alg_line: liftedCompile: preprocessing} is given in Table~\ref{tab:lineages_lifted} (top right). Since the variable $d_{1,2}$ appears in all clauses of $\varphi$, the procedure 
factors out $d_{1,2}$ 
using independent-and: $d_{1,2} \odot \big((a_{1,3} \wedge m_3) \vee 
(a_{1,3} \wedge m_2) \vee (a_{2} \wedge m_3)\big)$ and proceeds recursively on both sides of $\odot$.
Since $\ell$ maps $d_{1,2}$ to the read-once formula $d_1 \vee d_2$,
the procedure returns the latter formula for $d_{1,2}$. The formula on the right side of $\odot$ cannot be decomposed using $\odot$ or $\oplus$, thus the procedure performs Shannon expansion on variable $a_{1,3}$, since it is one of the most frequent variables. It then performs lifting on each of the new formulas resulted from the Shannon expansion. 
This process continues until we obtain the final d-tree presented in Figure~\ref{fig:Lifted_Dtree_Example}.
\end{example}}

\begin{example}
\label{example: lifted_compile}
We apply Algorithm~\ref{alg:lifted_compile} to the lineage of $Q_1$ in Table~\ref{tab:lineages_lifted} (top left). The saturated lifted lineage $(\varphi, \ell)$ constructed in Line~\ref{alg_line: liftedCompile: preprocessing} is given in Table~\ref{tab:lineages_lifted} (top right). Since the variable $d_{1,2}$ appears in all clauses of $\varphi$, we 
factor out $d_{1,2}$ from all clauses and decompose the formula using independent-and: $d_{1,2} \odot \big((a_{1,3} \wedge m_3) \vee 
(a_{1,3} \wedge m_2) \vee (a_{2} \wedge m_3)\big)$ and proceeds recursively on both sides of $\odot$.
Since $\ell$ maps $d_{1,2}$ to the read-once formula $d_1 \vee d_2$,
we return the latter formula for $d_{1,2}$. The formula on the right side of $\odot$ can only be decomposed using Shannon expansion on one of its variables. We choose $a_{1,3}$, one of its most frequent variables. We then perform lifting on each of the new formulas resulted from the Shannon expansion. This process continues until we obtain the d-tree from Figure~\ref{fig:Lifted_Dtree_Example} (ignore for now the numeric values appearing next to nodes).
\end{example}

Examples~\ref{example:lifted lineage} and \ref{example: lifted_compile} demonstrate two benefits of our lifting approach over compilation without lifting: faster compilation time and smaller resulting d-tree. The compilation is more efficient since the lifted lineage can be small and Shannon expansion can be applied to the new fresh variables, which can represent large formulas.

\input{Figures/arithmetic_circuit_example1}
\nop{Algorithm~\ref{alg:lifted_compile} describes how to utilize the lifting procedure to efficiently compile lineage to a d-tree. The compilation process begins with a pre-processing step (Line~\ref{alg_line: liftedCompile: preprocessing}) of lifting the lineage before the incremental compilation. Then, if the lineage is not trivial (line~\ref{alg_line: liftedCompile: trivial}), it tries to decompose the lifted lineage using {\em independent Or} (line~\ref{alg_line: liftedCompile: ind_or}), {\em independent And} (line~\ref{alg_line: liftedCompile: ind_and}) and Shannon expansion (lines~\ref{alg_line: liftedCompile: shannon_start}-\ref{alg_line: liftedCompile: shannon_end}) gates. After each Shannon expansion, the lifting procedure is reapplied (line~\ref{alg_line: liftedCompile: shannon_end}) to uncover potential new symmetries.

\begin{example}
\label{example: lifted_compile}
We analyze Algorithm~\ref{alg:lifted_compile} on the lineage of $Q_1$ from Table~\ref{tab:lineages_lifted}. The initial lifting of line~\ref{alg_line: liftedCompile: preprocessing} is described in Example~\ref{example:lifted lineage}. Then, the algorithm begins the compilation by decomposing the resulting lineage with the $\odot$ gate, as $\textcolor{red}{(d_1\vee d_2)}$ appears in all clauses. The algorithm proceeds recursively to each of the resulting formulas. For $\textcolor{red}{(d_1\vee d_2)}$ it returns it as a d-tree since it is a lifted variable. The remaining formula cannot be decomposed using $\odot$ or $\oplus$ gates, thus the algorithm performs a Shannon expansion on the most frequent variable, $\textcolor{green}{(a_1\vee a_3)}$. It then performs lifting on each of the new formulas resulted from the Shannon Expansion. The formula $(m_2) \vee (m_3)$ resulting from the True branch of the Shannon Expansion is lifted to $\textcolor{orange}{(m_2 \vee m_3)}$. This process continues until obtaining the final d-tree presented in Figure~\ref{fig:Lifted_Dtree_Example}.
\end{example}

Examples~\ref{example:lifted lineage} and \ref{example: lifted_compile} demonstrate two benefits of our approach. Firstly, the lifted lineage is significantly shorter making operations on it more efficient. Moreover, by compiling the lifted lineage we can perform compilation operations that were not possible before, such as performing Shannon expansion on formulas representing multiple variables instead of one. This demonstrates that the compilation results in a smaller circuit with a faster compilation time.

\begin{algorithm}[h]
\caption{LiftedCompile}
\label{alg:lifted_compile}
\begin{algorithmic}[1]
\Require DNF lineage $\varphi$ for query $Q$ on database $D$
\Ensure $d-tree$ for $\varphi$
\State $\psi \gets \text{Lift}(\varphi)$ \label{alg_line: liftedCompile: preprocessing}
\State \Return \_Compile($\psi$)
\item[]
\Procedure{\_Compile}{$\psi$}
\If{$\psi$ is a variable or a lifted variable}
    \Return $\psi$ \label{alg_line: liftedCompile: trivial}
\EndIf
\SWITCH{ $\psi$}
     \State \TAB \CASE{ $\psi_1 \vee \dots \vee \psi_n$ for independent $\psi_1, \dots, \psi_n$} \label{alg_line: liftedCompile: ind_or}
         \State \TAB \TAB $T_\varphi \gets \bigoplus \_Compile(\psi_i)$ 
    \State \TAB \CASE{ $\psi_1 \wedge \dots \wedge \psi_n$ for independent $\psi_1, \dots, \psi_n$} \label{alg_line: liftedCompile: ind_and}
    \State \TAB \TAB $T_\varphi \gets \bigodot\_Compile(\psi_i)$
    \State \TAB \DEFAULT
        \State \TAB \TAB Pick a variable $y$ in $\psi$ \label{alg_line: liftedCompile: shannon_start}
        \State \TAB \TAB $T_\varphi \gets (y \odot\_Compile( \text{Lift}(\psi[y:=1]))) \oplus\\ (\neg y \odot \_Compile(\text{Lift}(\psi[y:=0]))$ \label{alg_line: liftedCompile: shannon_end}

\Return $T_{\varphi}$
\EndProcedure
\end{algorithmic}
\end{algorithm}

}

\subsection{Gradient-Based Computation}
\label{subsubsect: arithmetic}

We propose a novel approach for Banzhaf and Shaply computation, yielding the following result:
\begin{theorem}
\label{theorem: linear Banzhaf and Shapley}
For a d-tree $T$, the Banzhaf and Shapley values of all variables in $T$ can be computed in time $\mathcal{O}(|T|)$ and respectively $\mathcal{O}(|T|\cdot |\vars(T)|^2)$.
\nop{Given a d-tree $T$ we can compute the Banzhaf (respectively Shapley) values with respect to all variables of $T$ in $\mathcal{O}(|T|)$ (respectively $\mathcal{O}(|T|\cdot |\vars(T)|)$).}  
\end{theorem}
We will interpret a d-tree as a function over random variables
and reduces the computation of Banzhaf and Shapley values 
to the computation of the gradient of this function.\footnote{There are two seminal works related to our gradient-based approach: The Baur-Strassen result on the efficient computation of a multivariate function defined by an arithmetic circuit and its gradient vector~\cite{BAUR1983317}, and the Darwiche result on efficient inference in Bayesian networks \cite{ADiffrentialApproachToInferenceInBayesianNetworks}. To the best of our knowledge, no prior work uses gradient-based computation for Banzhaf/Shapley values.} 
We first show the computation for Banzhaf and then extend it to Shapley values. 

\paragraph*{Probabilistic Interpretation of D-Trees} We associate with each Boolean variable $x$ occurring in a d-tree leaf, a variable $p_x$ taking values in $[0,1]$. A value assigned to $p_x$ is the probability for $x=1$. $Pr[T]$ is then a function over these variables, defined recursively using the definitions in Table~\ref{tab:probability_and_gradients_for_gates}. Its value is the probability that the Boolean formula represented by the d-tree 
$T$ evaluates to 1.  

\paragraph*{From $Pr[T]$ to Banzhaf values} 
We next connect the Banzhaf value of a variable $x$ to the
partial derivative of $Pr[T]$ w.r.t. $x$. 

\begin{proposition}
\label{Prop: Banzhaf_is_Gradient}
Given a d-tree $T$ and a variable $x\in \vars(T)$:
\begin{align*}
    \banz(T,x) =& 2^{|\vars(T)| - 1} \cdot \left(\frac{\partial Pr[T]}{\partial p_x}\left(\vec{\frac{1}{2}}\right)\right) 
\end{align*}
where $p_x$ is the Boolean variable corresponding to $x$, $\frac{\partial Pr[T]}{\partial p_x}$ is the partial derivative of $Pr[T]$ with respect to $p_x$, and $\vec{\frac{1}{2}}$ is the column vector  $(\frac{1}{2},\dots , \frac{1}{2})$ of length $|\vars(T)|$.  
\end{proposition}

\begin{proof}
We can rewrite $Pr[T]$ as
{\small \begin{align*}
Pr[T] =&\ \sum_{S\subseteq \vars(T)} Pr[S]\cdot T[S] = \sum_{S\subseteq \vars(T)} \left( \prod_{z\in S} p_z \cdot \prod_{y\notin S} (1 - p_y) \right) T[S],
\end{align*} \normalsize}
where $T[S]$ is the Boolean value of the formula represented by $T$ under the valuation $S$.
Given a variable $x\in \vars(T)$, we separately consider for each valuation $S$ the cases 
where $x$ appears in $S$ and where it does not. For the cases where $x$ appears in we obtain:
{\small 
\begin{align*}
Pr[T] = \sum_{S\subseteq \vars(T)\setminus \{x\}} \Big( & \prod_{z\in S} p_z \cdot \prod_{y\notin S\cup \{x\}} (1 - p_y) \cdot (p_x\cdot T[S\cup\{x\}])\Big)
\end{align*} \normalsize
}

For the cases where $x$ does not appear in, we obtain:
{\small 
\begin{align*}
Pr[T] = \sum_{S\subseteq \vars(T)\setminus \{x\}} \Big( & \prod_{z\in S} p_z \cdot \prod_{y\notin S\cup \{x\}} (1 - p_y) \cdot \big((1 - p_x)\cdot T[S]\big)\Big)
\end{align*} \normalsize
}

And over all cases, we obtain:
{\small 
\begin{align*}
Pr[T] = \sum_{S\subseteq \vars(T)\setminus \{x\}} \Big( & \prod_{z\in S} p_z \cdot \prod_{y\notin S\cup \{x\}} (1 - p_y)  \notag 
\\ & \cdot \big(p_x \cdot T[S\cup\{x\}] + (1-p_x) \cdot T[S] \big)\Big)
\end{align*} \normalsize
}

By taking the partial derivative of the above function 
with respect to $p_x$ at the point $\vec{\frac{1}{2}}$, we get:
{\small \begin{align*}
\frac{\partial Pr[T]}{\partial p_x}(\vec{\frac{1}{2}}) =& \sum_{S\subseteq \vars(T) \setminus \{x\}} \left(\frac{1}{2}\right)^{|\vars(T)|-1} \cdot (T[S\cup \{x\}] - T[S]) \notag \\ 
=&  \left(\frac{1}{2}\right)^{|\vars(T)|-1}  \cdot \banz(T,x) \qedhere
\end{align*}\normalsize}
\end{proof}

Intuitively, evaluating the derivative at $\vec{\frac{1}{2}}$ gives each subset equal likelihood. Thus, the partial derivative increases the weight of subsets containing the variable and decreases it for those that do not.
\paragraph*{Efficient Computation} 
Using Prop.~\ref{Prop: Banzhaf_is_Gradient}, we can express the probability of a d-tree as a function of the probabilities of its leaf variables and derive its partial derivatives with respect to each variable by recursively applying the chain rule. Prop.~\ref{Prop: Banzhaf_is_Gradient} shows that these derivatives coincide with the Banzhaf values up to normalization.
Table~\ref{tab:probability_and_gradients_for_gates} gives rules for computing the probability of a d-tree $T$ from its sub-trees, following standard probability theory, as well as the partial derivatives $\partial Pr[T] / \partial Pr[T_i]$ for each sub-tree $T_i$, which together enable gradient computation with respect to the leaf variables.

The procedure $\gradientBanzhaf$ in Algorithm~\ref{alg:GradientBanzhaf} implements this as follows. Given a d-tree $T_\varphi$, the algorithm gradually computes two quantities for each tree node $v$: $v.p$ which is the probability of $T_v$, and $v.g$ which is $\frac{\partial Pr(T_{\varphi})}{\partial Pr(T_v)}(\vec{\frac{1}{2}})$, where $T_v$ is the d-tree rooted at $v$. First, $v.p$ values are initialized to $\frac{1}{2}$, as our goal is to compute derivatives as in  Table~\ref{tab:probability_and_gradients_for_gates}, and evaluate them at $\vec{\frac{1}{2}}$, as in Prop.~\ref{Prop: Banzhaf_is_Gradient}; the $v.g$ value of the root is initialized to $2^{\mid \vars(T_\varphi)\mid -1}$ to account for normalization. Probabilities are computed bottom-up  using the  expressions in the second column of Table~\ref{tab:probability_and_gradients_for_gates} (Lines~\ref{alg_line: GradientBanzhaf: compute probabilities - start} - \ref{alg_line: GradientBanzhaf: compute probabilities - end}). Partial derivatives are computed in a top-down fashion, using the expressions in the third column of Table~\ref{tab:probability_and_gradients_for_gates} and the chain rule (Lines~\ref{alg_line: GradientBanzhaf: grad backpropagation start}-~\ref{alg_line: GradientBanzhaf: grad backpropagation end}). For a variable $v$, the sum of partial derivatives for different leaf nodes representing $v$ is equal to $\frac{\partial Pr[T]}{p_v} (\vec{\frac{1}{2}})$ which equals its Banzhaf value by Prop.~\ref{Prop: Banzhaf_is_Gradient}.



\begin{table}[]
\caption{Equations defining the probability and the partial derivatives for different gates. $T$ is the d-tree rooted at the gate and the $T_i$'s are its child sub-trees.}
    \centering
    \scriptsize
    \renewcommand{\arraystretch}{1.5}
    \begin{tabular}{|c|l|l|}
        \hline
        Gate & Probability Expression $Pr[T]$ & Partial Derivative \\ \hline
        $\oplus$ & $1 - \underset{i\in[n]}\prod (1 - Pr[T_i])$ 
        & $\frac{\partial(Pr[T])}{\partial(Pr[T_i])} = \underset{{j\in[n]\setminus \{i\}}}{\prod} (1 - Pr[T_j]) = \frac{1 - Pr[T]}{1 - Pr[T_i]}$ \\ \hline
         $\odot$ & $\prod_{i\in[n]} Pr[T_i]$ 
        & $\frac{\partial(Pr[T])}{\partial(Pr[T_i])} =\prod_{j\in[n]\setminus \{i\}} Pr[T_j] =  \frac{Pr[T]}{Pr[T_i]}$ \\ \hline
\multirow{3}{*}{$\sqcup_f$}  
         & \multirow{3}{*}{\shortstack{$Pr[f] \cdot Pr[T_1]$ \\  
                            $+ (1 - Pr[f])\cdot Pr[T_0]$}}    
        & $\frac{\partial(Pr[T])}{\partial(Pr[T_1])} = Pr[f]$ \\ 
        & & $\frac{\partial(Pr[T])}{\partial(Pr[T_0])} = 1 - Pr[f]$ \\ 
        & & $\frac{\partial(Pr[T])}{\partial(Pr[f])} = Pr[T_1] - Pr[T_0]$ \\ \hline

    \end{tabular}
    \label{tab:probability_and_gradients_for_gates}

\end{table}

\begin{example}
   Reconsider the d-tree from Figure~\ref{fig:Lifted_Dtree_Example}, where its nodes are now annotated with $p$ and $g$ values. Consider the $\oplus$ node $v$ that is the root's left child. Its $p$ value is computed using the equation for $\oplus$ from Table~\ref{tab:probability_and_gradients_for_gates} and the probability of its sub-trees $d_1$ and $d_2$: $v.p = 1- (1-d_1.p)\cdot (1- d_2.p) = 1 - 0.5\cdot 0.5 = 0.75$. We then have $v.g = v.parent.g\cdot \frac{\partial v.parent.p}{\partial v.p} = v.parent.g \cdot \frac{v.parent.p}{v.p} = 64 \cdot \frac{0.46875}{0.75} = 40$.   
%
\end{example}

\paragraph*{Shapley values} We next link Shapley values to gradients: 
\begin{align*}
\shap(T, x) =& 2^{|\vars(T)| - 1} \cdot \sum_{k \in [|\vars(T)|]} C_{k-1} \cdot \left(\frac{\partial Pr_k[T]}{\partial p_x}\left(\vec{\frac{1}{2}}\right)\right),  \text{where}\end{align*} 
\nop{where the $k$-probability $Pr_k[T] = \sum_{S\subseteq \vars(T), |S|=k} Pr[S]\cdot T[S]$ is the probability of $T$ computed only over the valuations of size $k$.}
\begin{align*} 
Pr_k[T] =& \sum_{S\subseteq vars(T), \mid S \mid = k} \left(\frac{1}{2}\right)^{|\vars(T)| - k}\sum_{S'\subseteq S} T[S']\cdot \prod_{y\in S'}p_y \prod_{z\in S\setminus S'}\left(\frac{1}{2}-p_z\right)
\end{align*}
To compute Shapley values, the equations in Table \ref{tab:probability_and_gradients_for_gates} are adapted to compute $Pr_k[T]$ instead of $Pr[T]$, and $\frac{\partial (\sum_k C_k \cdot \Pr_k[T])}{\partial Pr_j[T_i]}(\vec{\frac{1}{2}})$ for each $j \in [n]$ instead of $\frac{\partial Pr[T]}{\partial Pr[T_i]}(\vec{\frac{1}{2}})$. Similar equations to those shown in Table~\ref{tab:probability_and_gradients_for_gates} are obtained for these partial derivatives. Lines~\ref{alg_line: GradientBanzhaf: compute probabilities - start} - \ref{alg_line: GradientBanzhaf: compute probabilities - end} are then changed to compute a vector of all $k$-probabilities for each node. Line~\ref{alg_line: GradientBanzhaf: initialize_grad_root} is modified to initialize the root's vector of derivatives to be the Shapley coefficients, and Lines~\ref{alg_line: GradientBanzhaf: grad backpropagation start} - \ref{alg_line: GradientBanzhaf: grad backpropagation end} are modified to compute the partial derivative for each value of $j$. 


\paragraph*{Algorithms} The \lexaban  (\lexashap) algorithm for computing Banzhaf (Shapley) values uses Algorithm~\ref{alg:lifted_compile} to compile the query lineage to a d-tree $T$ and then Algorithm~\ref{alg:GradientBanzhaf} (its adaptation to Shapley values) to $T$ to compute Banzhaf (Shapley) values for all variables in $T$, or equivalently for all tuples contributing to the lineage. 

\paragraph*{Complexity} For each node $v$, $Pr[v]$ is computable in time linear in its number of children, so 
$\mathcal{O}(|T|)$ for all nodes. For Banzhaf, the gradient w.r.t. each leaf node can be computed in $\mathcal{O}(1)$ following Table~\ref{tab:probability_and_gradients_for_gates}, so in $\mathcal{O}(|T|)$ total time. For Shapley, we need to compute $|\vars(T_v)|$ expressions and partial derivatives at each node $v$ and each may be computed in $\mathcal{O}(|\vars(T)|)$ time, resulting in $\mathcal{O}(|\vars(T)|^2)$ per node and $\mathcal{O}(|T|\cdot |\vars(T)|^2)$ in total.

\nop{For Banzhaf computation, both the probabilities and partial derivatives are computed in $\mathcal{O}(1)$ time per d-tree node. For Shapley computation, the time complexity is $O(|\vars(n)|)$ per d-tree node $n$ due to the need to compute a vector of probabilities rather than a single probability value.}   


\nop{
We first present an extension of relational databases named ``Symbolic Probabilistic Databases``. Next, we identify a connection between the Banzhaf and Shapley values to derivatives of the probability of a query under that model. Lastly, we use these insights to derive back-propagation algorithms that compute simultaneously the attribution values of all variables yielding the improved complexity.

\paragraph{Symbolic Probabilistic Databases and formulas}
{\em Tuple-independent probabilistic database} (or simply, probabilistic databases) are a common extension of relational databases, where for a database $D$, each tuple $x$ in $D$ is assigned a numeric probability $p_x$ of being in the database \cite{suciu2022probabilistic}.
We consider a {\em symbolic probabilistic database}, where for a database $D^s = (D_n,D_x)$ each tuple $x\in D_n$ is annotated with a symbolic probability token $p_x$ of being in the database. This representation captures an uninstantiated probabilistic database, allowing for probabilistic analysis at the symbolic level. Given a database $D = (D_n,D_x)$, we refer to the process of assigning symbolic probability tokens to the tuples in $D_n$ as its {\em symbolic probabilistic extension}.  
For a sub-database $D'\subseteq D_n$, we define the probability of $D'$ to be:
 \begin{align}
    Pr[D'] \overset{def}{=} \prod_{x\in D'}{p_x} \cdot \prod_{y\in D_n\setminus D'}{(1-p_y)}
\end{align} 
Here, $Pr[D']$ is a function  $[0,1]^{|P|} \to \mathbb{R}$ where $P = \{p_x| x\in D_n\}$ is the set of symbolic probabilities. Next, we overload notation and define the probability of a query. Given a symbolic probabilistic database $D$ and a boolean query $Q$, the probability of $Q$ over $D$ is:
 \begin{align*}
    Pr[Q(D)] \overset{def}{=} \sum_{D'\in D_n}{Pr[D']\cdot Q(D'\cup D_x)}
\end{align*} 

Similarly, we define the k-probability of $Q$ over $D$ to be:
 \begin{align*}
    Pr_k[Q(D)] \overset{def}{=} \sum_{D'\in D_n, |D'| = k}{Pr[D']\cdot Q(D'\cup D_x)}
\end{align*} 
$Pr[Q(D)]$ and  $Pr_k[Q(D)]$ are also functions $[0,1]^{|P|} \to \mathbb{R}$.

\nop{Next, we define the expected value of a query. Given a symbolic probabilistic database $D$, and a query $Q$, the expectation of $Q$ over $D$ is:
\begin{align}
    E[Q(D)] \overset{def}{=} \sum_{D'\in D_n}{Pr[D']\cdot Q(D'\cup D_x)}
\end{align}
Similarly, we define the k-expected value of $Q$ over $D$ as:
\begin{align}
    E_k[Q(D)] \overset{def}{=} \sum_{D'\in D_n, |D'| = k}{Pr[D']\cdot Q(D'\cup D_x)}
\end{align}

$E[Q(D)]$ and  $E_k[Q(D)]$ are also functions $[0,1]^{|P|} \to \mathbb{R}$.} 

\nop{\begin{example}
\daniel{Not sure that this example is useful in its current form. To discuss.}
    In the symbolic database extension $D^s$ of the database $D$ from Figure~\ref{fig:running_example}, each tuple $f$ annotated with a variable $v_i$ is assigned a probability of $p_{v_i}$. The probability of the sub-database $\{a_1,a_2,m_1,m_2,m_3\}$ is
    $p_{a_1}\cdot p_{a_2} \cdot p_{m_1}\cdot p_{m_2}\cdot p_{m_3} \cdot (1 - p_{a_3}) \cdot (1 - p_{a_4})$.
    Consider the query $Q_1$ from Figure~\ref{fig:running_example}. The probability of $Q_1$ on $D^s$ is:
    \begin{align*}
        Pr[Q_1(D^s)] = \sum_{D'\subseteq D: Q(D')=1}{Pr[D']}
    \end{align*}
\end{example}}


Similarly, given a Boolean formula $\varphi$, we define the symbolic probabilistic extension of $\varphi$ where we annotate each variable $v$ in $\varphi$ with a probabilistic token $p_v$. The probability of an assignment is defined analogously to the probability of a sub-database, and the probability of $\varphi$ is defined analogously to that of a query.

\paragraph{Attribution values coincide with the query probability gradient}
We provide an observation that the Banzhaf and Shapley values coincides with the partial derivative of the probability of a query in a symbolic probabilistic extension of a database. We denote by $\vec{\frac{1}{2}}$ a vector $(\frac{1}{2}, \dots ,\frac{1}{2})$ of the appropriate size.

\begin{theorem}
\label{theorem:GradientEquivalence}
Let $Q$ be a UCQ, $D = (D_n,D_x)$ be a database and $D^s$ its symbolic probabilistic extension. Let $f$ be a fact in $D_n$. Then: 

\begin{align}
    Banzhaf(Q,D,f) =& 2^{|D_n| - 1} \cdot \big(\frac{\partial Pr[Q(D^s)]}{\partial p_f}(\vec{\frac{1}{2}})\big)\\
   Shapley(Q,D, f) =& 2^{|D_n| - 1} \cdot \sum_{k \in [|D_n|]} C_k \cdot \big(\frac{\partial Pr_k[Q(D^s)]}{\partial p_f}(\vec{\frac{1}{2}})\big)
\end{align}

Where $C_k = \frac{|k!|\cdot|D_n - k - 1|!}{|D_n|!}$ is the Shapley coefficient for size $k$.
\end{theorem}

We provide a proof for the Banzhaf value, while the proof for the Shapley value is omitted for lack of space and is included in the full version of the paper.

\begin{proof}


 \begin{align}
Pr[Q(D^s)] =&\ \sum_{S\subseteq D_n} Pr[S]\cdot Q(S\cup D_x) =\notag\\
=& \sum_{S\subseteq D_n} \left( \prod_{x\in S} p_x \cdot \prod_{y\notin S} (1 - p_y) \right) Q(S\cup D_x)
\label{eq:mean_query_value}
\end{align} 

For each variable $f\in D_n$ we can separate the sub-databases based on whether $f$ appears in them or not, and get: 

 \begin{align}
Pr[Q(D^s)] =& \sum_{S\subseteq D_n\setminus \{f\}} \left( \prod_{x\in S} p_x \cdot \prod_{y\notin S\cup \{f\}} (1 - p_y) \right) \notag\\ \cdot& \left(p_f \cdot Q(S\cup D_x\cup \{f\}) + (1-p_f) \cdot Q(S\cup D_x) \right)
\label{eq:mean_query_value_diffrentiated}
\end{align} 

By taking the partial derivative of the function in Equation~\ref{eq:mean_query_value_diffrentiated} at the point $\vec{\frac{1}{2}}$ we get:

 \begin{align}
\frac{\partial Pr[Q(D^s)]}{\partial p_f} =& \sum_{S\subseteq D_n \setminus \{f\}} \frac{1}{2}^{|D_n|-1} \cdot (Q(S\cup D_x\cup \{f\}) - Q(S\cup D_x)) \notag \\ 
=& \frac{1}{2}^{|D_n|-1}\cdot Banzhaf(Q,D,f)
\label{eq:grad_mean_val}
\end{align}

\end{proof}

Let $\varphi$ be the lineage of $Q$ on $D$, where the symbolic probabilistic extension of $\varphi$, $\varphi^s$, is achieved by assigning each variable with corresponding symbolic probability token as the fact in $D^s$ it represents. It is easy to see the equivalence of $Pr[Q(D^s)]$ with $Pr[\varphi^s]$. Thus, we can compute the Banzhaf values by constructing a circuit that allows us, given a DNF Boolean function $\varphi$ and the probabilities of its variables, to compute $Pr[\varphi]$ and its gradient.
}
\nop{\paragraph{Constructing an arithmetic circuit for the probability}
We observe that the lineage can be compiled to such arithmetic circuit in a fashion similar to that of a d-tree. This can be done by taking an existing d-tree and replacing its gates with arithmetic gates, or by compiling the lineage directly into an arithmetic circuit that computes the lineage's probability. 
Due to the simplicity of the arithmetic gates, we can easily compute the gradients using back-propagation. 
It is a known result that computing a polynomial $g(x_1, x_2, ..., x_n)$ is essentially equivalent to computing $g$ and its gradient  \cite{shpilka2010arithmetic}\cite{BAUR1983317}. For clarity, we show how to calculate the partial derivatives directly.}

\nop{
\paragraph{Computation using d-trees} We show here the computation of the probability and its derivatives through d-trees. As proved before, these would produce the Banzhaf and Shapley values. 
\begin{proposition}
\label{theorem: gradient_computation_in_linear_time}
Given a constructed d-tree $T_n$ representing a function $\varphi$ with a symbolic probabilistic extension $\varphi^s$, we can compute in $\mathcal{O}(|T_n|)$ time the probability of $\varphi$ and the gradient of the probability according to each variable's symbolic probability. 
\end{proposition}

\begin{proof}
We prove Proposition~\ref{theorem: gradient_computation_in_linear_time} by examining the different d-tree gates and demonstrating the computation directly. Table~\ref{tab:probability_and_gradients_for_gates} presents the resulting equations. Back-propagation allows us to efficiently compute gradients by applying the chain rule ($\frac{\partial{f(g(x))}}{\partial x} = \frac{\partial f}{\partial g}\cdot \frac{\partial g}{\partial x}$) for derivatives repeatedly, propagating derivatives from output nodes back to input nodes. Thus, this approach alongside the equations in Table~\ref{tab:probability_and_gradients_for_gates} enable us to compute the Banzhaf value with $O(1)$ calculations per node.
\end{proof}

\nop{Consider the N-ary independent OR gate.
Let $\varphi = \varphi_1 \vee \varphi_2 \vee ... \vee \varphi_n$ be a monotone boolean formula such that \\$\forall i,j\in [n].\vars(\varphi_i)\cap \vars(\varphi_j) = \emptyset$

We can express the probability that $\varphi$ is satisfied, as a function of the probabilities of its sub-formulas:

\begin{align}
Pr[\varphi] =& \,1 - \prod_{i\in[n]} (1-Pr[\varphi_i]))
\label{eq:prob_independent_or}
\end{align}

Hence, the partial derivative, or equivalently, the portion of assignments where each sub-formula is ``critical`` (its value determines the value of $\varphi$), is achieved by:
\begin{align}
\frac{\partial(Pr[\varphi])}{\partial( Pr[\varphi_i])} =& \prod_{j\in[n]\setminus \{i\}} (1-Pr[\varphi_j]) = (1 - Pr[\varphi])/ (1-Pr[\varphi_i])
\label{eq:grad_independent_or}
\end{align}


Similarly, for the N-ary independent And gate:

\begin{align}
Pr[\varphi] =& \prod_{i\in[n]} Pr[\varphi_i]
\label{eq:prob_independent_and}
\end{align}

\begin{align}
\frac{\partial(Pr[\varphi])}{\partial(Pr[\varphi_i])} =& \prod_{j\in[n]\setminus \{i\}} Pr[\varphi_j] = Pr[\varphi]/ Pr[\varphi_i]
\label{eq:grad_independent_and}
\end{align}

Lastly, let $\varphi = f \wedge \varphi_1 \vee \neg f \wedge \varphi_2$ be a monotone Boolean formula such that $f \notin \varphi_1,\varphi_2$. We get that:

\begin{align}
Pr[\varphi] =& Pr[f] \cdot Pr[\varphi_1] + (1-Pr[f])\cdot Pr[\varphi_2]
\label{eq:prob_exclusive_or}
\end{align}

\begin{align}
\frac{\partial(Pr[\varphi])}{\partial(Pr[\varphi_1])} =& Pr[f] \\ 
\frac{\partial(Pr[\varphi])}{\partial(Pr[\varphi_2])} =& 
1 - Pr[f]\\
\frac{\partial(Pr[\varphi])}{\partial(Pr[f])} =& Pr[\varphi_1] - Pr[\varphi_2]
\label{eq:grad_exclusive_or}
\end{align}
}

\begin{table}[h]
\caption{Probability and derivative formulas for different gates. Here, $\varphi$ represents the formula derived by the gate, and $\varphi_i$ represent the subformula inputs.}

    \centering
    \scriptsize
    \renewcommand{\arraystretch}{1.5}
    \begin{tabular}{|c|l|l|}
        \hline
        Gate & Probability Expression & Partial Derivative \\ \hline
        $\oplus$ & $Pr[\varphi] = 1 - \prod_{i\in[n]} (1 - Pr[\varphi_i])$ 
        & $\frac{\partial(Pr[\varphi])}{\partial( Pr[\varphi_i])} = \frac{(1 - Pr[\varphi])}{(1 - Pr[\varphi_i])}$ \\ \hline
         $\odot$ & $Pr[\varphi] = \prod_{i\in[n]} Pr[\varphi_i]$ 
        & $\frac{\partial(Pr[\varphi])}{\partial(Pr[\varphi_i])} = \frac{Pr[\varphi]}{Pr[\varphi_i]}$ \\ \hline
        $\sqcup_f$ & $Pr[\varphi] = Pr[f] \cdot Pr[\varphi_1] + (1-Pr[f])\cdot Pr[\varphi_2]$ 
        & $\frac{\partial(Pr[\varphi])}{\partial(Pr[\varphi_1])} = Pr[f]$ \\ 
        & & $\frac{\partial(Pr[\varphi])}{\partial(Pr[\varphi_2])} = 1 - Pr[f]$ \\ 
        & & $\frac{\partial(Pr[\varphi])}{\partial(Pr[f])} = Pr[\varphi_1] - Pr[\varphi_2]$ \\ \hline
    \end{tabular}
    \label{tab:probability_and_gradients_for_gates}
\end{table}


Algorithm~\ref{alg:GradientBanzhaf} presents the algorithm for computing Banzhaf values using the gradient approach. Given a d-tree for the lineage $T_\varphi$, Lines~\ref{alg_line: GradientBanzhaf: compute probabilities - start} - \ref{alg_line: GradientBanzhaf: compute probabilities - end} initialize the leaves probabilities to $\frac{1}{2}$, and then traverses the d-tree in a bottom up order and computes the probabilities at every node of $T_\varphi$ according to Table~\ref{tab:probability_and_gradients_for_gates}. Then, Line~\ref{alg_line: GradientBanzhaf: initialize_grad_root} initializes the gradient value at the root to be $2^{|\vars(\varphi)|-1}$. Finally, Lines~\ref{alg_line: GradientBanzhaf: grad backpropagation start}-\ref{alg_line: GradientBanzhaf: grad backpropagation end} perform a top-down traversal of $T_\varphi$, computing the partial derivative at each node according to Table~\ref{tab:probability_and_gradients_for_gates}. The algorithm returns the gradient value of the variable leaves as the Banzhaf values.

To adapt \lexaban to Shapley values computation (\lexashap), we proceed as follows. \daniel{In Line...}

The algorithm for computing Shapley values \lexashap is with minor changes resulting from Theorem~\ref{theorem:GradientEquivalence}. Instead of computing the probability at each node, we compute the k-probabilities. We then initialize the $g$ property of the root to be the vector of Shapley coefficients. Finally, instead of propagating the $g$ values with equations representing $\frac{\partial Pr[\varphi]}{\partial Pr[\varphi_i]}$ we compute them as representing $\sum_k\frac{\partial Pr_k[\varphi]}{\partial Pr_{j}[\varphi_i]}$. These values are propagated using efficient equations to the variable leaves with $O(\vars(D_n))$ operations per d-tree node $D_n$ leading to the complexity in Proposition~\ref{theorem: linear Banzhaf and Shapley}.
}

\nop{\begin{algorithm}[t]
\caption{GradientBanzhaf}
\label{alg:GradientBanzhaf}
\begin{algorithmic}[1]
\Require D-tree $T_{\varphi}$
\Ensure Banzhaf values for $\varphi$
\For{each node $n$ in the tree in bottom-up order} 
\State Compute probability of $n$ using equations \ref{eq:prob_independent_or}, \ref{eq:prob_independent_and}, \ref{eq:prob_exclusive_or} \label{alg_line: GradientBanzhaf: compute probabilities}
\EndFor
\State Initialize the gradient of the root to be $2^{|\text{vars}(\varphi)|-1}$ \label{alg_line: GradientBanzhaf: initialize_grad_root}
\For{each node $n$ in the tree in top-down order} \label{alg_line: GradientBanzhaf: grad backpropagation start}
    \State Compute gradient of $n$ using equations \ref{eq:grad_independent_or}, \ref{eq:grad_independent_and}, \ref{eq:grad_exclusive_or} \label{alg_line: GradientBanzhaf: grad backpropagation end}
\EndFor
\State \Return Gradient of the variable leaves as the Banzhaf values
\end{algorithmic}
\end{algorithm}}

\nop{\begin{corollary}
There exists a linear-time reduction from computing the Banzhaf value for all database tuples to model counting.
\label{Corrolary: BanzIsModelCount}
\end{corollary}
Corollary \ref{Corrolary: BanzIsModelCount} provides a stronger result than the polynomial-time reduction derived from \cite{WhenIsShapleyValueAMatterOfCoutning}. For Shapley value, a similar linear-time reduction could be made to the problem of fixed-size model counting described in \cite{FromShapleyValueToModelCountingAndBack}.
\omer{are these true? there's a linear time reduction from computing gradient to computing a polynomial}}

\nop{\subsection{LExaBan and LExaShap Algorithms}
\label{subsect: LExaBan}
 Figure \ref{alg:LExaBan} shows the LExaBan algorithm that computes Banzhaf values for all variables in an input function $\varphi$ using a combination of the lifted compilation procedure and gradient-based computation.
The algorithm begins by applying the lifting compilation procedure detailed in Section \ref{subsubsect: lifting} and presented in Figure \ref{alg:lifted_compile}. It continues by treating the d-tree as an arithmetic circuit as described in \ref{subsubsect: arithmetic} and evaluating and differentiating that circuit to obtain the Banzhaf values for all variables in the function.
We note that a simple adaptation of the LExaBan algorithm could take a precompiled d-tree as input, enabling direct compilation into an arithmetic circuit followed by the computation of the value and gradient. The algorithm for Shapley value computation, LExaShap, is defined analogously.

\begin{algorithm}
\caption{LExaBan}
\label{alg:LExaBan}
\begin{algorithmic}[1]

\Require DNF $\varphi$
\Ensure Banzhaf values

    \State $T_\varphi \gets LiftedCompile(\varphi)$\\
    \Return $GradientBanzhaf(T_\varphi)$

\end{algorithmic}
\end{algorithm}}

\nop{\subsection{Expected Banzhaf and Shapley values}
\label{subsect: Expected_Banzhaf}
Our techniques extend to the expected Banzhaf or Shapley values suggested in \cite{Sig24:ExpectedShapleyLikeScores} for probabilistic databases.
\begin{definition}
Given a probabilistic database $D$, a fact $f\in D$ and a conjunctive query $Q$:
\begin{align*}
    &EBanzhaf(Q,f,D) \overset{def}{=}\\ &\sum_{D'\subseteq D_n, f\in D'} Pr[D']\cdot [Q(D' \cup D_x) - Q(D' \setminus \{f\} \cup D_x)]
\end{align*}
\end{definition}

Since we already treat each fact as if it is assigned a probability, Our algorithm and complexity results extend to this scenario with only minor modifications. The Expected Banzhaf value is achieved via evaluating the expectation and its gradient at the point of the probability vector of the different facts.  Thus, we are able to present an improved complexity for this scenario as well.}

\begin{algorithm}
\caption{$\gradientBanzhaf$}
\label{alg:GradientBanzhaf}
\begin{algorithmic}[1]
\Require D-tree $T_{\varphi}$ for formula $\varphi$
\Ensure Banzhaf values for all variables in $\varphi$
\State Initialize $v.p \gets \frac{1}{2}$ for each variable node $v$ \label{alg_line: GradientBanzhaf: compute probabilities - start}
\For{each node $v$ in the tree in bottom-up order} 
\State Compute $v.p$ according to Table~\ref{tab:probability_and_gradients_for_gates} and $v.gate$ \label{alg_line: GradientBanzhaf: compute probabilities - end}
\EndFor
\State  $T.root.g = 2^{|\text{vars}(\varphi)|-1}$ \label{alg_line: GradientBanzhaf: initialize_grad_root}
\For{each node $v$ excluding $T.root$ in top-down order}  \label{alg_line: GradientBanzhaf: grad backpropagation start}
\State $v.g \gets v.parent.g \cdot \frac{\partial v.parent.p}{\partial v.p}$ according to Table~\ref{tab:probability_and_gradients_for_gates} and $v.parent.gate$ \label{alg_line: GradientBanzhaf: grad backpropagation end}
\EndFor
\For{each variable $x\in\vars(\varphi)$}
\State $\banz(T_\varphi, x) \leftarrow \sum_{\text{node } v \text{ in } T \text{ for variable } x} v.g$ 
\EndFor
\State \Return $\banz(T_\varphi, x)$ for all variables $x\in\vars(\varphi)$ 
\end{algorithmic}

\end{algorithm}

%% file: Figures/arithmetic_circuit_example1.tex
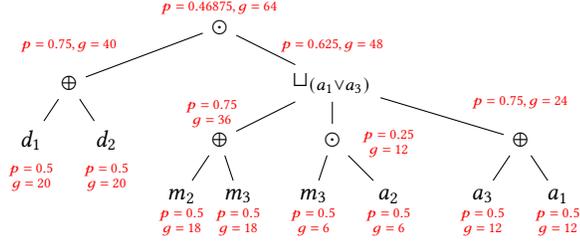
\begin{figure}
    \centering
    \begin{tikzpicture}
            \node at (0,0)  (root) {$\odot$};
            \node[] at (-2,-0.75)  (21) {$\oplus$} edge[-] (root);
            \node[font=\bfseries] at (1.5,-0.75)  (22) {$\sqcup_{(a_1\vee a_3)}$} edge[-] (root);
            \node[] at (-2.5,-1.5) (31) {$d_1$} edge[-] (21);
            \node[] at (-1.5,-1.5)  (32) {$d_2$} edge[-] (21);
            \node[font=\bfseries] at (0,-1.5)  (34) {$\oplus$} edge[-] (22);
            \node[font=\bfseries] at (1.5,-1.5)  (33) {$\odot$} edge[-] (22);
            \node[font=\bfseries] at (4,-1.5)  (35) {$\oplus$} edge[-] (22);
            \node[font=\bfseries] at (4.5, -2.25) (47) {$a_1$} edge[-] (35);
            \node[font=\bfseries] at (3.5, -2.25) (48) {$a_3$} edge[-] (35);
            \node[font=\bfseries] at (1.25, -2.25) (43) {$m_3$} edge[-] (33);
            \node[font=\bfseries] at (2.25, -2.25) (44) {$a_2$} edge[-] (33);
            \node[font=\bfseries] at (-0.5, -2.25) (45) {$m_2$} edge[-] (34);
            \node[font=\bfseries] at (0.25, -2.25) (46) {$m_3$} edge[-] (34);

            \node[red] at (0,0.25)  (x) {\tiny{$p=0.46875, g=64$}};
            \node[red] at (-2,-0.25)  (x) {\tiny{$p=0.75, g=40$}};
            \node[red] at (-2.5,-1.9)  (x) {\tiny{$p=0.5$}};
            \node[red] at (-2.5,-2.1)  (x) {\tiny{$g=20$}};
            \node[red] at (-1.5,-1.9)  (x) {\tiny{$p=0.5$}};
            \node[red] at (-1.5,-2.1)  (x) {\tiny{$g=20$}};
            \node[red] at (1.5,-0.25)  (x) {\tiny{$p=0.625, g=48$}};
            \node[red] at (-0.5,-2.5)  (x) {\tiny{$p=0.5$}};
            \node[red] at (-0.5,-2.7)  (x) {\tiny{$g=18$}};
            \node[red] at (0.25,-2.5)  (x) {\tiny{$p=0.5$}};
            \node[red] at (1.25,-2.5)  (x) {\tiny{$p=0.5$}};
            \node[red] at (0.25,-2.7)  (x) {\tiny{$g=18$}};
            \node[red] at (1.25,-2.7)  (x) {\tiny{$g=6$}};
            \node[red] at (2.25,-2.5)  (x) {\tiny{$p=0.5$}};
            \node[red] at (2.25,-2.7)  (x) {\tiny{$g=6$}};
            \node[red] at (4,-1)  (x) {\tiny{$p=0.75, g=24$}};
            \node[red] at (3.5,-2.5)  (x) {\tiny{$p=0.5$}};
            \node[red] at (3.5,-2.7)  (x) {\tiny{$g=12$}};
            \node[red] at (4.5,-2.5)  (x) {\tiny{$p=0.5$}};
            \node[red] at (4.5,-2.7)  (x) {\tiny{$g=12$}};
            \node[red] at (-0.1,-1.05)  (x) {\tiny{$p=0.75$}};
            \node[red] at (-0.1,-1.25)  (x) {\tiny{$g=36$}};
            \node[red] at (2.25,-1.45)  (x) {\tiny{$p=0.25$}};
            \node[red] at (2.25,-1.65)  (x) {\tiny{$g=12$}};
        \end{tikzpicture}
            \vspace{-0.3cm}

        \caption{Lifted d-tree for the lineage in Figure \ref{fig:running_example}. Each node is annotated with probability $p$ and partial derivative $g$ computed during Banzhaf value computation using gradients.}
        \label{fig:Lifted_Dtree_Example}
            \vspace{-1em}
\end{figure}

%% file: Sections_Arxiv_New/Aggregates.tex
\section{Attribution for Aggregate queries}
\label{sec: aggregates}


\nop{Notice that Hierarchical queries with aggregates are known to have a polynomial size d-trees that can be compiled in polynomial time from the lineage\cite{AggregationInProbabalisticDatabasesViaKnowlegeCompilation} and therefore, our proofs are going to extend the known tractability cases for boolean queries.}

\nop{Finally, we provide in Section~\ref{subsect: minmax} hardness result for min/max aggregate queries for any non-tractable case, completing the complexity dichotomy.}

\subsection{Basic Algorithms}
\label{subsect: linear aggregates}

We separately consider linear aggregates (SUM and COUNT) and non-linear but idempotent ones ($\minagg$ and $\maxagg$). 

\paragraph*{$\sumagg$ and $\countagg$} Our solution for linear aggregates uses linearity of Banzhaf and Shapley to reduce the problem to the case of SPJU queries. Let $\Phi=\sum_{i}^M \varphi_{i} \otimes m_i$ be a \nop{semimodule expression}\tensorset\ where the sum uses the $+_{M}$-operator of the monoid $M$, and 
let $\varphi = \bigvee \varphi_i$ be the Boolean part of 
$\Phi$. For  $\sumagg$, $+_{M}$ is the standard $+$, and we obtain: 

\begin{align*}
         &\banz(\Phi, x)  = \sum_{S\subseteq \vars(\varphi)\setminus\{x\}} \Phi[S\cup \{x\}] - \Phi[S] \\
        &= \sum_{S\subseteq \vars(\varphi)\setminus\{x\}} \sum_{i} \big(\varphi_i[S\cup \{x\}] - \varphi_i[S]\big) \cdot m_i  
        = \sum_{i} \banz(\varphi_i, x)\cdot m_i
    \end{align*}

Analogously, $\shap(\Phi, x)  = \sum_{i} \shap(\varphi_i, x)\cdot m_i$ (as observed already in \cite{TheShapleyValueofTuplesinQueryAnswering}).
For $\countagg$, we set $m_i=1$ for all $i$. Thus, to compute the  Banzhaf/Shapley value for a variable $x$ in $\Phi$, we (1) construct  a d-tree $T_i$ for each $\varphi_i$, (2) compute for each $T_i$ and each $x \in \varphi_i$ the  Banzhaf/Shapley value for $x$ as in Section~\ref{sec: Algorithm}, and (3) derive the Banzhaf/Shapley of $x$ in $\Phi$ using the above equations. 

\nop{

\begin{example}
Consider the semimodule expression
$\Phi = (a_1 \vee a_2 \vee a_3) \wedge m_3\otimes 377 +_{sum} (a_1 \vee a_3) \wedge m_2 \otimes 322 +_{sum}a_4 \wedge m_1\otimes 176$.
%
%
The Banzhaf value of variable $a_1$ with respect to each of the expressions 
$(a_1 \vee a_2 \vee a_3) \wedge m_3$, $(a_1 \vee a_3) \wedge m_2$, and $a_4 \wedge m_1$ is $16$, $32$, and $0$, respectively. Thus, the Banzhaf value of $a_1$ with respect to 
$\Phi$ is $16\cdot 377 + 32\cdot 322 +0 \cdot 176 = 16336$.
\end{example}
}

\paragraph{$\minagg$ and $\maxagg$} For $\minagg$ and $\maxagg$ which are not linear, we adopt a different approach. As explained in Section~\ref{sec:Preliminaries}, the lineage for queries with these aggregates is a semimodule expression. We first construct a d-tree for the entire semimodule expression and then compute Banzhaf/Shapley values from the d-tree. The former was shown in \cite{AggregationInProbabalisticDatabasesViaKnowlegeCompilation} and we next show how the latter is performed.  



Let $\pv_{\Phi}=\{\Phi[\theta] | \theta\subseteq \vars(\Phi)\}$ be the set of values to which the 
semimodule expression $\Phi$ can evaluate to under possible 
valuations. 
If $\Phi$ is the lineage of an aggregate query $\langle\minagg, \gamma, Q\rangle$ or $\langle\maxagg, \gamma, Q\rangle$, the size of 
$\pv_{\Phi}$ is bounded by the size of the result  of
$Q$, which is polynomial in the input
database size. We relate Banzhaf/Shapley values of a variable  $x$ in a semimodule expression
with the model counts of the expression under substitutions that set $x$ to $0$ or $1$:
\nop{
Note that for min and max monoids,the set of values to which $\Phi$ may be evaluated for some truth assignments is $PV_{\Phi}=\{\Phi(\theta) | \theta\subseteq \vars(\Phi)\}$. The size of $PV_{\Phi}$ equals the number of tuples in the underlying UCQ evaluation result, thus polynomial in the size of the database. We may restrict attention to $PV_{\Phi}$ when considering the Banzhaf and Shapley values, as follows.
}

\begin{proposition}
\label{corollary: Banzhaf for aggregates separated by value}
Given a semimodule expression $\Phi$ and a variable $x \in \vars(\Phi)$, it holds:
\begin{align*}
    \banz(\Phi, x) =& \sum_{p\in \pv_{\Phi}}{(\#^p \Phi[x:=1] - \#^p \Phi[x:=0])\cdot p}\\
\shap(\Phi, x) =& \sum_{p\in \pv_{\Phi}, k\in[\mid \vars(\Phi)\mid]}
\hspace{-1em}\big[ (\#_k^p \Phi[x:=1] - \#_k^p \Phi[x:=0])\cdot p\big] \cdot C_k
\end{align*}

where  $C_k = \frac{|k!|\cdot \mid \mid \vars(\Phi) \mid - k - 1|!}{\mid \vars(\Phi) \mid!}$ 
\end{proposition}
\begin{proof}
Proposition \ref{corollary: Banzhaf for aggregates separated by value} is proved by partitioning Equations \ref{eq: banzhaf value formula and semimodule},\ref{eq: shapley value formula and semimodule} based on the monoid values.

For the Banzhaf value:
{\small\begin{align*}
    &\banz(\Phi, x) = \sum_{\bm Y\subseteq \vars(\Phi) \setminus \{x\}}{\Phi(\bm Y\cup\{x\}) - \Phi(\bm Y)} \\=& \sum_{\bm Y\subseteq \vars(\Phi)\setminus \{x\}}{\Phi(\bm Y\cup\{x\})} - \sum_{\bm Y\subseteq \vars(\Phi)\setminus \{x\}}{\Phi(\bm Y)}\\ =&
    \sum_{p\in PV_\Phi}{\sum_{\substack{\bm Y\subseteq \vars(\Phi)\setminus \{x\},\\ \Phi(\bm Y\cup \{x\}) = p}}{\Phi(\bm Y\cup\{x\})}} - \sum_{p\in PV_\Phi}{\sum_{\substack{\bm Y\subseteq \vars(\Phi)\setminus \{x\},\\ \Phi(\bm Y) = p}}{\Phi(\bm Y)}} \\ =& \sum_{p\in PV_\Phi}{\sum_{\substack{\bm Y\subseteq \vars(\Phi)\setminus \{x\},\\ \Phi(\bm Y\cup \{x\}) = p}}{p}} - \sum_{p\in PV_\Phi}{\sum_{\substack{\bm Y\subseteq \vars(\Phi)\setminus \{x\},\\ \Phi(\bm Y) = p}}{p}} \\=&
    \sum_{p\in \pv_{\Phi}}{(\#^p \Phi[x:=1]} - \sum_{p\in \pv_{\Phi}}{\#^p \Phi[x:=0])\cdot p} \\ =& 
    \sum_{p\in \pv_{\Phi}}{(\#^p \Phi[x:=1] - \#^p \Phi[x:=0])\cdot p}
\end{align*}\normalsize}

The proof for Shapley values is similar.

\end{proof}



We can further show, via a dynamic programming algorithm:

\begin{proposition}
\label{prop:counts}
Let $T$ be a d-tree representing a semimodule expression $\Phi$ over a variable set $\bm X$ and the monoid $M=(\overline{\mathbb{R}}, \max,-\infty)$ or the monoid  $M=(\overline{\mathbb{R}}, \min,\infty)$. For $p \in \pv_\Phi$ and $k \in [|\vars(T)|]$, we can compute $\#^p \Phi$ and $\#_k^p \Phi$ in PTIME. 
\end{proposition}

\begin{proof}
We show the computation for binary gates and the results may easily be extended to handle non-binary gates as well. In the following, let $M\in \{\min, \max\}$ be the monoid.

\begin{lemma}
\label{lemma: prob_max independent And}

     Let $\Phi$ be a semimodule expression and let $\psi$ be a boolean formula such that $\psi$ and the boolean part of $\Phi$ are independent. Let $S = \psi \otimes \Phi$ and $p\in M$, $p\neq 0_M$, then:
{\small \begin{align*}
\#^p S &= \#\psi \cdot \#^p\Phi\\
\#_k^p S &= \sum_{i\in [k]}{\#_{i}\psi \cdot \#_{k-i}^p\Phi }
\end{align*} \normalsize}
\end{lemma}

\begin{proof}
The equations hold because each assignment (of size $k$) where $S$ evaluates to $p$ is constructed out of a satisfying assignment of $\psi$ and an assignment for $\Phi$ that evaluates to $p$ (where the combined sizes are $k$). 
\end{proof}

\begin{lemma}
\label{lemma: prob_max independent Or}
    Let $S = \Phi_1 +_M \Phi_2$ such that $\Phi_1$ and 
    $\Phi_2$ are independent, then:
\small{
\begin{align*}
&\#^p S = \sum_{r,t\in PV_{\Phi}, \, r +_M t = p} \left( \#^r\Phi_1 \cdot \#^t\Phi_2 \right) \\
&\quad + \#^p\Phi_1 \cdot \left( 2^{|\vars(\Phi_2)|} - \#\Phi_2 \right) 
+ \#^p\Phi_2 \cdot \left( 2^{|\vars(\Phi_1)|} - \#\Phi_1 \right) \\
&\#_k^p S = \sum_{i\in [k]} \Bigg[ \sum_{r,t\in PV_{\Phi}, \, r +_M t = p} \left( \#_i^r\Phi_1 \cdot \#_{k-i}^t\Phi_2 \right) \\
&\quad + \#_i^p\Phi_1 \cdot \left( \binom{|\vars(\Phi_2)|}{k-i} - \#_{k-i}\Phi_2 \right)
\quad + \#_{k-i}^p\Phi_2 \cdot \left( \binom{|\vars(\Phi_1)|}{i} - \#_i\Phi_1 \right) \Bigg]
\end{align*}}

\end{lemma}

\begin{proof}
The equations hold because every satisfying assignment (of size $k$) where $S$ evaluates to $p$ is constructed out of a assignment for which $\Phi_1$ evaluates to $r$ and an assignment for which $\Phi_2$ evaluates to $t$ (where the sum of assignments sizes is $k$) such that $r +_M t = p$. This can be seen as the sum of such assignments, by whether $\Phi_1$ alone is satisfied, $\Phi_2$ alone is satisfied, or both. 
\end{proof}

\begin{lemma}
\label{lemma: prob_max exclusive or}

Let $\Phi_1, \Phi_2$ be exclusive semimodule expressions s.t. $\Phi_1$ and 
    $\Phi_2$ have identical sets of variables. Let $S = \Phi_1 +_M \Phi_2$, then:
\begin{align*}
\#^p S &= \#^p \Phi_1 + \#^p\Phi_2\\
\#_k^p S &= \#_k^p \Phi_1 + \#_k^p\Phi_2
\end{align*}
\end{lemma}

\begin{proof}
Since the satisfying assignments for the boolean parts of the semimodule expressions are disjoint, any assignment (of size $k$) of $S$ that evaluates to $p$ is either an assignment (of size $k$) that evaluates to $p$ of $\Phi_1$ or an assignment that evaluates to $p$ of $\Phi_2$.
\end{proof}

Having shown the computation for each type of gate, a bottom-up algorithm to compute assignment counts for all (polynomialy many) relevant values $p$ and $k$ and all d-tree nodes follows directly, thereby concluding the proof of Proposition \ref{prop:counts}.\end{proof}

\begin{example}
    Figure \ref{fig:complete_semimodule_d-tree} presents a d-tree for the lineage of $Q_2$ from Table \ref{tab:lineages_lifted}, along with some steps of the model counts computation. Denote the $\oplus$ node that is the right child of the right child of the root by $T$. Its left subtree has one assignment with value of $322$, and its right subtree has one assignment with value of $377$. Applying the equation from Lemma \ref{lemma: prob_max independent Or} we get that $\#^{377}T = 1\cdot 1 + 1\cdot 1 +0\cdot 1 = 2$, as seen in the Figure. 
\end{example}
\input{Figures/semimodule_dtree}

Combined, the two propositions provide a PTIME algorithm for computing Banzhaf and Shapley values given a d-tree representation of the lineage of an aggregate query.



\nop{
\begin{proof}
We show the computation for binary gates and the results may easily be extended to handle non-binary gates as well. In the following, let $M\in \{\min, \max\}$ be the monoid.

\begin{lemma}
\label{lemma: prob_max independent And}

     Let $\Phi$ be a semimodule expression and let $\psi$ be a boolean formula such that $\psi$ and the boolean part of $\Phi$ are independent. Let $S = \psi \otimes \Phi$ and $p\in M$, $p\neq 0_M$, then:
 \begin{align*}
\#^p S &= \#\psi \cdot \#^p\Phi\\
\#_k^p S &= \sum_{i\in [k]}{\#_{i}\psi \cdot \#_{k-i}^p\Phi }
\end{align*} 
\end{lemma}

\begin{proof}
The equations hold because each assignment (of size $k$) where $S$ evaluates to $p$ is constructed out of a satisfying assignment of $\psi$ and an assignment for $\Phi$ that evaluates to $p$ (where the combined sizes are $k$). 
\end{proof}

\begin{lemma}
\label{lemma: prob_max independent Or}
    Let $S = \Phi_1 +_M \Phi_2$ such that $\Phi_1$ and 
    $\Phi_2$ are independent, then:

\begin{align*}
&\#^p S = \sum_{r,t\in PV_{\Phi}, \, r +_M t = p} \left( \#^r\Phi_1 \cdot \#^t\Phi_2 \right) \\
&\quad + \#^p\Phi_1 \cdot \left( 2^{|\vars(\Phi_2)|} - \#\Phi_2 \right) 
+ \#^p\Phi_2 \cdot \left( 2^{|\vars(\Phi_1)|} - \#\Phi_1 \right) \\
&\#_k^p S = \sum_{i\in [k]} \Bigg[ \sum_{r,t\in PV_{\Phi}, \, r +_M t = p} \left( \#_i^r\Phi_1 \cdot \#_{k-i}^t\Phi_2 \right) \\
&\quad + \#_i^p\Phi_1 \cdot \left( \binom{|\vars(\Phi_2)|}{k-i} - \#_{k-i}\Phi_2 \right)
\quad + \#_{k-i}^p\Phi_2 \cdot \left( \binom{|\vars(\Phi_1)|}{i} - \#_i\Phi_1 \right) \Bigg]
\end{align*}

\end{lemma}

\begin{proof}
The equations hold because every satisfying assignment (of size $k$) where $S$ evaluates to $p$ is constructed out of a assignment for which $\Phi_1$ evaluates to $r$ and an assignment for which $\Phi_2$ evaluates to $t$ (where the sum of assignments sizes is $k$) such that $r +_M t = p$. This can be seen as the sum of such assignments, by whether $\Phi_1$ alone is satisfied, $\Phi_2$ alone is satisfied, or both. 
\end{proof}

\begin{lemma}
\label{lemma: prob_max exclusive or}

Let $\Phi_1, \Phi_2$ be exclusive semimodule expressions s.t. $\Phi_1$ and 
    $\Phi_2$ have identical sets of variables. Let $S = \Phi_1 +_M \Phi_2$, then:
\begin{align*}
\#^p S &= \#^p \Phi_1 + \#^p\Phi_2\\
\#_k^p S &= \#_k^p \Phi_1 + \#_k^p\Phi_2
\end{align*}
\end{lemma}

\begin{proof}
Since the satisfying assignments for the boolean parts of the semimodule expressions are disjoint, any assignment (of size $k$) of $S$ that evaluates to $p$ is either an assignment (of size $k$) that evaluates to $p$ of $\Phi_1$ or an assignment that evaluates to $p$ of $\Phi_2$.
\end{proof}

Having shown the computation for each type of gate, a bottom-up algorithm to compute assignment counts for all (polynomialy many) relevant values $p$ and $k$ and all d-tree nodes follows directly, thereby concluding the proof of Proposition \ref{prop:counts}.\end{proof}}
\nop{
\begin{example}
    Figure \ref{fig:complete_semimodule_d-tree} presents a d-tree for the lineage of $Q_2$ from Table \ref{tab:lineages_lifted}, along with some steps of the model counts computation. Denote the $\oplus$ node that is the right child of the right child of the root by $T$. Its left subtree has one assignment with value of $322$, and its right subtree has one assignment with value of $377$. Applying the equation from Lemma \ref{lemma: prob_max independent Or} we get that $\#^{377}T = 1\cdot 1 + 1\cdot 1 +0\cdot 1 = 2$, as seen in the Figure. 
    
    
\end{example}
}


Combined, the two propositions provide a PTIME algorithm for computing Banzhaf and Shapley values given a d-tree representation of the lineage of an aggregate query.

\nop{
\paragraph{Hardness result}
To complete the complexity analysis, we provide a reduction from computing Banzhaf values for CQ to computing Banzhaf values for CQ with the Max aggregate. 

\begin{theorem}
\label{theorem: hardness_max_banzhaf}
    Let $Q$ be  a  CQ and $D$ be a database and let $t$ be a tuple in $D$. Let $\alpha$ be a function that maps each output of $Q$ to value 1. 
We get that:
\begin{align*}
\banz(Q,D,t) =& \banz(\langle \maxagg, \alpha, Q\rangle,D,t)
\end{align*}

\begin{proof}
     Let $T = \sum_{Max} \varphi_i \otimes m_i$ be the semimodule expression of the lineage of $\langle \maxagg, \alpha, Q\rangle$. We denote the boolean part of $T$ by $\varphi = \bigvee \varphi_i$. By the building of $\alpha$, we get that for all $i$ $m_i=1$.
     Using axiom 2 in definition \ref{def: semimodule}, we can write:
     \begin{align*}
         T = \varphi \otimes 1
     \end{align*}

     Let $A$ be a truth valuation for the variables in $\varphi$. 
     We get that \begin{align*}
         T(A) =\begin{cases} 
    1, & \text{if } \varphi(A) = 1 \\
    0, & \text{otherwise}
    \end{cases}
     \end{align*}

     Based on Equation \ref{eq: banzhaf value formula and semimodule} We get that \begin{align*}
         \banz&(\langle \maxagg, \alpha, Q\rangle,D,t) = \#T[t=1] - \#T[t=0]\\ 
         =&  \#\varphi[t=1] - \#\varphi[t=0]
     \end{align*}
     
     Based on the definition of the lineage, we get that $lin(Q,D) = \varphi$, and based on the definition of the Banzhaf value we get that \begin{align*}\banz(Q,D,t) =&  \#\varphi[t=1] - \#\varphi[t=0] \\=&  \banz(\langle \maxagg, \alpha, Q\rangle,D,t)\end{align*}
     \end{proof}
\end{theorem}}



\subsection{Lifting for Aggregates}
\label{sec: Aggregate_extension} 
We extend our lifting technique (Section~\ref{subsubsect: lifting}) from UCQs to aggregate queries. First, for linear aggregates ($\sumagg$ and $\countagg$), our solution in Section~\ref{subsect: linear aggregates} constructs a separate d-tree for each 
formula $\varphi_i$ in the Boolean part $\bigvee_i \varphi_i$ 
of the given \tensorset\ expression. Thus, we can apply $\lift$  and $\liftedcompile$ in Algorithms~\ref{alg:Lift} and \ref{alg:lifted_compile} to each formula $\varphi_i$.
For $\minagg$ and $\maxagg$, new machinery is needed. We show the construction for $\maxagg$; the construction for $\minagg$ is analogous.

\paragraph*{Extending the Definitions}
Given disjoint variable sets $\bm X$ and $\bm Y$,  a {\em lifted semimodule expression} over $\bm X \cup \bm Y$ is a pair $(\varphi, \ell)$ where $\varphi$ is a Boolean formula over $\bm X\cup \bm Y$ and $\ell: \vars(\varphi) \to  \posbool(\bm X)\cup (\posbool(\bm X)\otimes M)$ is a mapping from variables in $\varphi$ to either positive Boolean formulas over $\bm X$ or semimodule expressions over $\bm X$ and semimodule $M$ (denoted by $\posbool(\bm X)\cup (\posbool(\bm X)\otimes M)$). If $\varphi$ is in DNF, we refer to $(\varphi, \ell)$ as a {\em lifted DNF semimodule expression}.  
The {\em inlining} $\inline(\varphi, \ell)$ is obtained
in two steps: First, we replace each variable $x\in \vars(\varphi)$ by $\ell(x)$; then, we replace each $\vee$ between semimodule expressions by $+_M$ and each $\wedge$ between a formula and a semimodule expression by $\otimes$. We call a lifted semimodule expression {\em valid} if exactly one variable in each DNF clause is mapped to a semimodule expression.

\paragraph*{Translating Lineage to a Lifted DNF Semimodule Expression} Recall that the lineage of an aggregate query is captured by a semimodule expression $\Phi = \sum_i^M \varphi_i \otimes m_i$, where each $\varphi_i = \bigvee_{j}C_{ij}$ is in DNF. We transform $\Phi$ to the equivalent form $\Phi'=\sum_{i,j}^MC_{i,j} \otimes m_i$. This preserves equivalence for $(\overline{\mathbb{R}}, \max, -\infty)$ and  $(\overline{\mathbb{R}}, \min, \infty)$ due to their idempotence. We further transform $\Phi'$ into a lifted DNF semimodule expression $(\varphi, \ell)$ 
with $\varphi = \bigvee_{i,j} (C_{i,j}\wedge w_i)$ and  $\ell = \{w_i\mapsto m_i\} \cup \{x\mapsto x \mid x\in \vars(\varphi) \text{ and } x\neq w_i\}$. Note that $(\varphi,\ell)$ is valid and its inlining is equivalent to $\Phi$. We add this translation as an initial preprocessing step to Algorithm \ref{alg:lifted_compile}.

\paragraph{Extending the Lifting Algorithm}
Given a lifted DNF semimodule expression $(\varphi,\ell)$, 
cofactor-equivalence and interchangeability are defined as for lifted Boolean DNF formulas.
The application of  $\conlift$ or $\dislift$ to a set of cofactor-equivalent or interchangeable variables maintains the validity of the expression. Algorithm \ref{alg:Lift} may then be applied to lifted DNF semimodule expressions.

%


\nop{
\paragraph*{Extending the Definitions} Given disjoint variable sets $\bm X$ and $\bm Y$,  a {\em lifted DNF semimodule expression} over $\bm X \cup \bm Y$ is a pair $(\varphi, \ell)$ where $\varphi$ is a positive DNF formula over $\bm X\cup \bm Y$ and $\ell: \vars(\varphi) \to  \posbool(\bm X)\cup (\posbool(\bm X)\otimes M)$ is a mapping from variables in $\varphi$ to either semimodule expressions or Boolean formulas. The {\em inline} operation replaces, similarly to its definition in Section~\ref{subsubsect: lifting}, every variable $x$ by $\ell(x)$; note that if $\ell(x)$ is in $\posbool(\bm X)\otimes M$.

\paragraph*{Extending the Lifting Algorithm} Recall that the lineage of an aggregate query may be captured by a semi-module expression $\sum \varphi_i \otimes m_i$ where each $\varphi_i = \bigvee_{j}C_ij$ is in DNF. We transform each term $\varphi_i \otimes m_i$ to the equivalent $\sum_{j}C_ij \otimes m_i$. This preserves equivalence for $(\overline{\mathbb{R}}, \max, -\infty), (\overline{\mathbb{R}}, \min, \infty)$ due to their idempotence.

The inline operation works similarly to Section~\ref{subsubsect: lifting}, where we map every variable $v$ in $\varphi$ to its respective mappings in $\ell$. Note that for simplicity we treat the $\wedge$ similarly to $\otimes$ and $\vee$ similarly to $+_M$ when it is appropriate.

Given a semimodule expression $\Phi$, we begin by compiling it to  the form 
$\sum_i C_i\otimes m_i$ where $C_i$ is a conjunction of literals and the summation uses the $+_M$ operator. Note that since $m +_M m = m$, this may be done while maintaining equivalence with $\Phi$. We further transform it to lifted semimodule DNF by introducing a fresh variable $w_i$ for each clause $C_i$ and create a mapping $\ell$ such that $\ell(w_i) = 1\otimes m_i$. Now, Algorithm~\ref{alg:lifted_compile} works the same, only that we replace each execution of algorithm \ref{alg:Lift} with two executions. For the first one, we replace every variable $v\in \vars(\varphi)$ such that $\ell(v) = 1\otimes x$ where $x\in \overline{\mathbb{R}}$ with a variable $w_x$, and update $\ell$ accordingly.
For the second execution, we replace every appearance of a variable $v\in \vars(\varphi)$ such that $\ell(v) = 1\otimes x$ where $x\in \overline{\mathbb{R}}$ with a fresh variable $w_v$, and update $\ell$ accordingly.
}

\begin{example}
    Consider the lineage of Query $Q_2$ in Table~\ref{tab:lineages_lifted}. We first translate it into the lifted DNF semimodule expression $\varphi = (a_1\wedge m_3\wedge t_1)\vee
            (a_2\wedge m_3\wedge t_1) \vee 
            (a_3\wedge m_3 \wedge t_1) \vee 
            (a_1 \wedge m_2 \wedge t_2) \vee 
            (a_3 \wedge m_2 \wedge t_2) \vee (a_4 \wedge m_1\wedge t_3)$; $\ell = \{t_1 \mapsto 377, t_2 \mapsto 322, t_3\mapsto 176\}\cup \{x\mapsto x | x\in \vars(\varphi)\setminus\{t_1,t_2,t_3\}\}$. Then, we execute the initial lifting in Algorithm \ref{alg:lifted_compile}. The algorithm searches for  cofactor-equivalent or interchangeable sets of variables. 
            It detects that $a_1$ and $a_3$ are cofactor-equivalent and performs $\dislift$ on them. It also detects that 
            each of the sets $\{m_3,t_1\}$, $\{m_2,t_2\}$, and $\{a_4,m_1,t_3\}$ are interchangeable and thus performs $\conlift$ on them. Table \ref{tab:lineages_lifted} gives the resulting expression. The algorithm continues he compilation as described in Section~\ref{subsubsect: lifting} to obtain a d-tree for the lineage.
\end{example}
\nop{
\begin{example}
Consider the semimodule expression that is the lineage of Query $Q_2$ as presented in Table \ref{tab:lineages_lifted}. We start by translating the expression to the lifted DNF semimodule expression $\varphi = ((a_1\wedge m_3) \otimes 377 +_{\max}
            (a_2\wedge m_3) \otimes 377 +_{\max}
            (a_3\wedge m_3) \otimes 377 +_{\max}
            (a_1 \wedge m_2) \otimes 322 +_{\max}
            (a_3 \wedge m_2) \otimes 322
            +_{\max} (a_4 \wedge m_1) \otimes 176)$. We first start the lifting process by treating each occurrence of variables representing a monoid value as a recurring appearance of the same variable. Similarly to Section~\ref{subsubsect: lifting}, the algorithm tries to find cofactor-equivalent or interchangeable sets of variables. The algorithm detects that $a_1$ and $a_3$ are cofactor-equivalent, thus it performs $\dislift$ on them. In addition, the algorithm detects that the sets $\{m_3,377\}$, $\{m_2,322\}$ and $\{a_1,m_1,176\}$ appear in similar clauses, thus it performs $\conlift$ on them. The resulting expression after these operations is in Table \ref{tab:lineages_lifted}. Next, it tries to treat each variable mapped to a monoid value as a fresh variable, but no such variable exists in the mapping $\ell$. Thus, it returns the lineage from the previous step as the lifted lineage.
\end{example}}

\subsection{Gradients for Aggregates}
We now adapt the gradient-based computation from Section~\ref{subsubsect: arithmetic} to the MAX aggregate (the case of MIN is similar).  
In contrast to d-trees with Boolean outcomes, a semimodule expression can yield a numeric value for a given valuation over its Boolean variables. Our probabilistic interpretation looks at  {\em expected values}. 

\paragraph*{Expected Values and Gradients} Given a d-tree $T$ for a semimodule 
expression $\Phi$, $E[T]$ denotes the expected value of $T$ given probabilities $p_x$ for the Boolean variables: $E[T] = \sum_{S\subseteq \vars(T)}{Pr[S]\cdot T[S]}$. 
Prop. \ref{Prop: Banzhaf_is_Gradient} extends to expected values:
\begin{align}
\label{Eq: Banzhaf_is_expected_grad}
    \banz(T,x) =& 2^{|\vars(T)| - 1} \cdot \left(\frac{\partial E[T]}{\partial p_x}\left(\vec{\frac{1}{2}}\right)\right) 
\end{align}

We define the probability and k-probability of a d-tree $T$ as follows: 
$Pr[T=m_i]= \sum_{S\subseteq \vars(T), T[S]=m_i} {Pr[S]}$ and \nop{$Pr_k[T=m_i]= \sum_{S\subseteq \vars(T), \mid S \mid = k, T[S]=m_i} {Pr[S]}$.}
$ 
Pr_k[T=m_i] = \left(\frac{1}{2}\right)^{|\vars(T)| - k} \sum_{{\substack{S\subseteq \vars(T) \\ |S| = k,T[S]=m_i}}}\sum_{S'\subseteq S} T[S']\cdot \prod_{y\in S'}p_y \prod_{z\in S\setminus S'}\left(\frac{1}{2}-p_z\right)
$.

Using the linearity of expectation, we can write:
\begin{align*}
\frac{\partial E[T]}{\partial p_x} =& \sum_{m_i\in \pv_{\Phi}} m_i\cdot \frac{\partial Pr[T=m_i]}{\partial p_x}
\end{align*}

\paragraph*{From Gradients to Banzhaf and Shapley} Similarly to Section~\ref{subsubsect: arithmetic}, we obtain the Banzhaf and Shapley values of all facts via partial derivatives. Assume an order over the variables and then use variable names as indices (identifying each variable with its position in the order). Let $J$ be the Jacobian matrix, namely $J^i_x = \frac{\partial(v_i\cdot Pr[T=v_i])}{\partial p_x}$. Further let $J^k$ be the matrix defined by $J^{k,i}_x = \frac{\partial(v_i\cdot Pr_k[T=v_i])}{\partial p_x}$ for each variable $x$. 
Note that the entries $J^i_x$ and $J^{k,i}_x$ in the Jacobian matrices are functions of all variables in $\vars(T)$.
We can show:
 \begin{align*}
\banz(T,x) =& 2^{|\vars(T)| - 1} \cdot \sum_{i \in [n]}{J^i_x \left(\vec{\frac{1}{2}}\right)}\\
\shap(T,x) =& 2^{|\vars(T)| - 1} \cdot \sum_{i \in [n]}{\sum_{k\in [|\vars(T)|]}{J^{k,i}_x \left(\vec{\frac{1}{2}}\right) \cdot C_{k-1}}}
\end{align*} 
where $C_k = \frac{|k!|\cdot|D_n - k - 1|!}{|D_n|!}$. 

\paragraph*{Banzhaf and Shapley Computation} Similarly to Section~\ref{subsubsect: arithmetic}, we can compute the Banzhaf or Shapley value by using the d-trees to compute the expectation, and then compute the Jacobians via back propagation. We adapt Algorithm \ref{alg:GradientBanzhaf} to incorporate these changes for a d-tree representing the lineage of an aggregate query. \nop{We adapt Algorithm~\ref{alg:GradientBanzhaf} to compute Jacobians via back-propagation, for a d-tree representing the lineage of an aggregate query.}  Lines \ref{alg_line: GradientBanzhaf: compute probabilities - start}-\ref{alg_line: GradientBanzhaf: compute probabilities - end} are replaced with the computation of the probability for each different possible value. The root derivative in Line \ref{alg_line: GradientBanzhaf: initialize_grad_root} is initialized to a vector of possible values. The computation in Lines \ref{alg_line: GradientBanzhaf: grad backpropagation start}-\ref{alg_line: GradientBanzhaf: grad backpropagation end} is replaced with the computation of the Jacobian matrix values.  This allows us to compute a vector/matrix of derivatives at each node and obtain the Banzhaf/Shapley values at the leaves according to Equation \ref{Eq: Banzhaf_is_expected_grad}.

\nop{Let $T = \sum {\varphi_i \otimes m_i}$ be the semimodule expression lineage of $Q$ on $D$. The symbolic probabilistic extension of $T$ is achieved similarly to the boolean case by attaching each variable $v$ in its Boolean part a probability $p_v$. We define the expected value of $T$ analogously to that of a query, such that the expected value of $T$ coincides with the expected value of $Q$ under the mapping of database facts to lineage variables. We can define the probability of a certain value to be:
\begin{align*}
    Pr[T=m_i]= \sum_{\theta\subseteq \vars(T), T(\theta)=m_i} {Pr[\theta]}
\end{align*}

Examining the possible value distribution for $T$ over the variable probabilities, we get that the expected value of $T$ can be written as:
\begin{align*}
E[T] =& \sum_{m_i} m_i \cdot Pr[T=m_i]
\end{align*}

The partial derivative of the expected value over the probability of a variable $f$ is obtained by:

\begin{align*}
\frac{\partial E[T]}{\partial p_f} =& \sum_{m_i} m_i\cdot \frac{\partial Pr[T=m_i]}{\partial p_f}
\end{align*}

Hence, we can obtain the Banzhaf and Shapley values of all variables by performing a simple calculation over the matrix of partial derivatives (also known as the Jacobian) of multivalued function that gives us the distribution of answers. Let $F = (v_1\cdot Pr[T=v_1], \dots, v_n \cdot Pr[T=v_n])$  be the multivalued function that computes the distribution of answers for $n$ possible values. We order the variables of $T$ such that for each variable $x$, index $x$ corresponds to the number of that variable in the order.

\begin{align*}
\banz(T,f) =& 2^{|T| - 1} \cdot \sum_{i \in [n]}{J^i_f (\vec{\frac{1}{2})}}\\
\shap(T,f) =& 2^{|T| - 1} \cdot \sum_{i \in [n]}{\sum_{k}{^kJ^i_f (\vec{\frac{1}{2})} \cdot C_k}}\\
\end{align*}
Where $C_k = \frac{|k!|\cdot|D_n - k - 1|!}{|D_n|!}$ is the Shapley coefficient for size $k$.
}

\nop{Hence, we can obtain the Banzhaf values of all variables by performing a simple calculation over the Jacobian matrix of the multivalued function that gives us the distribution of the answers. Let $F = (v_1\cdot Pr[T=v_1], \dots, v_n \cdot Pr[T=v_n])$ be the multivalued function that computes the distribution of answers for $n$ possible values. We order the variables of $T$ such that for each variable $f$, index $f$ corresponds to the number of that variable in the order. Let $J$ be the Jacobian matrix such that $J^i_f = \frac{\partial(v_i\cdot Pr[T=v_i])}{\partial p_f}$ and $^kJ$ be the matrix such $^kJ^i_f = \frac{\partial(v_i\cdot Pr_k[T=v_i])}{\partial p_f}$ for each variable $f$ representing a fact in the database. We get that:}

\nop{We now outline the computation of the Jacobian matrix. To formalize the calculation, it is useful to treat each MAX semimodule expression as a distribution over its different values, with a special value $v_\neg$  representing unsatisfied assignments and being smaller than all other values. For boolean formulas, we consider a distribution of $0_M$ for satisfying assignments, and $v_\neg$ for unsatisfying ones. With this perspective, the $+_M$ gate for semimodule expressions and the $\vee$ gate for formulas both correspond to taking the maximal value of the distributions. Similarly, the $\otimes$ and $\wedge$ gates correspond to taking the maximal value of the distributions, provided no value is $v_\neg$. Lastly, conditioning on a variable $x$ with probability $p_x$ corresponds to sampling from one distribution with a probability of $p_x$ and from the other with probability of $1-p_x$.

\begin{lemma}
Let $D = d_1, d_2, \dots, d_n$ be independent distributions over ordered values $V = v_\neg, v_1, v_2, \dots, v_k$. Let $D_{max}$ be the distribution of the maximal value between the distributions of $D$. We mark by $J^i(D_{max})_{j,k}$ the partial derivative of the probability of value $v_j$ in $D_{max}$ with respect to the probability of $v_k$ in $d_i$. We get that:
\begin{align*}
    J^i(D_{max})_{j,k} =& (\Delta^{\leq}(j,k) - 1)\cdot\prod_{l\in [n]-\{i\}} CDF_l(m_k) + (\Delta^{<}(j,k) + 1) \\ \cdot& \prod_{l\in [n]-\{i\}} (CDF_l(m_k) - d_l(m_k))
\end{align*}
Where $CDF_l$ is the cumulative distribution function of the distribution $d_l$  and $\Delta^{op}(j,k) = 1$ if $m_j \,  op  \, m_k$ and $0$ otherwise.
\end{lemma}

Essentially, this value is equivalent to the change in the  probability distribution of $D_{max}$ if we were to replace $d_i$ with a constant distribution of the value $v_j$ and subtract from that the distribution without the $d_i$. Since the probability of each value in $D_{max}$ is multilinear with respect to the probabilities of the values in the different distributions of $D$, this gives us the Jacobian matrix.
The Jacobian calculation for the rest of the gates is omitted for lack of space.}

\nop{

\begin{lemma}
Let $D = d_1, d_2, \dots, d_n$ be independent distributions over ordered values $V = v_\neg, v_1, v_2, \dots, v_k$. Let $D^+_{max}$ be the distribution achieved by taking $v_\neg$ if any of values is $v_\neg$, and otherwise taking the maximal value between the distributions of $D$. We mark by $J^i(D_{max})_{j,k}$ the partial derivative of the probability of value $v_j$ in $D_max$ with respect to the probability of $v_k$ in $d_i$. We get that:
\begin{align*}
    J^i(\varphi_{max})_{j,k} =& (\Delta^{\leq}(j,k) - 1)\cdot\prod_{l\in [n]-\{i\}} CDF^+_l(m_k) + (\Delta^{<}(j,k) + 1) \\\cdot& \prod_{l\in [n]-\{i\}} (CDF^+_l(m_k) - d_l(m_k))
\end{align*}
Where $CDF_l$ is the cumulative distribution function of the distribution $d_l$ and $CDF^+_d(m) = CDF_d(m) - d(v_\neg)$.  We define $\Delta^{op}(j,k) = 1$ if $m_j \,  op  \, m_k$ and $0$ otherwise.
\end{lemma}




\begin{lemma}
    Let $L$ and $R$ be distributions over ordered values $V = v_\neg, v_1, v_2, \dots , v_k$. Let $f$ be a random variable with a probability of $p_f$. Let $D_{exclusive}$ be the distribution acquired by sampling from $L$ with probability $p_f$ and sampling from $R$ with probability $1-p_f$. 

    \begin{align*}
    J^L (D_{exclusive})_{j,k} =& \, \Delta^{=}(j,k) \cdot D_{exclusive}(j) \cdot p_f\\
    J^R (D_{exclusive})_{j,k} =& \, \Delta^{=}(j,k) \cdot D_{exclusive}(j) \cdot (1 - p_f)\\
    J(D_{exclusive})_{p_f, k} =& \, T(m_k) - F(m_k)
\end{align*}
\end{lemma}

}

\nop{\subsection{LExaBan for Min/Max}
\label{subsect:aggregate_LExaBan}
Now, applying the changes to the lift algorithm and the gradient approach gives us the LExaBan algorithm for the MIN/MAX case.}

\paragraph*{GROUP-BY and Additional Aggregates}
We have focused so far on aggregate queries without grouping and with MIN, MAX, SUM and COUNT. 
Our solution directly extends to queries with $\groupbyagg$: as in prior work \cite{ProvenanceForAggregateQueries, TheShapleyValueofTuplesinQueryAnswering}, lineage is defined for each group as a \tensorset/semimodule expression. Shapley and Banzhaf values are then defined and computed with respect to each group, exactly as for aggregates without $\groupbyagg$. Linear aggregates can be supported like  $\sumagg$ and $\countagg$, by
computing the Banzhaf/Shapley value for each term in the \tensorset and then combining these values using the aggregate function. In contrast, non-linear aggregates such as AVG require non-trivial development, in particular to define the aggregate-specific gradients as in Table \ref{tab:probability_and_gradients_for_gates}. This is left as future work.

%% file: Figures/semimodule_dtree.tex
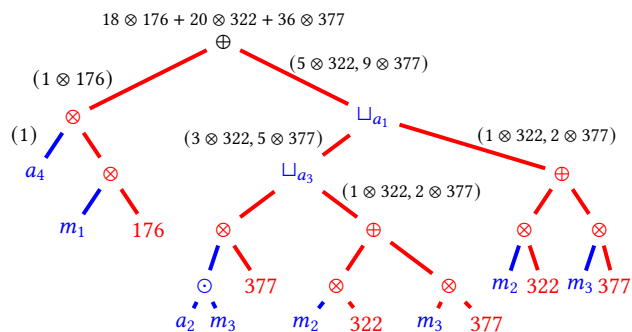
\begin{figure}[t!]
\begin{tikzpicture}
  \node at (0,0)  (root) {$\oplus$};
  \node[red] at (-2,-1)  (21) {$\otimes$} edge[-, red, line width=1.5pt] (root);
  \node[blue, font=\bfseries] at (2,-1)  (22) {$\sqcup_{a_1}$} edge[-, red, line width=1.5pt] (root);
  \node[blue] at (-2.5,-1.75) (31) {$a_4$} edge[-, blue, line width=1.5pt] (21);
  \node[red] at (-1.5,-1.75)  (32) {$\otimes$} edge[-, red, line width=1.5pt] (21);
\node[blue] at (-2,-2.5)  (41) {$m_1$} edge[-, blue, line width=1.5pt] (32);
\node[red] at (-1,-2.5)  (42) {$176$} edge[-, red, line width=1.5pt] (32);

  \node[blue, font=\bfseries] at (1,-1.75)  (34) {$\sqcup_{a_3}$} edge[-, red, line width=1.5pt] (22);
  \node[red, font=\bfseries] at (4.5,-1.75)  (33) {$\oplus$} edge[-, red, line width=1.5pt] (22);
  
  \node[red, font=\bfseries] at (4, -2.5) (43) {$\otimes$} edge[-, red, line width=1.5pt] (33);
  \node[red, font=\bfseries] at (5, -2.5) (44) {$\otimes$} edge[-, red, line width=1.5pt] (33);
    \node[blue, font=\bfseries] at (4.75, -3.25) (51) {$m_3$} edge[-, blue, line width=1.5pt] (44);
      \node[red, font=\bfseries] at (5.2, -3.25) (52) {$377$} edge[-, red, line width=1.5pt] (44);
  
  \node[blue, font=\bfseries] at (3.75, -3.25) (53) {$m_2$} edge[-, blue, line width=1.5pt] (43);
  \node[red, font=\bfseries] at (4.25, -3.25) (54) {$322$} edge[-, red, line width=1.5pt] (43);
  
  \node[red] at (0, -2.5) (45) {$\otimes$} edge[-, red, line width=1.5pt] (34);
  \node[red, font=\bfseries] at (2, -2.5) (46) {$\oplus$} edge[-, red, line width=1.5pt] (34);
  
  \node[blue, font=\bfseries] at (-0.25, -3.25) (56) {$\odot$} edge[-, blue, line width=1.5pt] (45);
    \node[blue, font=\bfseries] at (-0.5, -3.75) (61) {$a_2$} edge[-, blue, line width=1.5pt] (56);
        \node[blue, font=\bfseries] at (0, -3.75) (62) {$m_3$} edge[-, blue, line width=1.5pt] (56);

  \node[red, font=\bfseries] at (0.5, -3.25) (57) {$377$} edge[-, red, line width=1.5pt] (45);
  
  \node[red, font=\bfseries] at (1.5, -3.25) (58) {$\otimes$} edge[-, red, line width=1.5pt] (46);
  \node[red, font=\bfseries] at (3, -3.25) (59) {$\otimes$} edge[-, red, line width=1.5pt] (46);
  
  \node[blue] at (1.1, -3.75) (63) {$m_2$} edge[-, blue, line width=1.5pt] (58);
  \node[red] at (1.9, -3.75) (64) {$322$} edge[-, red, line width=1.5pt] (58);
  
  \node[blue] at (2.75, -3.75) (63) {$m_3$} edge[-, red, line width=1.5pt] (59);
  \node[red] at (3.5, -3.75) (64) {$377$} edge[-, red, line width=1.5pt] (59);


\node at (0,0.3)  (x) {\footnotesize$18\otimes176 + 20\otimes322 + 36\otimes377$};

\node at (4.3,-1.25)  (x) {\footnotesize$(1\otimes 322, 2\otimes 377)$};

\node at (1.8,-0.3)  (x) {\footnotesize{$(5\otimes 322, 9\otimes 377)$}};

\node at (2.5,-2)  (x) {\footnotesize{$(1\otimes 322, 2\otimes 377)$}};

\node at (0.4,-1.3)  (x) {\footnotesize{$(3\otimes 322, 5\otimes 377)$}};

\node at (-2,-0.5)  (x) {\small{$(1\otimes 176)$}};

\node at (-2.65,-1.25)  (x) {\small{$(1)$}};

\end{tikzpicture}
\caption{A d-tree for the semimodule expression lineage of $Q_2$ from Table \ref{tab:lineages_lifted}. The blue part represents the boolean component of the expression while the red one represents a semimodule component.}
\label{fig:complete_semimodule_d-tree}
\end{figure}

%% file: Sections_Arxiv_New/Experiments.tex
\section{Experiments}
\label{sec:experiments}
This section details our experimental setup and results.

\subsection{Experimental Setup and Benchmarks}

We implemented all algorithms in Python 3.11 and performed experiments on a Linux Debian 14.04 machine with 1TB of RAM and an Intel(R) Xeon(R) Gold 6252 CPU @ 2.10GHz processor.

\paragraph*{Algorithms} We benchmarked variants of our algorithms \lexaban and \lexashap, with and without the two optimization techniques, for SPJU and aggregate queries. We compare them with the state-of-the-art solutions for Banzhaf and Shapley computation ~\cite{Sig24:BanzhafValuesForFactsInQueryAnswering} that only work for SPJU queries and do not use the variable lifting and gradient optimizations: \exaban and \exashap for exact computation of Banzhaf and Shapley values and \adaban for approximation of Banzhaf values with error guarantees (such approximation is not available for Shapley). \nop{\adaban was shown~\cite{Sig24:BanzhafValuesForFactsInQueryAnswering} to significantly faster and more accurate than a Monte-Carlo randomized approximation from prior work~\cite{ComputingTheShapleyValueOfFactsInQueryAnswering}.} For aggregate queries with SUM/COUNT, we also compared to a variant where we use \exaban for each boolean formula and then use linearity as explained in Section \ref{subsect: linear aggregates} to compute Banzhaf values\nop{with respect to the aggregate queries}.

 
\paragraph*{Datasets} The workload (see Table \ref{tab:db_stat}) of 301 queries over three datasets: Academic, IMDB and TPC-H (SF1) is based on prior work on  Banzhaf/Shapley values for query answering ~\cite{Sig24:BanzhafValuesForFactsInQueryAnswering,ComputingTheShapleyValueOfFactsInQueryAnswering}. Lineage for all output tuples of all queries was constructed using ProvSQL~\cite{senellart2018provsql}. We created two variants of each benchmark: one including SPJU queries, and one including aggregate queries. The SPJU variant includes all queries from IMDB and Academic as is (they do not include aggregate queries), and all queries without nested subqueries and with aggregates removed from TPC-H. The aggregate queries variant (named Academic$_{agg}$, IMDB$_{agg}$, TPC-H$_{agg}$) was constructed as follows. For TPC-H it includes 7 queries that were preserved as is, namely without removing aggregation. These are the queries for which ProvSQL \nop{managed to compute}computed all lineages in an hour or less. All 7 TPC-H aggregate queries have $\groupbyagg$ clauses, resulting in 140 lineage instances. To further extend the benchmark, we created multiple variants of each of these queries: for TPC-H we changed the aggregation to  $\maxagg,\countagg$, $\sumagg$ (the solution for $\minagg$ is almost identical to that of $\maxagg$); for IMDB and Academic, we created a synthetic aggregate variant with each of $\maxagg,\countagg$ and $\sumagg$ introduced to each query, applied over randomly chosen values. These datasets include very large and highly challenging lineage expressions. In total, the workload consists of nearly 1M Boolean lineage expressions and 429 lineage expressions for each aggregate.


\begin{table}
\caption{Statistics of the datasets used in the experiments.}
 
\centering
\begin{tabular}{| l | r | r |r| r|}
    \hline
        {\textbf{Dataset}} &  {\makecell{\textbf{\#Queries}}} &
         \textbf{\# Lineages}  & \textbf{\#Vars} & \textbf{\#Clauses} \\
         & & & \textbf{avg/max} & \textbf{avg/max} \\
         \hline
    \texttt{Academic} & 92 & 7,865 & 79 / 6,027 & 74 / 6,025\\
    \hline
    \texttt{IMDB} & 197 & 986,030 & 25 / 27,993 & 15 / 13,800 \\
    \hline
    \texttt{TPC-H} & 12 & 165 & 1,918 / 139,095 & 863 / 75,983\\
    \hline
    \texttt{TPC-H$_{agg}$} & 7 & 140 & 746 / 6959 & 233 / 2076\\
    \hline

\end{tabular}
\label{tab:db_stat}
\end{table}

\paragraph*{Measurements} We define an instance as the computation of the Banzhaf or Shapley values for all variables in a lineage of an output tuple of a query over a dataset. We report the success rate of each algorithm over the instances, where failure is declared in case an algorithm did not terminate an instance within one hour. We also report statistics of runtimes across all instances (average, median, maximal runtime, and percentiles). The  p$X$ columns in the tables give the execution times for the $X$-th percentile of the instances. Percentiles are computed with respect to the set of instances each algorithm succeeded on.  Since the algorithms success rate is different, these percentiles generally correspond to different instances.


\subsection{Banzhaf Computation for SPJU Queries}
\label{sec:exp-banzhaf-boolean}

\paragraph{Success rate}
Table~\ref{tab:successrate} reports the success rates of \lexaban, \exaban, and \adaban for queries without aggregates. On Academic, \lexaban succeeded on all instances. On IMDB, \lexaban topped \exaban and even the approximate solution \adaban in the lineage-level and achieved significantly greater query-level success (>95\%).
\lexaban succeeds on more than 95\% (90\%) of instances on which \exaban (\adaban) failed. For TPC-H, many queries exhibit hard lineage instances and so the query success rates are lower for all algorithms. \lexaban again tops both \exaban and \adaban in lineage success rate and is tied with \adaban in query success rate.

\begin{table}
\caption{Query success rate: Percentage of  queries for which the algorithms finish for all instances of a query. Lineage success rate:  Percentage of instances (over all queries) for which the algorithms finish. The timeout is one hour.}
 
\centering
\begin{tabular}{| c | r|r|r|}
\hline
\textbf{Dataset} & \textbf{Algorithm} & \textbf{Query Success Rate}
& \textbf{Lineage Success Rate} \\
\hline
{\multirow{5}{*}{{\texttt{Academic}}}}
& \lexaban  & 100.00\% & 100.00\%\\
& \exaban  & 98.85\% & 99.99\%\\
& \adaban  & 98.85\% & 99.99\%\\
\cdashline{2-4}
& \lexashap  & 97.82\% & 98.84\%\\
& \exashap  & 97.82\% & 98.84\%\\

\hline
{\multirow{5}{*}{{\texttt{IMDB}}}}
& $\textsc{LExaBan}$ & 95.43\% & 99.98\%  \\
& $\textsc{ExaBan}$ & 81.73\% & 99.63\%  \\
& $\textsc{AdaBan}$ & 87.82\% & 99.81\%  \\
\cdashline{2-4}
& \lexashap  & 87.31\% & 99.89\%\\
& \exashap  & 65.48\% & 99.53\%\\
\hline

{\multirow{5}{*}{{\texttt{TPC-H}}}}
& $\textsc{LExaBan}$ & 75.00\% & 93.33\%  \\
& $\textsc{ExaBan}$ & 58.33\% & 91.52\%  \\
& $\textsc{AdaBan}$ & 75.00\% & 92.73\%  \\
\cdashline{2-4}
& \lexashap  & 58.33\% & 91.52\%\\
& \exashap  & 50.00\% & 85.46\%\\

\hline
\end{tabular}

\label{tab:successrate}
\end{table}

\paragraph{Runtime performance}
Table~\ref{tab:total performance boolean} gives the runtime of \lexaban, \exaban and \adaban for the Academic, IMDB and TPC-H datasets. \lexaban outperforms both \nop{\exaban and \adaban}competitors significantly. Despite achieving success for more instances, it keeps a smaller mean runtime \nop{compared with} than \exaban and \adaban. For Academic, except for the unique instance where it \nop{achieves success}succeeds and the other algorithms fail on, runtime is below 5 seconds, \nop{which represents 119X improvement over \exaban, and 36X improvement over \adaban.}yielding 119X and 36X speedups over \exaban and \adaban, respectively. Similar trends are observed for the IMDB and TPC-H datasets, where \lexaban achieves 110X and respectively 106X improvements for the most expensive instance for \exaban (Table~\ref{tab:success_ExaBan}). 
For the instances on which \exaban failed (Table~\ref{tab:failure_ExaBan}), \lexaban still achieves success and low runtimes. 
For IMDB, this is an almost 95\%  success rate with a median runtime of 4.73 seconds, representing at least 761X speedup for more than half of those instances. For TPC-H, \lexaban succeeded on 3 more instances; for one such instance \lexaban needed 1.95 seconds, which means a speedup of at least 1846X over \exaban\nop{ which did not terminate in 1 hour}.

\begin{table}
\caption{Runtime of \lexaban, \exaban and \adaban. Note that the instance for Academic on which \lexaban is slowest is one on which \exaban and \adaban failed on.}

    \centering
    \begin{tabular}{| c | r | r| r r r r r r|}
        \hline
         \multirow{1}{*}{\textbf{Dataset}} &  \multirow{1}{*}{\textbf{Algorithm}} & \multirow{1}{*} {\makecell{\textbf{Success} \\ \textbf{rate}}} &
         \multicolumn{6}{c |}{\textbf{Execution times [sec]}} \\
         & & &\textbf{Mean} &  \textbf{p50} & \textbf{p90} & \textbf{p95} & \textbf{p99} & \textbf{Max} \\
         \hline

        {\multirow{3}{*}{{\texttt{Academic}}}}
         & $\textsc{LExaBan}$ & 100.00\% & 0.44 & 4E-4 &   
         2E-3 &
        5E-3 & 0.14 & 3455\\

         & $\textsc{ExaBan}$ & 99.99\% & 2.07 & 1E-3 &   
         0.01 &
        0.20 & 164.5 & 563.5\\

         & $\textsc{ADaBan}$ & 99.99\% & 0.76 & 1E-3 &  7E-3 &
        0.05 & 60.05 & 173.7\\

        \hline

        \multirow{3}{*}{\texttt{IMDB}} 
        & $\textsc{LExaBan}$ & 99.98\% & 0.18 & 1E-3 & 
        5E-3 & 0.01 & 0.09 & 
        1219\\  

        & $\textsc{ExaBan}$ & 99.63\% & 1.58 & 2E-3  &
        0.02 & 0.08 & 10.37 & 
        1793\\  

        & $\textsc{ADaBan}$ & 99.81\% & 1.82 & 1E-3 & 
        0.01 & 0.05 & 7.81 & 
        1802\\
                
        \hline

        \multirow{3}{*}{\texttt{TPC-H}} 
        & $\textsc{LExaBan}$ & 93.33\% & 0.12 & 7E-4 & 
        3E-3 & 0.62 & 1.95 & 
        8.73\\ 
        & $\textsc{ExaBan}$ & 91.52\% & 4.23 & 0.89 & 
        0.94 & 51.05 & 61.98 & 
        69.18\\
        & $\textsc{ADaBan}$ & 92.73\% & 2.37 & 3E-3 & 
        0.14 & 3.15 & 81.55 & 
        166.3\\
        \hline

    \end{tabular}%
    
    \label{tab:total performance boolean}
\end{table}

\begin{table}
\caption{Runtime of \lexaban and \exaban on instances for which \exaban succeeded.}

    \centering
    \begin{tabular}{| c | c| c c c c c c c|}
        \hline
         \multirow{1}{*}{\textbf{Dataset}} &  \multirow{1}{*}{\textbf{algorithm}} &
         \multicolumn{7}{c |}{\textbf{Execution times [sec]}} \\
          & &\textbf{Mean} &  \textbf{p50} & \textbf{p75} & \textbf{p90} & \textbf{p95} & \textbf{p99} & \textbf{Max} \\
         \hline

        \multirow{2}{*}{\texttt{Academic}}
         & $\textsc{LExaBan}$  & 4E-3 & 4E-4 &   
        1E-3 & 2E-3 &
        5E-3 & 0.14 & 4.99 \\

         & $\textsc{ExaBan}$  & 2.07 & 1E-3 &   
        2E-3 & 0.01 &
        0.20 & 164.5 & 563.5\\

        \hline

        \multirow{2}{*}{\texttt{IMDB}} 
        & $\textsc{LExaBan}$ & 7E-3 & 1E-3 & 2E-3 &
        4E-3 & 9E-3 & 0.05 & 
        16.26\\ 
        
        & $\textsc{ExaBan}$ & 1.58 & 2E-3 & 3E-3 &
        0.02 & 0.08 & 10.37 & 
        1793\\
        
        \hline

                \multirow{2}{*}{\texttt{TPC-H}}
        & $\textsc{LExaBan}$ & 0.04 & 7E-4 & 8E-4 &
        1E-3 & 0.53 & 0.65 & 
        0.65\\

        & $\textsc{ExaBan}$ & 4.23 & 0.89 & 0.93 &
        0.94 & 51.05 & 61.98 & 
        69.18\\ 
                
        \hline
    \end{tabular}%
    
    \label{tab:success_ExaBan}
\end{table}

\begin{table}
\caption{Runtime of \lexaban on instances for which \exaban failed. There was one such instance in Academic.}

    \centering
    \begin{tabular}{| c | c| c c c c c c c|}
        \hline
         \multirow{1}{*}{\textbf{Dataset}} &  \multirow{1}{*} {\makecell{\textbf{Success} \\ \textbf{rate}}} &
         \multicolumn{7}{c |}{\textbf{Execution times [sec]}} \\
          & &\textbf{Mean} &  \textbf{p50} & \textbf{p75} & \textbf{p90} & \textbf{p95} & \textbf{p99} & \textbf{Max} \\
         \hline

        \multirow{1}{*}{\texttt{Academic}}
         & 100\%  & 3455 & - &   
        - & - &
        - & - & 3455 \\

        \hline

        \multirow{1}{*}{\texttt{IMDB}} 
        & 94.74\% & 48.97 & 4.73 & 22.88 &
        168.2 & 375.9 & 448.2 & 
        1219\\ 
        
        \hline
        
        \multirow{1}{*}{\texttt{TPC-H}} 
        & 21.43\% & 4.21 & 1.95 & 5.34 &
        7.38 & 8.05 & 8.60 & 
        8.73\\ 
        \hline

    \end{tabular}%
    
    \label{tab:failure_ExaBan} 
\end{table}

\subsection{Banzhaf Computation for Aggregate Queries}

Table~\ref{tab:TPCHAggregateTimes} reports the results for \lexaban on aggregate queries on TPC-H$_{agg}$. The performance is excellent for SUM and COUNT aggregates (under 1 msec for each instance). \lexaban is also fast for most instances with the MAX aggregate. \lexaban without optimizations succeeds in only $86.43\%$ of the instances, compared to $93.57\%$ success rate when using the optimizations. When the unoptimized version succeeds, it is still slightly slower than \lexaban, though both algorithms perform excellently (under $0.26$ sec).

Table \ref{tab:SumCountTimes_regular_dbs} shows the results for the synthetic and challenging IMDB$_{agg}$ and Academic$_{agg}$ datasets. We again observe good performance for SUM (and COUNT, for which the results are very similar and omitted). Specifically, \lexaban succeeds on all SUM query instances for Academic$_{agg}$ and $95.43 \%$ of query instances for IMDB$_{agg}$. It is clearly superior to the variant that uses \exaban for the individual Boolean formulas in the semimodule expressions. For $\maxagg$, there is no baseline to compare against and \lexaban is  successful in 78\% (61\%) of the instances for Academic$_{agg}$ (IMDB$_{agg}$). For both datasets, it finishes in less than $0.1$ sec for over 50\% of the lineage instances.
Figure~\ref{fig:Max_run_time_over_lineitem_facts} shows the average  runtime for all lineage instances of a representative hard query (TPC-H Query 5) when the size of the  \texttt{Lineitem} relation is varied. \lexaban needs about 1.6 sec on average even for the 3M facts in \texttt{Lineitem}, compared to over 10 min for a ``naive" version 
without optimizations.


\begin{table}
\caption{Success rate and runtime of \lexaban for the TPC-H$_{agg}$ aggregate queries. }
    \centering
    \begin{tabular}{|c| c | c c c c c c c |}
        \hline
          \multirow{1}{*}{\textbf{Aggregate}} & \multirow{1}{*} {\makecell{\textbf{Success} \\ \textbf{rate}}} &
         \multicolumn{7}{c |}{\textbf{Execution times [sec]}} \\
          & &\textbf{Mean} &  \textbf{p50} & \textbf{p75} & \textbf{p90} & \textbf{p95} & \textbf{p99} & \textbf{Max} \\
         \hline

        $\textsc{SUM}$ & 100\% & 2E-4 & 2E-6 &   
        2E-6 & 1E-3 &
        2E-3 & 2E-3 & 2E-3\\

        \hline

        $\textsc{COUNT}$ & 100\% & 2E-4 & 1E-6 & 2E-6 &
        1E-3 & 2E-3 & 2E-3 & 
        2E-3\\
        \hline

        $\textsc{MAX}$ & 93.57\% & 224.8 & 2E-3 & 3E-3 &
        5E-3 & 2990 & 3349 & 
        3517\\

        \hline

    \end{tabular}%
    \label{tab:TPCHAggregateTimes}
\end{table}




\begin{table}
\caption{Success rate and runtime of \lexaban and \exaban for queries with SUM aggregate (similar results for COUNT aggregate are not shown) and MAX aggregates.}

    \centering
    \begin{tabular}{| c | r | r| r r r r r r|}
        \hline
         \multirow{1}{*}{\textbf{Dataset}} &  
         \multirow{1}{*}{\textbf{Algorithm}} &
         \multirow{1}{*} {\makecell{\textbf{Success} \\ \textbf{rate}}} &
         \multicolumn{6}{c |}{\textbf{Execution times [sec]}} \\
         & 
         \textbf{\& Aggregate} & &\textbf{Mean} &  \textbf{p50} & \textbf{p90} & \textbf{p95} & \textbf{p99} & \textbf{Max} \\
         \hline

        \multirow{3}{*}{\shortstack{{\texttt{Academic}} \\ $_{agg}$}}

         & 
         \textsc{LExaBan}$_{sum}$ & 
         100\% & 40.1 & 0.08 &   
        0.48 & 0.58 &
        496 & 3561\\

         & 
         \textsc{ExaBan}$_{sum}$ & 
         97.70\% & 6.20 & 0.09 &    
        1.88 & 12.6 &
        105 & 379.4\\

         & 
         \textsc{LExaBan}$_{\max}$ & 
         78.26\% & 15.11 & 3E-3 &    
        37.04 & 104.5 &
        253.4 & 289.1\\
        \hline

        \multirow{3}{*}{\shortstack{{\texttt{IMDB}} \\ $_{agg}$}} 
        & \textsc{LExaBan}$_{sum}$ 
        & 95.43\%  & 37.54 & 0.06 &    
        9.77 & 71.56 &
        774 & 2450\\

        & \textsc{ExaBan}$_{sum}$ & 81.19\% & 62.7  & 0.27 & 
        65.8 & 117 &
        1767 & 3286\\

        & 
         \textsc{LExaBan}$_{\max}$ & 
         61.93\% & 52.25 & 0.09 &   
        23.79 & 242.6 &
        734.9 & 2407\\
        
        \hline
    \end{tabular}%
    
    \label{tab:SumCountTimes_regular_dbs}
\end{table}




\begin{figure}
    \centering
    \includegraphics[width=0.75\linewidth]{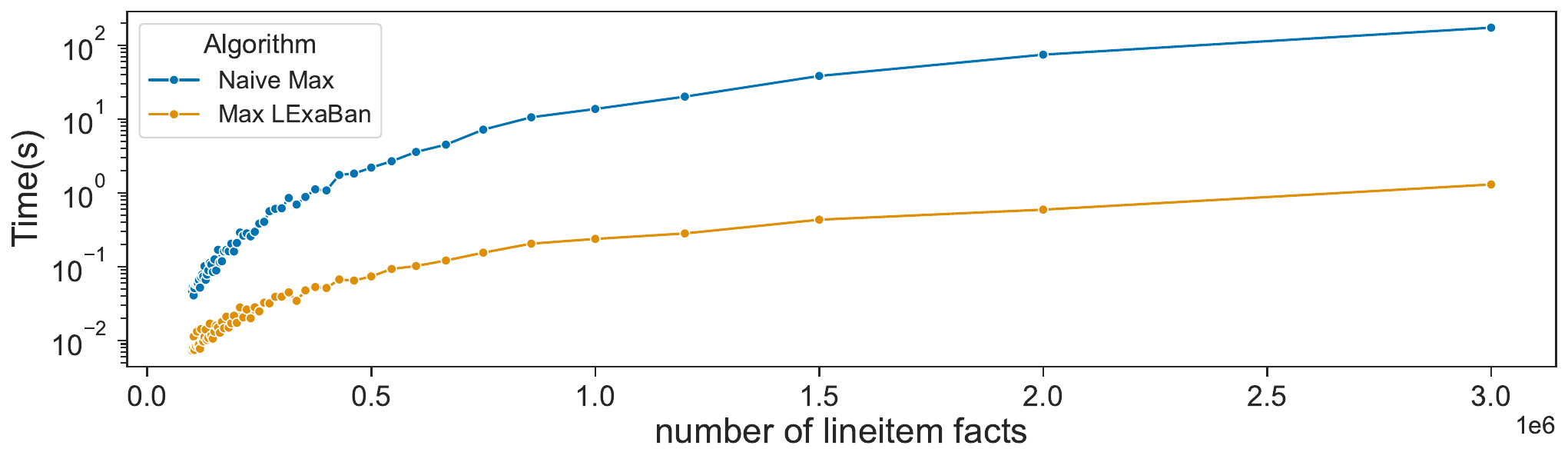}
    \caption{Runtime of \lexaban for a TPC-H query 5 with MAX aggregate, when varying the number of \texttt{Lineitem} facts.}
    \label{fig:Max_run_time_over_lineitem_facts}
\end{figure}


\subsection{Effect of Lifting and Gradient Techniques}

Figure~\ref{fig:effect-lifting} shows that variable lifting can speed-up lineage compilation by over one order of magnitude and reduce the d-tree size by over two orders of magnitude. These benefits increase as the lineage size grows. This experiment includes all Boolean queries that can be processed within 10 min.
Figure~\ref{fig: gradient_benefits} shows the success rate and runtime of \lexaban with and without the gradient technique for  Boolean queries lineages. The technique is beneficial for all instances; the benefit increases with the instance size. For instances with over 1600 variables, the average speedup is more than 2 orders of magnitude.
\begin{figure}[t]
     \centering
    \begin{subfigure}{0.48\linewidth}
     \includegraphics[width=\linewidth]{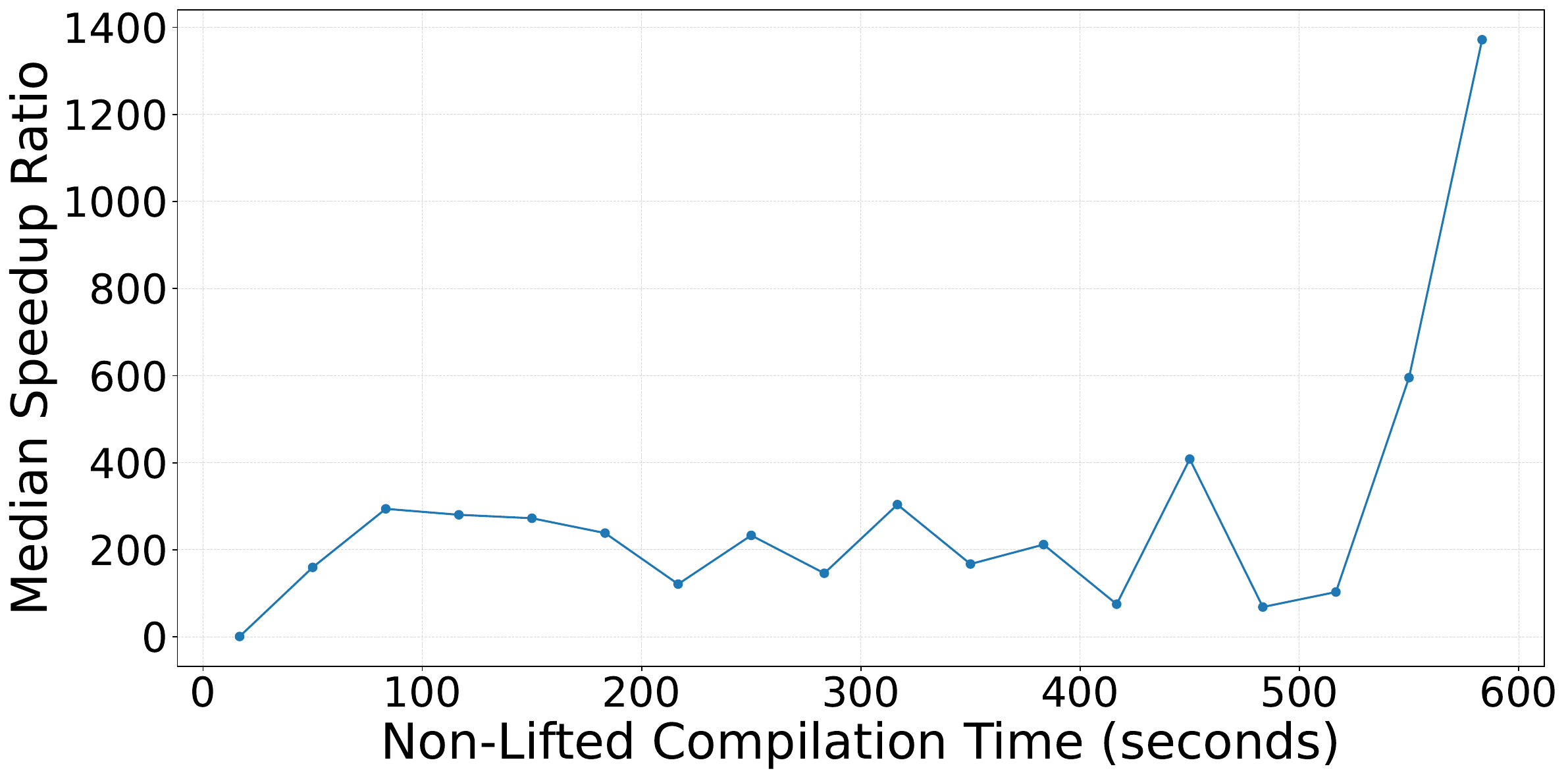}
     \caption{D-tree compilation time with and without variable lifting.}
     \label{fig: Lifting_time}
     \end{subfigure}
        \begin{subfigure}{0.48\linewidth}
     \includegraphics[width=\linewidth]
     {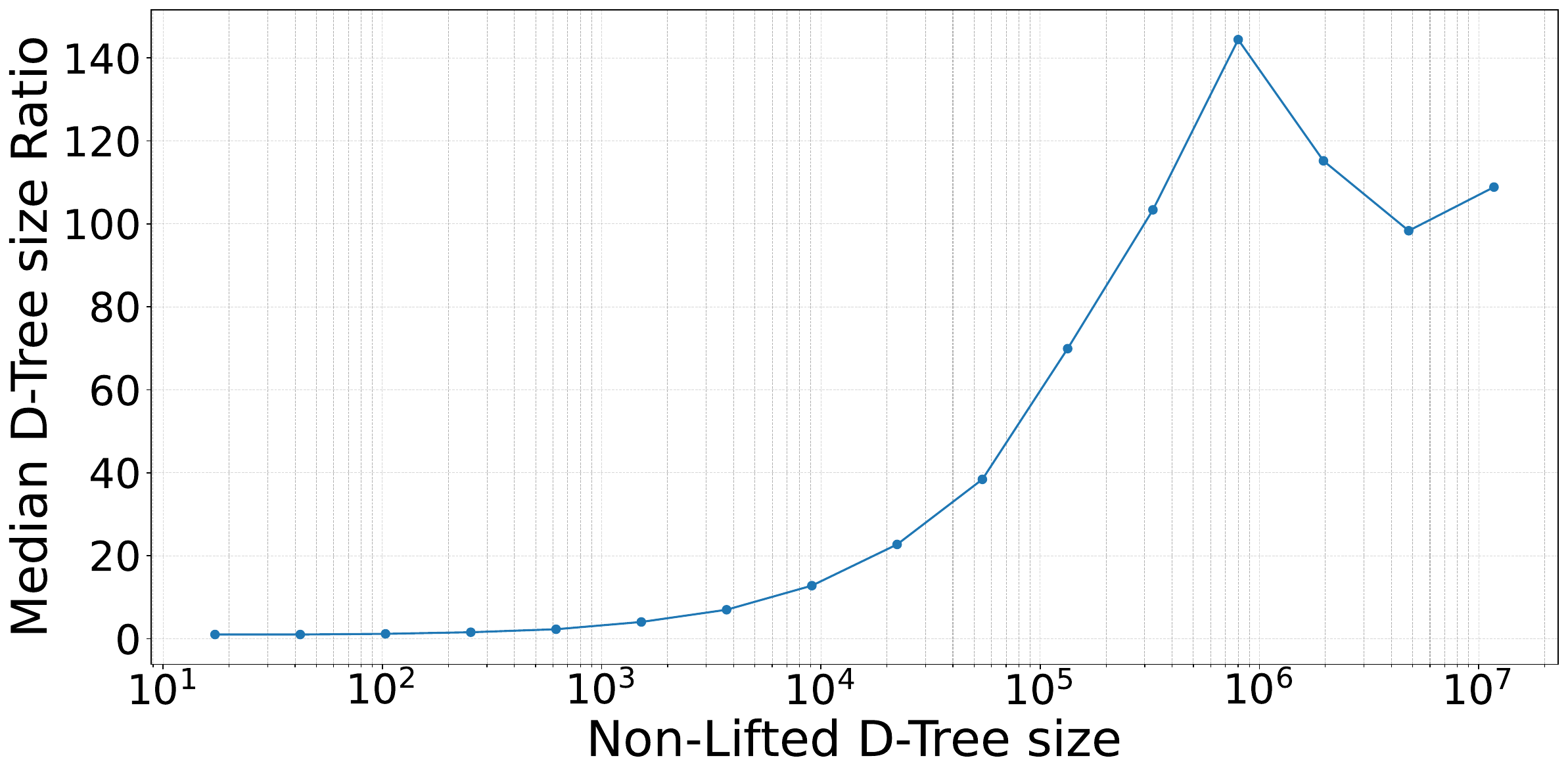}
     \caption{D-tree size with and without variable lifting.}
     \label{fig: Lifting_size}
        \end{subfigure}

        \caption{Effect of variable lifting on \lexaban's compilation of lineage of Boolean queries into d-trees.}
        \label{fig:effect-lifting}
\end{figure}

\begin{figure}[t]
     \centering
     \includegraphics[width=0.48\linewidth]{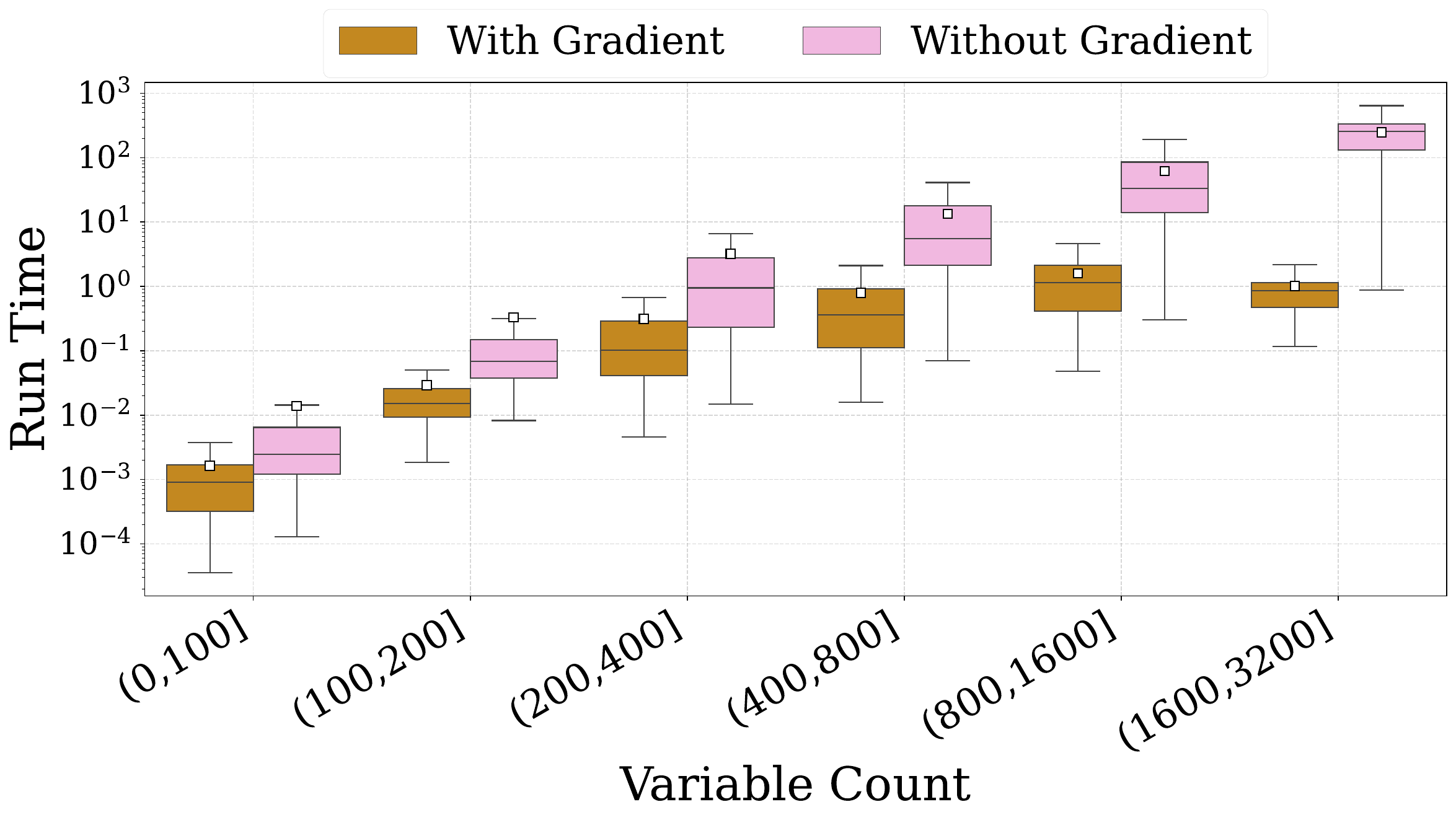}
          \includegraphics[width=0.48\linewidth]{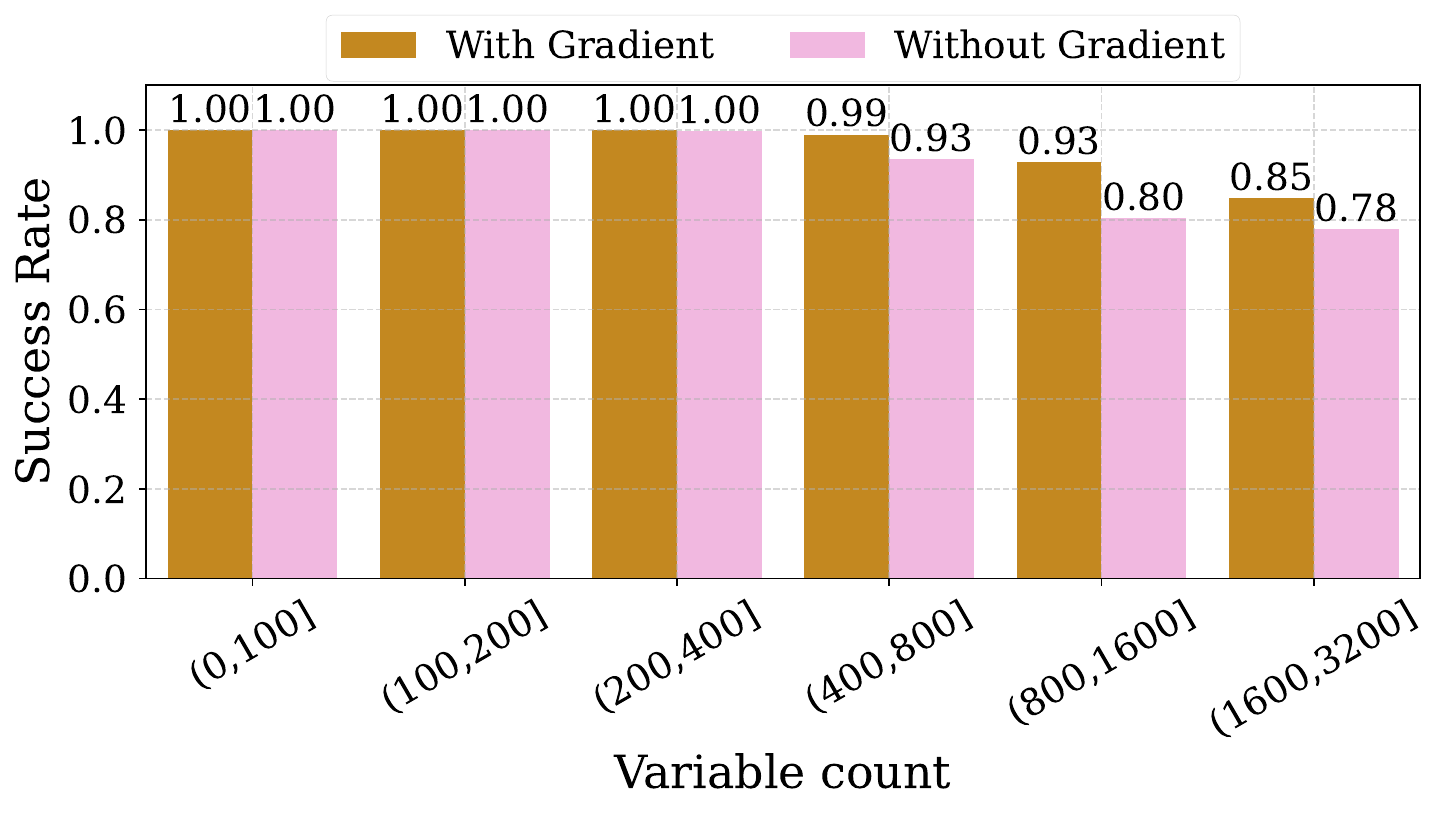}

     \caption{Effect of the gradient technique on the success rate and runtime of \lexaban for lineage of Boolean queries.}
     \label{fig: gradient_benefits}

\end{figure}

\nop{One might expect a speedup of at most $n$ for instances with $n$ variables, since the gradient technique allows to compute the Banzhaf value simultaneously for all $n$ variables. However, not all branches of the d-tree participate in the Banzhaf computation for each variable, and thus there is potential for significant computation sharing across variables.
}

\subsection{Shapley Computation}

Table \ref{tab:successrate} shows that our extension of \lexaban to compute the Shapley value achieves a higher query and lineage success rates than the state-of-the-art \exashap ~\cite{Sig24:BanzhafValuesForFactsInQueryAnswering}.
Table~\ref{tab:Shapley_ExaShap_success} shows the runtimes of \lexashap and \exashap on instances where \exashap succeeds. Table~\ref{tab:Shapley_ExaShap_failure} shows \lexashap's success rates and runtimes on instances where \exashap fails. \lexashap's speedup over \exashap is over an order of magnitude on average. \lexashap also has a much higher success rate: 79\% for the IMDB instances on which \exashap fails, with a median time of 13.04 sec; and 41\% for the TPC-H instances on which \exashap fails, with a median runtime of 23.94 sec, yielding speedups of at least 276X and 150X, respectively.

\begin{table}
\caption{Runtimes for \lexashap and \exashap on instances for which \exashap succeeds.}
     
    \centering
    \begin{tabular}{| c | r| r r r r r r r|}
        \hline
         \multirow{1}{*}{\textbf{Dataset}} &  \multirow{1}{*}{\textbf{Algorithm}} &
         \multicolumn{7}{c |}{\textbf{Execution times [sec]}} \\
          & &\textbf{Mean} &  \textbf{p50} & \textbf{p75} & \textbf{p90} & \textbf{p95} & \textbf{p99} & \textbf{Max} \\
         \hline

        \multirow{2}{*}{\texttt{Academic}}
         & $\textsc{LExaShap}$ & 0.02 & 2E-3 &   
        4E-3 & 9E-3 &
        0.01 & 0.05 & 48.40\\

         & $\textsc{ExaShap}$ & 0.45 & 1E-3 &   
        3E-3 & 0.04 &
        0.22 & 0.37 & 1397\\

        \hline

        \multirow{2}{*}{\texttt{IMDB}} 
        & $\textsc{LExaShap}$ & 0.07 & 6E-3 & 0.02 &
        0.05 & 0.10 & 0.57 & 
        520.3\\

        & $\textsc{ExaShap}$ & 3.35 & 3E-3 &   
        0.01 & 0.12 &
        0.79 & 44.00 & 3574\\  
        
        \hline

        \multirow{2}{*}{\texttt{TPCH}} 
        & $\textsc{LExaShap}$ & 2E-3 & 2E-3 &
        4E-3 & 5E-3 & 5E-3 & 
        5E-3 & 5E-3\\

        & $\textsc{ExaShap}$ & 7E-3 & 3E-3 &   
        3E-3 & 5E-3 &
        9E-3 & 0.02 & 0.56\\  
        
        \hline
    \end{tabular}%
    
    \label{tab:Shapley_ExaShap_success}   
\end{table}

\begin{table}
\caption{Runtimes for \lexashap on instances for which \exashap fails.}

    \label{tab:Shapley_ExaShap_failure}   
     
    \centering
    \begin{tabular}{| c | c | c| c c c c c c|}
        \hline
         \multirow{1}{*}{\textbf{Dataset}} &  
         \multirow{1}{*} {\makecell{\textbf{Success} \\ \textbf{rate}}} &
         \multicolumn{7}{c |}{\textbf{Execution times [sec]}} \\
         & &\textbf{Mean} &  \textbf{p50} & \textbf{p75} & \textbf{p90} & \textbf{p95} & \textbf{p99} & \textbf{Max} \\
         \hline

        \multirow{1}{*}{\texttt{IMDB}} 
        & 79.13\% & 61.54 & 13.04 & 59.41 &
        171.8 & 295.7 & 689.3 & 
        2835\\

        \multirow{1}{*}{\texttt{TPCH}} 
        & 41.68\% & 22.90 & 23.94 & 26.82 &
        27.89 & 29.02 & 29.92 & 
        30.14\\

        \hline
    \end{tabular}%
\end{table}

\subsection{Memory Consumption}

Table~\ref{tab:memory} shows the memory consumption\footnote{This is measured as peak resident memory, sampled throughout execution using a background monitoring thread that uses the \texttt{psutil} Python library.} of \lexaban, \exaban, \lexashap, and \exashap on representative instances for which \exaban succeeds, with d-tree sizes spanning different percentiles of the overall distribution. For 99\% of these instances, \lexaban uses less than 15MB and up to 37\% less memory than \exaban. For the largest instance solved by \exaban, \exaban uses over 7.3GB, while \lexaban uses only 109MB, 69x less. \nop{On harder instances where \exaban fails to terminate, \lexaban still succeeds, with memory usage reaching up to 120GB.} \lexashap incurs an overhead that is independent of the problem size and due to the \texttt{scipy.signal} library that handles the gradient (not needed by the other algorithms). For the largest instance solved by \exaban, \lexashap uses 268MB, while \exashap fails.

\begin{table}
\caption{Memory usage for various percentiles of d-tree sizes on instances (all datasets) for which \exaban succeeds.}
 
\centering
\begin{tabular}{| c | r r r r r r|}
    \hline
    \multirow{2}{*}{\textbf{Algorithm}} & \multicolumn{6}{c|}{\textbf{Memory consumption (MB)}} \\
     & \textbf{p25} & \textbf{p50} & \textbf{p75} & \textbf{p90} & \textbf{p99} & \textbf{Max} \\
    \hline
    $\textsc{LExaBan}$ & 11.65 & 11.71 & 11.77 & 11.90 & 13.31 & 109.67 \\
    $\textsc{ExaBan}$ & 18.54 & 18.71 & 18.82 & 19.06 & 21.13 & 7516.16 \\
    \cdashline{1-7}
    $\textsc{LExaShap}$ & 74.85 & 75.05 & 75.25 & 75.70 & 77.20 & 268.10 \\
    $\textsc{ExaShap}$ & 18.52 & 18.62 & 18.87 & 18.94 & 21.45 & - \\
    \hline

\end{tabular}
\label{tab:memory}
\end{table}

\subsection{Effect of Lineage Structure}

Figure~\ref{fig:success_rate_exact_clauses_ratio} shows the effect of the number of variables and clauses in the lineage on the performance of \lexaban and \exaban. 
\lexaban's success rate is consistently higher and its runtime increases slower than for \exaban.
\nop{Lifting helps \lexaban compile the lineage faster and into a smaller d-tree, while the gradient helps \lexaban share the computation among many variables.}Figure \ref{fig: parameter_test}  presents the effect of two measurements of lineage connectivity: (1) h-index \cite{hindex_pnas}, namely the maximal $k$
so that at least $k$ variables occur at least $k$ times in the lineage; (2) 
maximal connected component size of a graph whose nodes are variables and edges reflect variable co-occurrence in a clause. All algorithms slow down as connectivity grows (fluctuations are due to small sample sizes for some values of these measurements), yet \lexaban and \lexashap handle highly interconnected lineages much better than \exaban and \exashap. \exaban and \exashap do not terminate for lineages with h-index greater than 53, whereas \lexaban and \lexashap can handle lineages with h-index of 120.

 \begin{figure}[]
     \centering
     \begin{subfigure}{\linewidth}
     \includegraphics[width=0.5\linewidth]{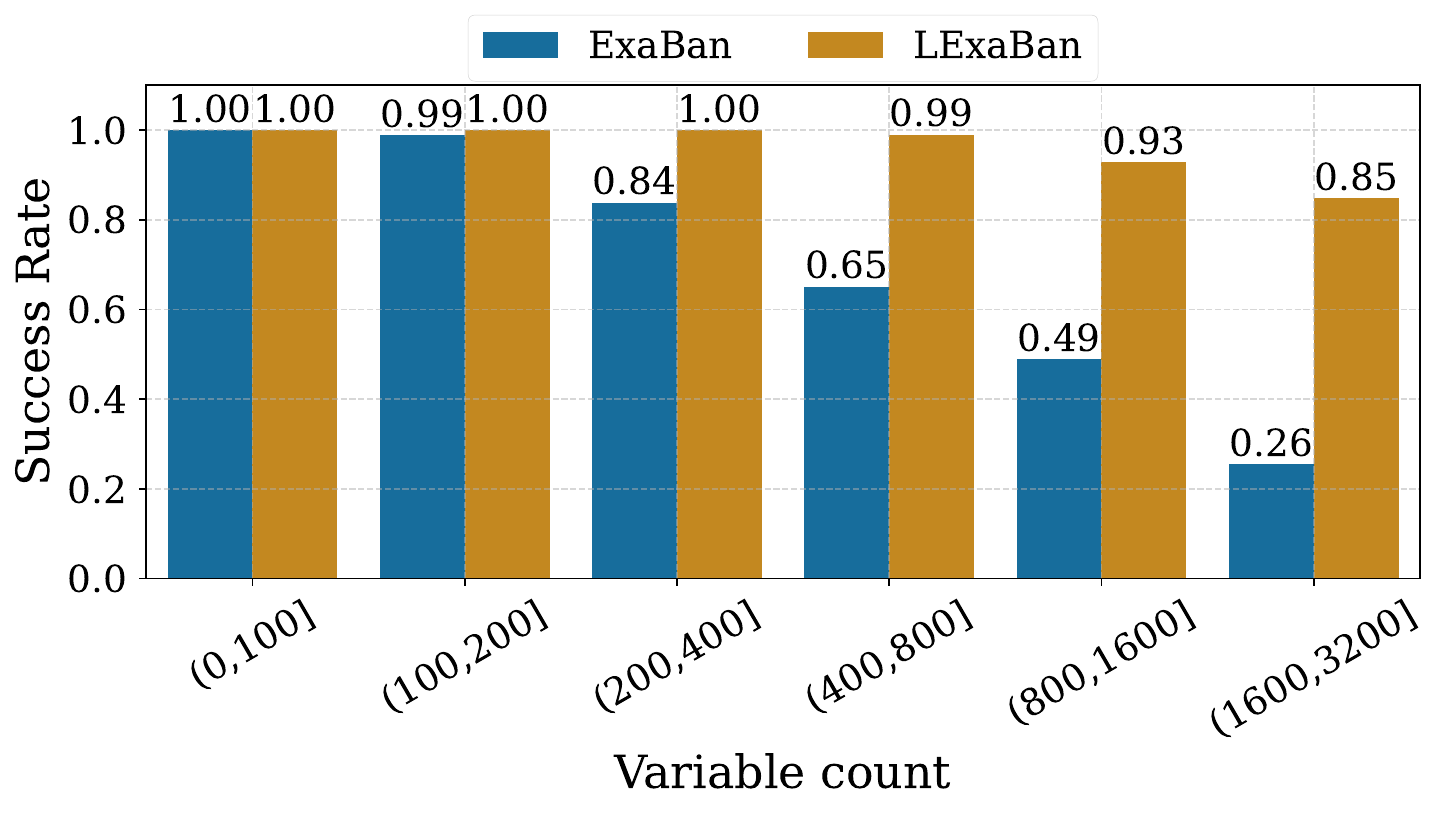}%
     \includegraphics[width=0.5\linewidth]{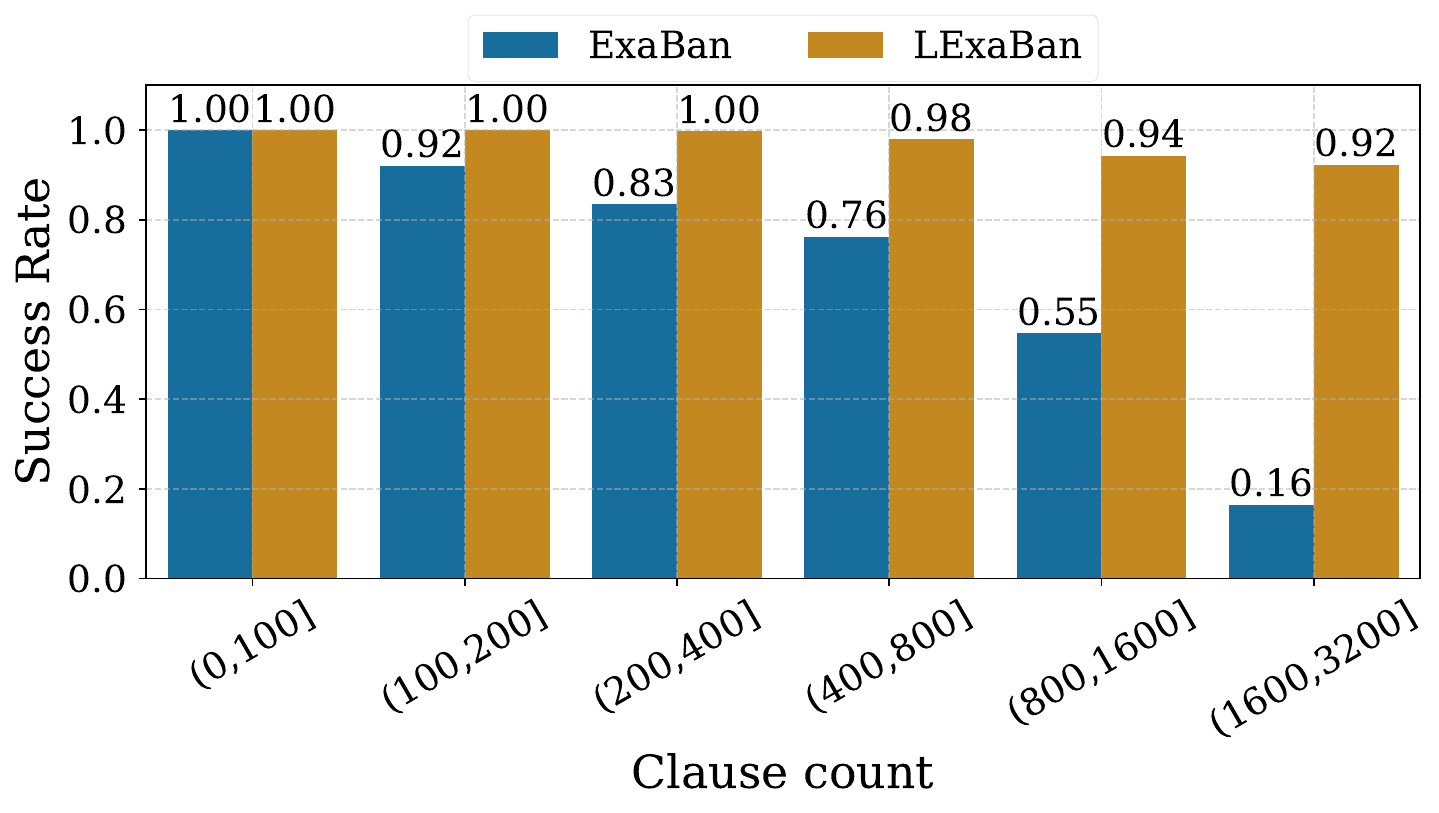}
     \caption{Success rate (average over all instances in each group)}
     \label{fig:success_rate_exact_tuples}
     \end{subfigure}
     \hfill
     \begin{subfigure}{\linewidth}
         \includegraphics[width=0.49\linewidth]{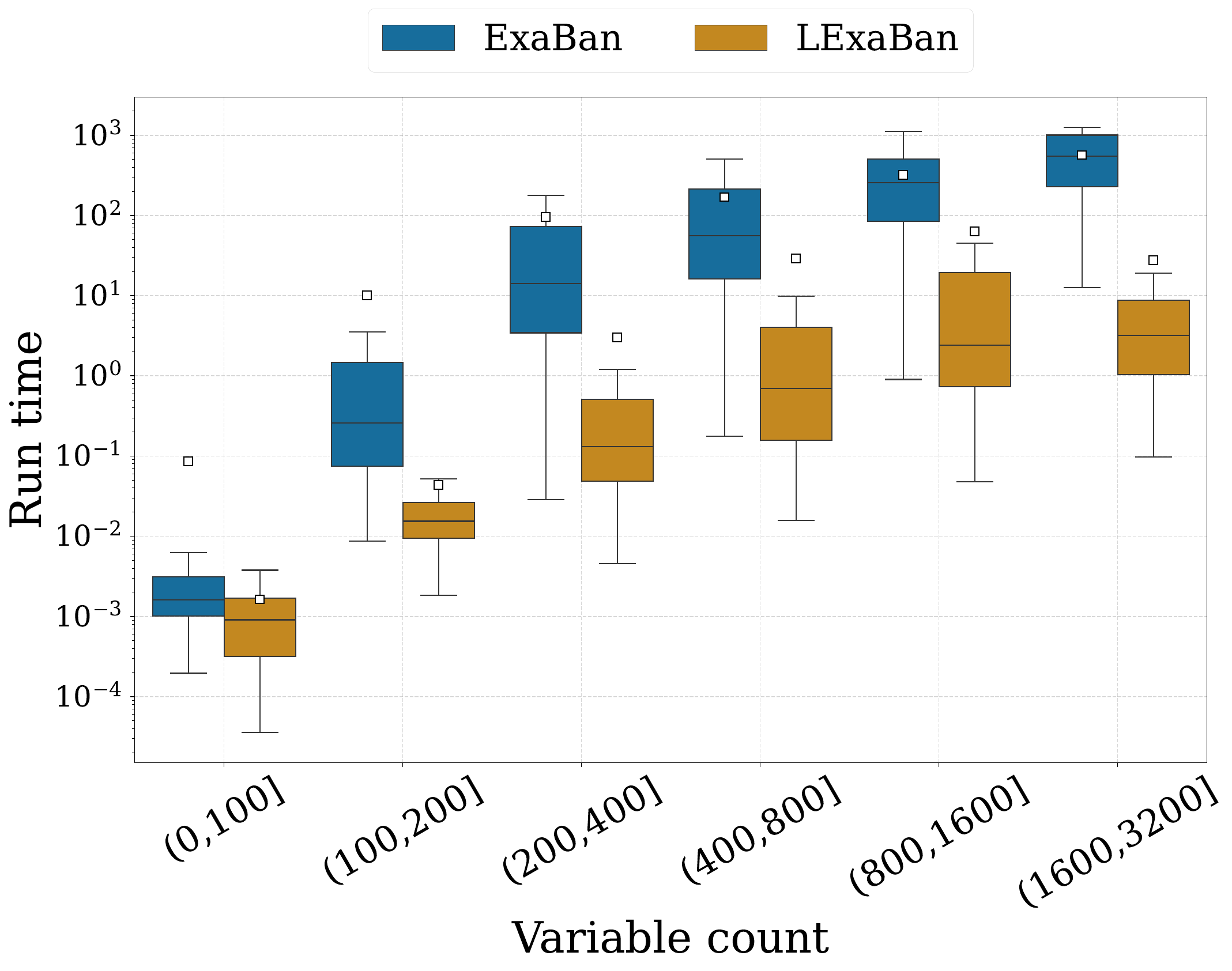}
         \includegraphics[width=0.49\linewidth]{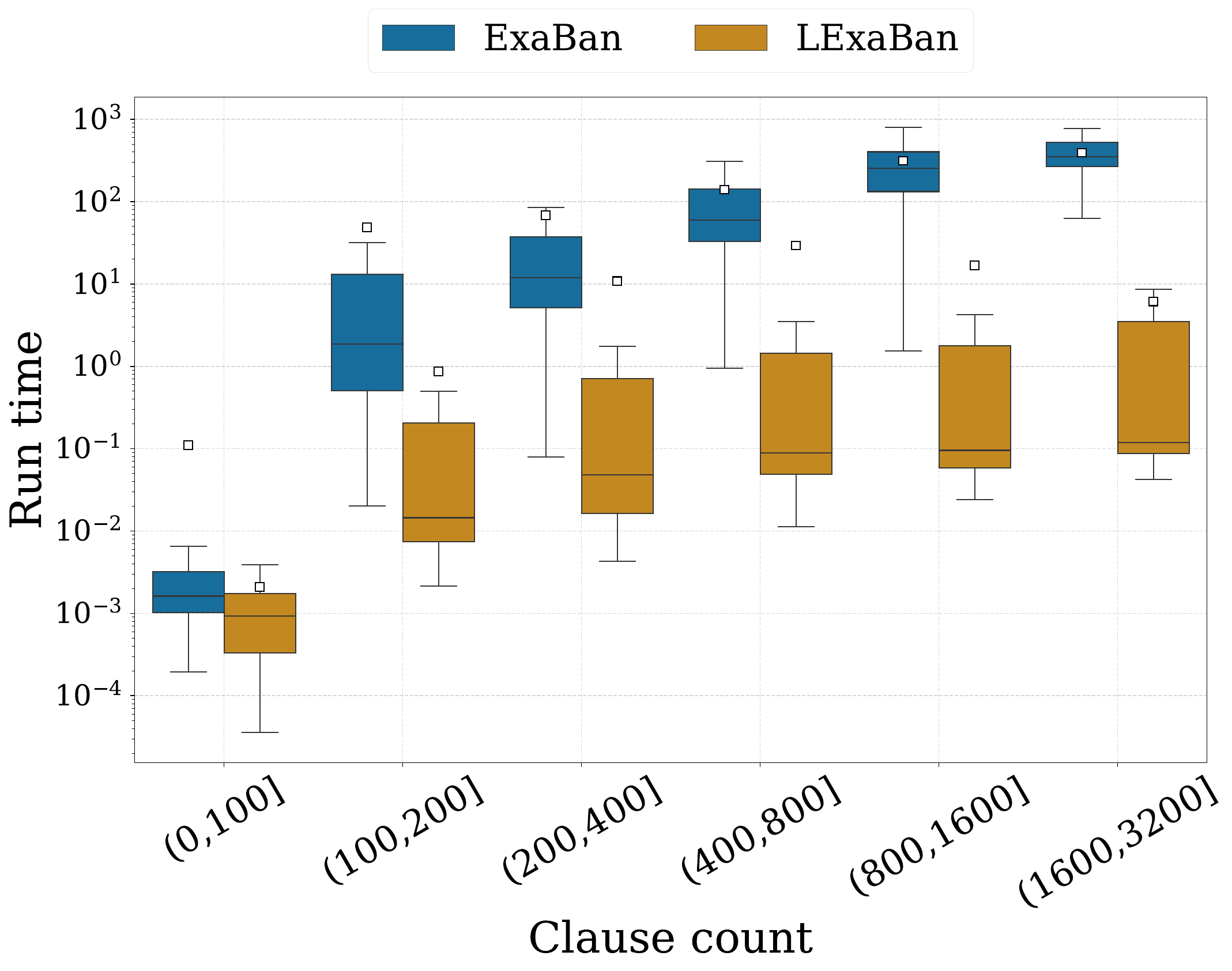}
     \caption{Runtime (ranges over all instances in each group)}
     \label{fig:run_time_exact_tuples}
     \end{subfigure}
     \caption{Performance of \lexaban and \exaban for all lineages, grouped by number of variables/clauses. $[i,j]$ represents lineages with \#vars (\# clauses) between $i$ and $j$.}

     \label{fig:success_rate_exact_clauses_ratio}
 \end{figure}

\subsection{Breakdown of Total Runtime}
We conclude with a breakdown of the runtime, including the time to compute the output tuples (query execution), the lineage via ProvSQL, and the Banzhaf values for all datasets (Figure~\ref{fig:runtime_breakdown_queries}). As queries may produce many output tuples and ProvSQL outputs lineage for each of them, we break down \lexaban runtime by percentiles of the execution time w.r.t. output tuples. \lexaban usually incurs a small overhead: for 81\% (85\%) of queries, attribution for all input tuples w.r.t any output tuple is computed faster than query execution (resp. query execution and lineage computation, combined). For 7\% of queries, attribution for at least one output tuple is 2 OOM slower than query execution. Overall, the total explanation time (lineage + Banzhaf) exceeds the query execution time in only 27\% of queries, and exceeds it by more than 10x in 13\% of queries. For \lexashap (Figure~\ref{fig:runtime_breakdown_queries_shap}), the overhead is typically higher: only 47\% (51\%) of queries complete attribution for all output tuples faster than query execution (resp. query execution + lineage). In 19\% of queries, at least one output tuple incurs a 2 OOM slowdown. Overall, explanation time exceeds query execution in 57\% of queries and exceeds it by over 10x in 25\%.

\begin{figure}
     \centering
     \includegraphics[width=0.48\linewidth]{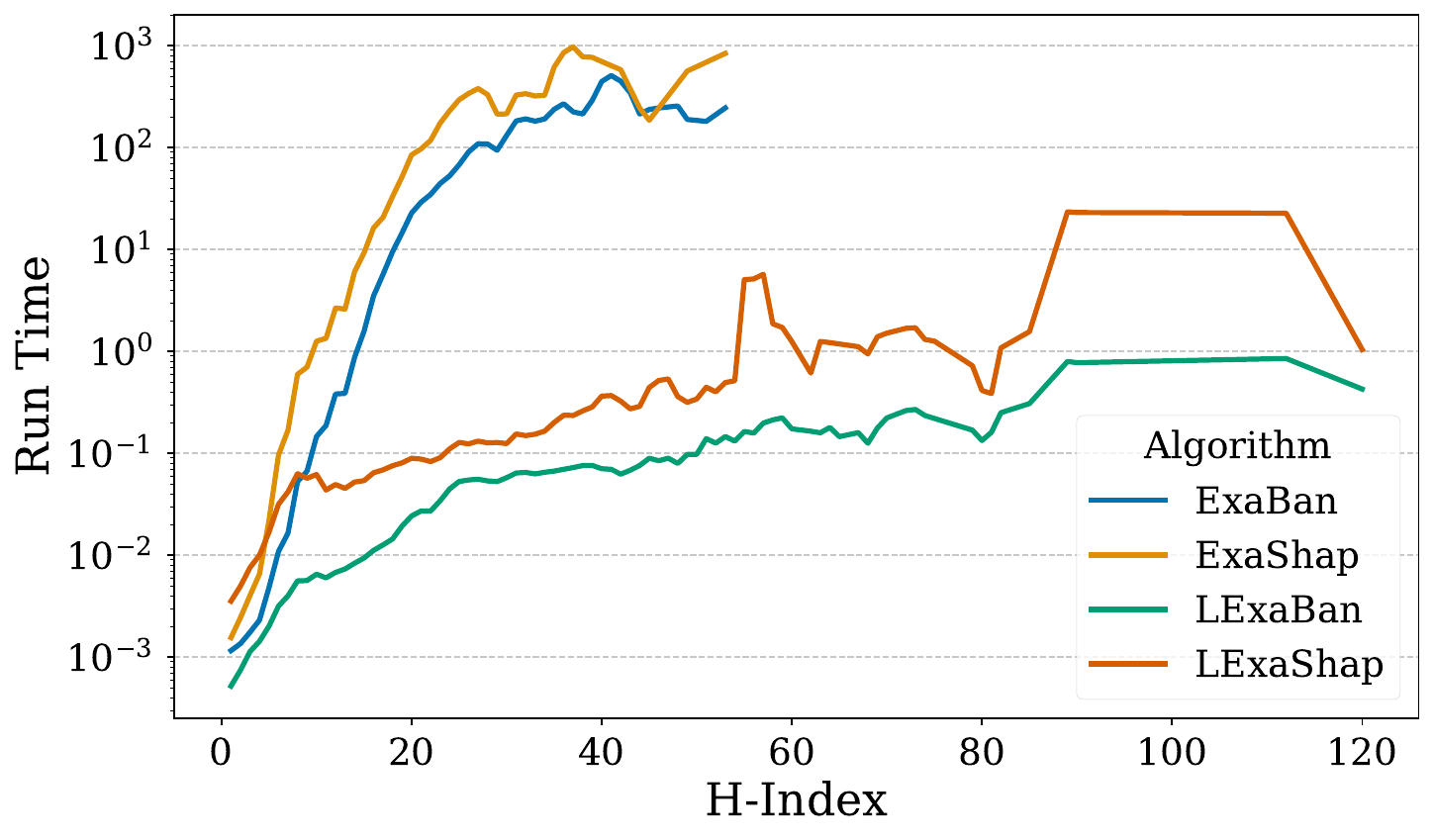}
          \includegraphics[width=0.48\linewidth]{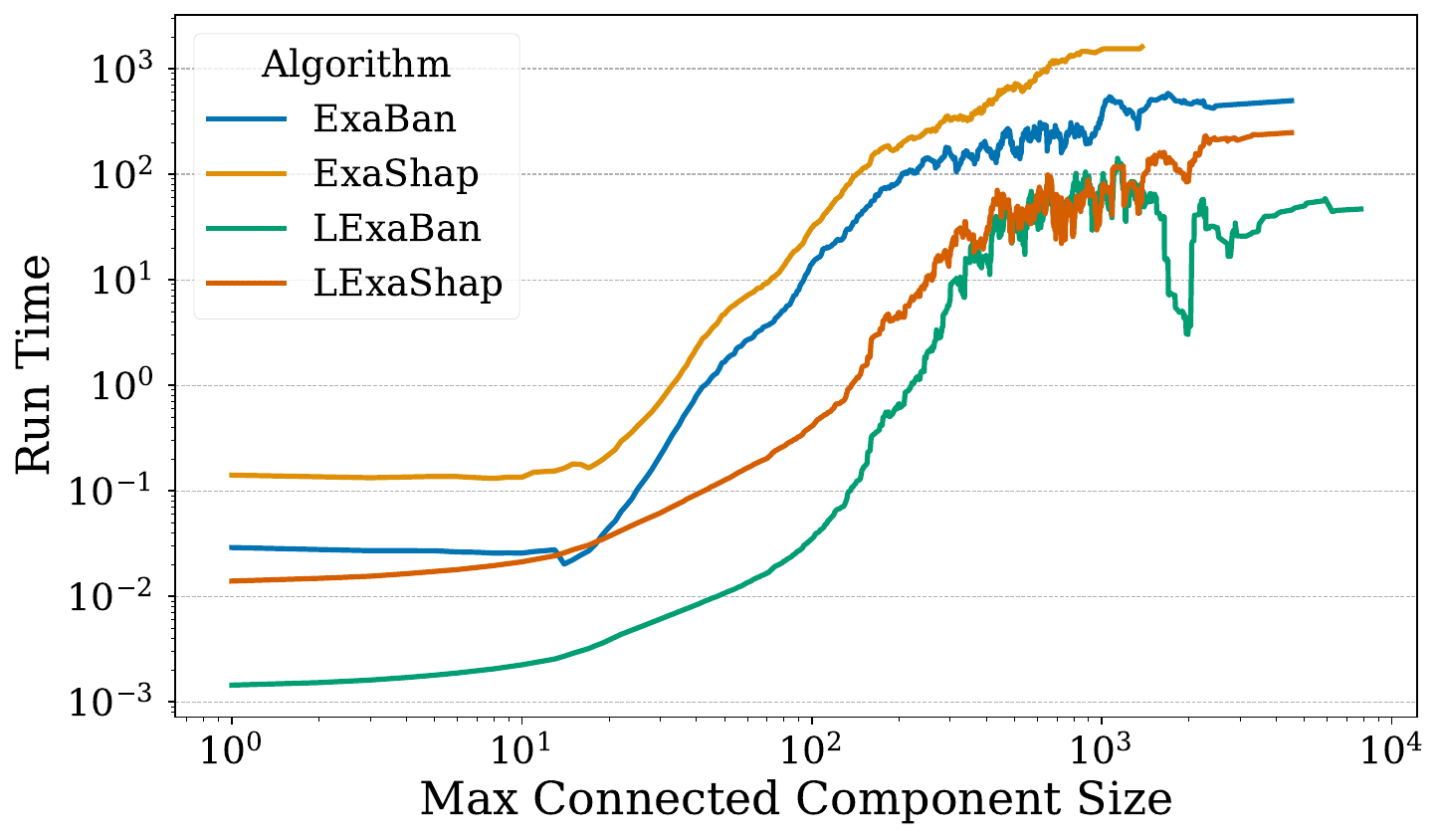}
     \caption{Effect of the H-index and maximum component size of the lineage on the algorithms' execution time. Both metrics were smoothed using a moving average, with window sizes 3 (H-index) and 25 (component size).}
     \label{fig: parameter_test}

\end{figure}


\begin{figure}
    \centering
    \includegraphics[width=\linewidth]{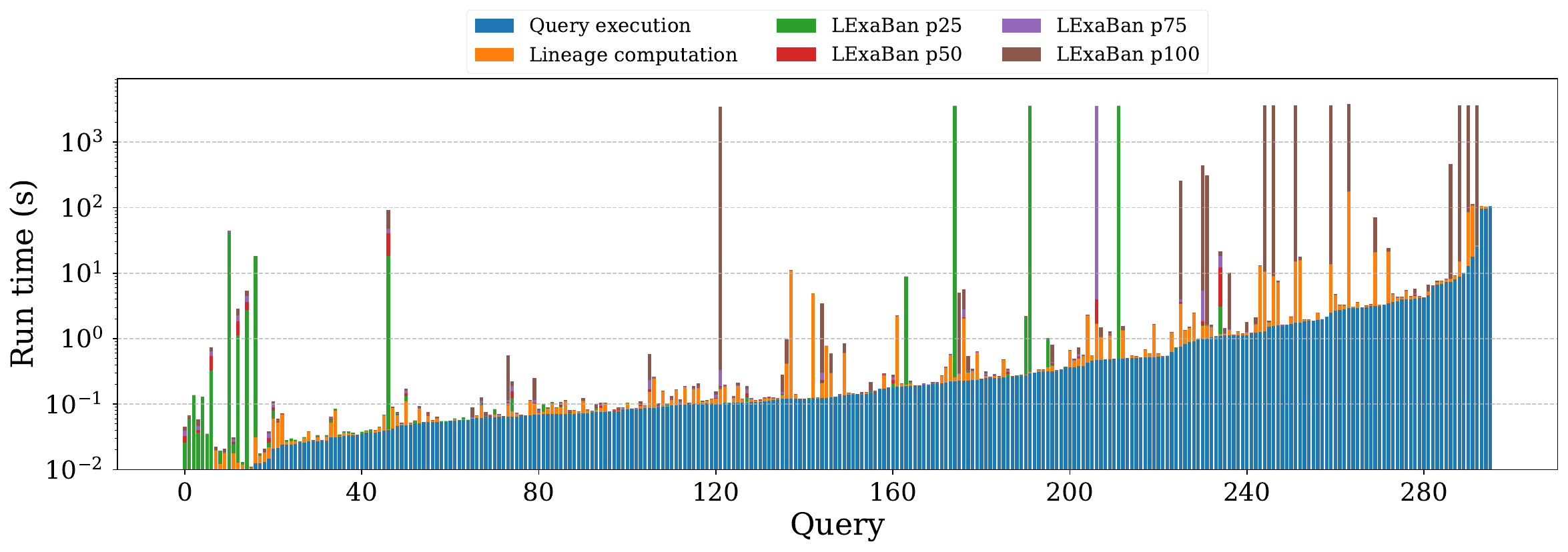}
    \caption{Breakdown of runtime for all queries sorted by the time to compute the output tuples (query execution time) using \lexaban.}
    \label{fig:runtime_breakdown_queries}
    
    \vspace{1em}

    \includegraphics[width=\linewidth]{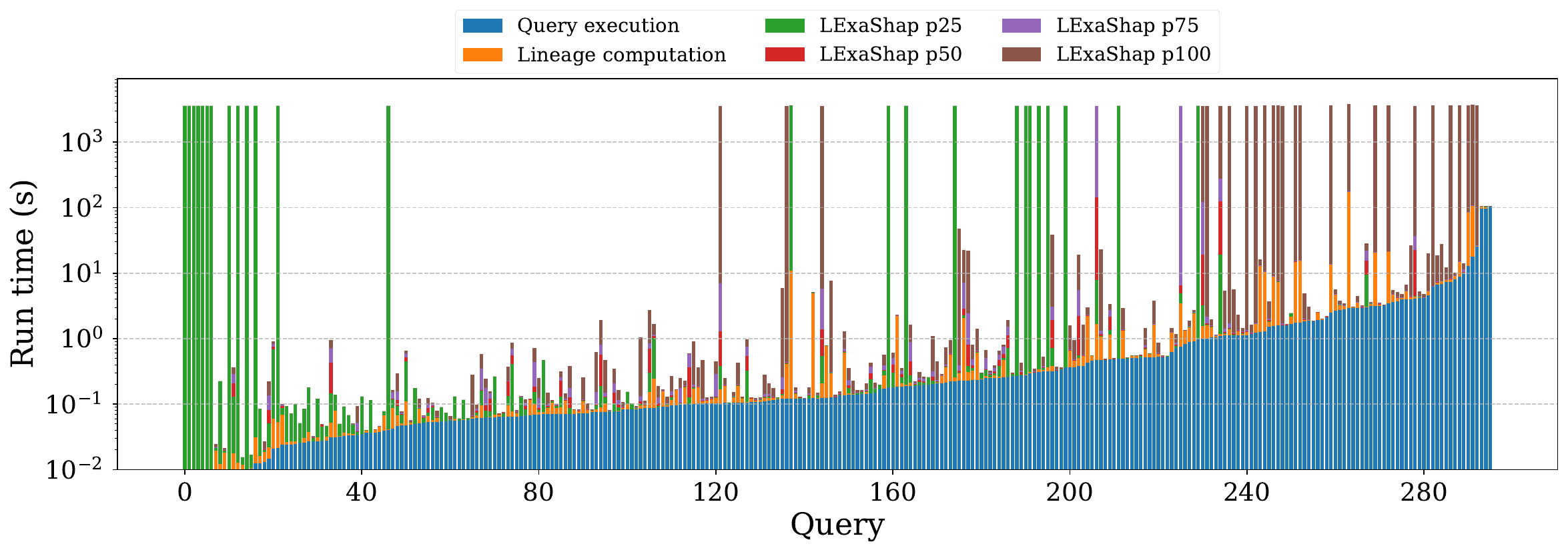}
    \caption{Breakdown of runtime for all queries sorted by the time to compute the output tuples (query execution time) using \lexashap.}
    \label{fig:runtime_breakdown_queries_shap}
\end{figure}

%% file: Sections_Arxiv_New/Related.tex
\section{Related Work}

\nop{Banzhaf and Shapley values were originally introduced in the context of game theory to quantify the contribution of individual players to the value of a cooperative game \cite{shapley1953value,Banzhaf:1965}.}

Recent work laid the theoretical foundation of attribution via Banzhaf and Shapley values in query answering~\cite{TheShapleyValueofTuplesinQueryAnswering,ComputingTheShapleyValueOfFactsInQueryAnswering, TheimpactofnegationonthecomplexityoftheShapleyvalueinconjunctivequeries,Sig24:ExpectedShapleyLikeScores,ShapleyValueInDataManagement,FromShapleyValueToModelCountingAndBack}. All prior practical approaches to computing such values rely on compiling the lineage into tractable circuits such as d-DNNFs~\cite{ComputingTheShapleyValueOfFactsInQueryAnswering} or d-trees~\cite{Sig24:BanzhafValuesForFactsInQueryAnswering}. In this paper, we follow the latter approach and non-trivially extend it along two dimensions.
First, we accommodate aggregate queries. This is challenging for two reasons: (1) the lineage of aggregate queries has a more complex form than that of SPJU queries studied in all prior work, and (2) attribution computation requires new machinery, accounting for the interaction between aggregate functions and Shapley/Banzhaf values. Ours is the first practical approach for aggregate queries: \cite{TheShapleyValueofTuplesinQueryAnswering} includes a theoretical investigation of aggregate queries but no practically efficient algorithm, whereas \cite{Sig24:BanzhafValuesForFactsInQueryAnswering,ComputingTheShapleyValueOfFactsInQueryAnswering} do not support aggregates.

Second, our approach consistently and significantly outperforms the state-of-the-art, by up to 3 (2) orders of magnitude for Banzhaf (Shapley). This speedup enables the computation of attribution for instances that prior work cannot handle (over 90\% of these difficult instances for Banzhaf and over 75\% for Shapley). The novel techniques of {\em lifting} and {\em gradient-based computation} introduced in this paper are the reason for this speedup and enable the processing of expensive SPJU  queries well beyond the reach of prior work, and of aggregate queries. In particular, lifting improves compilation time by over 2 orders of magnitude; gradient-based computation yields speedups across all instance sizes, with improvements exceeding 2 orders of magnitude for the largest instances.

Beyond Shapley and Banzhaf, further notions have been used to quantify fact contribution in query answering: causality \cite{VLDB2014:Causality_and_explanations_in_databases}, responsibility \cite{Thecomplexityofcausalityandresponsibilityforqueryanswersandnon-answers}, and counterfactuals \cite{BringingProvenancetoItsFullPotentialUsingCausalReasoning}. The SHAP score \cite{lundberg2017unified} leverages and adapts the definition of Shapley values to explain model predictions. Extending our techniques to further forms of attribution is an intriguing challenge for future work.

%% file: Sections_Arxiv_New/Conclusions.tex
\section{Conclusions}
We have introduced a novel approach for computing attribution for query answering.
Our approach is more general, as it supports aggregate queries, and faster by several orders of magnitude than the state of the art. Future work includes extending our approach to broader query classes, e.g., queries with AVG aggregation and negation, and further attribution measures, e.g., the SHAP score.

%% file: Sections_Arxiv_New/App_Algorithm_Arxiv.tex
\section{Missing Details in Section \ref{sec: Algorithm}}

We provide a full proposition for the Shapley analogue for Proposition \ref{Prop: Banzhaf_is_Gradient}. 

Recall the definition of the $k$-probability function. For a given d-tree $T$ and an integer $k\in \mathbb{N}$
 \begin{align*}
    Pr_k[T] =& \sum_{S\subseteq vars(T), \mid S \mid = k}\sum_{S'\subseteq S} T[S']\\&\cdot \prod_{y\in S'}p_y \prod_{z\in S\setminus S'}(\frac{1}{2}-p_z) \prod_{w\in vars(T)\setminus S} \frac{1}{2} 
\end{align*}

\begin{proposition}
\label{Prop: Shapley_is_Gradient}
Given a d-tree $T$ and a variable $x\in \vars(T)$, the following holds:
\begin{align*}
    \shap(T, x) =& 2^{|\vars(T)| - 1} \cdot \sum_{k \in [|\vars(T)|]} C_{k-1} \cdot \left(\frac{\partial Pr_k[T]}{\partial p_x}(\vec{\frac{1}{2}})\right)
\end{align*} 
 where $C_k = \frac{|k!|\cdot \mid \mid \vars(T) \mid - k - 1|!}{\mid \vars(T) \mid!}$ is the shapely coefficient.
\end{proposition}

\begin{proof}
Given a variable $x\in \vars(T)$, we separately consider for each valuation $S$ and subset $S'$ the cases 
where $x$ appears in $S'$, in $S\setminus S'$ or in $vars(T)\setminus S$, and obtain:
 
\begin{align*}
&Pr_k[T] = \sum_{S\subseteq vars(T), \mid S \mid = k, x\in S}\sum_{S'\subseteq S\setminus \{x\}} \big(T[S'](\frac{1}{2}-p_x) + T[S'\cup \{x\}]p_x\big)\\&\cdot \prod_{y\in S'}p_y \prod_{z\in S\setminus S'\cup \{x\}}(\frac{1}{2}-p_z) \prod_{w\in vars(T)\setminus S} \frac{1}{2}\\ & + \sum_{S\subseteq vars(T)\setminus \{x\}, \mid S \mid = k}\sum_{S'\subseteq S} T[S']\cdot \prod_{y\in S'}p_y \prod_{z\in S\setminus S'}(\frac{1}{2}-p_z) \prod_{w\in vars(T)\setminus S} \frac{1}{2}
\end{align*}

By taking the partial derivative of the above function 
with respect to $p_x$ we get:

\begin{align*}
    &\sum_{S\subseteq vars(T), \mid S \mid = k, x\in S}\sum_{S'\subseteq S\setminus \{x\}} \big(T[S'\cup \{x\}]- T[S']\big)\\&\cdot \prod_{y\in S'}p_y \prod_{z\in S\setminus S'\cup \{x\}}(\frac{1}{2}-p_z) \prod_{w\in vars(T)\setminus S} \frac{1}{2}
\end{align*}

Evaluating the partial derivative at the point $\vec{\frac{1}{2}}$ we obtain:
\begin{align*}
    &\sum_{S\subseteq vars(T), \mid S \mid = k, x\in S} \big(T[S]- T[S\setminus \{x\}]\big)\\&\cdot \prod_{y\in S}\frac{1}{2} \prod_{w\in vars(T)\setminus S} \frac{1}{2} \\=& (\frac{1}{2})^{\mid vars(T)\mid -1} \sum_{S\subseteq vars(T)\setminus{x}, \mid S \mid = k -1} \big(T[S\cup \{x\}]- T[S]\big)
\end{align*}

Where the first equality holds from taking only the subsets $S'$ without multiplications of $1-p_y$ since those would be equal to zero, and the second equality holds by reindexing the summation.
Finally, we can sum up the partial derivatives multiplied by the relevant constants and get:

\begin{align*}
    &\sum_{k \in [|\vars(T)|]} C_{k-1} \cdot \left(\frac{\partial Pr_k[T]}{\partial p_x}(\vec{\frac{1}{2}})\right) =
    (\frac{1}{2})^{\mid \vars(T)\mid -1}\\ \cdot& \sum_{k \in [|\vars(T)|]} C_{k-1}\cdot \sum_{S\subseteq vars(T)\setminus{x}, \mid S \mid = k -1} \big(T[S\cup \{x\}]- T[S]\big)\\
    &= \sum_{S\subseteq vars(T)\setminus{x}} C_{|S|} \big(T[S\cup \{x\}]- T[S]\big) =(\frac{1}{2})^{\mid \vars(T)\mid -1}\shap(T,x)\\
\end{align*}
\end{proof}









We provide here the analogous equations to Table \ref{tab:probability_and_gradients_for_gates} that are necessary for the computations of Shapley values.

\begin{table*}[h!]
\caption{Equations defining the probability and the partial derivatives for different gates. $T$ is the d-tree rooted at the gate and the $T_i$'s are its child sub-trees. Here, $T_{-T_i}$ represents the d-tree $T$ when excluding the sub-tree $T_i$, and $A_n^m = \frac{1}{2^n}\cdot\binom{n}{m}$ represents the portion of assignments of n variables of size $m$.}

    \centering
    \renewcommand{\arraystretch}{3}
    \begin{tabular}{|c|l|l|}
        \hline
        Gate & Probability Expression $Pr_k[T]$ & Partial Derivative \\ \hline
        $\oplus$ & $A^k_{\mid \vars(T)\mid} -\sum_{j_1,\dots,j_n, \sum j_1,\dots j_n = k} \prod_{i\in [n]} A^{j_i}_{\mid \vars(T_i)\mid} - Pr_{j_i}[T_i]$ 
        & $A^{k-j}_{\mid\vars(T_{-T_i})\mid} - Pr_{k-j}[T_{-T_i}]$
        \\ \hline
        $\odot$ & $\sum_{j_1\dots,j_n, \sum j_1\dots,j_n = k, \forall_i.j_i>0} \prod_{i\in [n]} Pr_{j_i}[T_i]$ 
        & $Pr_{k-j}[T_{-T_i}]$
        \\ \hline
        \multirow{3}{*}{$\sqcup_f$}  
         & \multirow{3}{*}{$\sum_{j\in[k]}Pr_j[T_1]\cdot Pr_{k-j}[f] + Pr_j[T_0]\cdot (A^{k-j} - Pr_{k-j}[f])$}    
        & $\frac{\partial Pr_k[T]}{\partial Pr_j[T_1]} = Pr_{k-j}[f]$ \\ 
        & & $\frac{\partial Pr_k[T]}{\partial Pr_j[T_0]} = (A^{k-j} - Pr_{k-k}[f])$ \\ 
        & & $\frac{\partial Pr_k[T]}{\partial Pr_j[f]} = Pr_{k-j}[T_1] - Pr_{k-j}[T_0]$ \\ \hline
    \end{tabular}
    \label{tab:Shapley_probability_and_gradients_gates}

\end{table*}

\nop{
Gate: $\odot$:
$$Pr_k[T] = \sum_{\sum j_1,j_2,...j_n = k, \forall_i.j_i>0} \prod_{i\in [n]} Pr_{j_i}[\varphi_i]$$

$$\frac{\partial Pr_k[T]}{\partial Pr_j[T_i]} = Pr_{k-j}[T_{-T_i}]$$
Where $T_{-T_i}$ represents the formula derived by the gate without the subformula $T_i$

Gate: $\oplus$:
$$Pr_k[T]= A^k_{\sum_{i\in [n]} \mid \vars(\varphi_i)\mid} -\sum_{\sum j_1,j_2,...j_n = k} \prod_{i\in [n]} A^{j_i}_{\mid \vars(\varphi_i)\mid} - Pr_{j_i}[\varphi_i]$$

Where $A^m_l = \frac{\binom{l}{m}}{2^l}$ is the portion of assignments of $l$ variables with a size of $m$. 

$$\frac{\partial Pr_k[T]}{\partial Pr_j[T_i]} = A^{k-j}_{\mid\vars(T_{-T_i})\mid} - Pr_{k-j}[T_{-T_i}]$$
Where $T_{-T_i}$ represents the formula derived by the gate without the subformula $T_i$. 

Gate: $\sqcup_f$:

$$Pr_k[T] = \sum_{j\in[k]}Pr_j[T_1]\cdot Pr_{k-j}[f] + Pr_j[T_0]\cdot (A^{k-j} - Pr_{k-k}[f])$$

$$\frac{\partial Pr_k[T]}{\partial Pr_j[T_1]} = Pr_{k-j}[f]$$

$$\frac{\partial Pr_k[T]}{\partial Pr_j[T_0]} = (A^{k-j} - Pr_{k-k}[f])$$

$$\frac{\partial Pr_k[T]}{\partial Pr_j[f]} = Pr_{k-j}[T_1] - Pr_{k-j}[T_0]$$
}

%% file: main_arxiv.bbl

\begin{thebibliography}{32}


\ifx \showCODEN    \undefined \def \showCODEN     #1{\unskip}     \fi
\ifx \showDOI      \undefined \def \showDOI       #1{#1}\fi
\ifx \showISBNx    \undefined \def \showISBNx     #1{\unskip}     \fi
\ifx \showISBNxiii \undefined \def \showISBNxiii  #1{\unskip}     \fi
\ifx \showISSN     \undefined \def \showISSN      #1{\unskip}     \fi
\ifx \showLCCN     \undefined \def \showLCCN      #1{\unskip}     \fi
\ifx \shownote     \undefined \def \shownote      #1{#1}          \fi
\ifx \showarticletitle \undefined \def \showarticletitle #1{#1}   \fi
\ifx \showURL      \undefined \def \showURL       {\relax}        \fi
\providecommand\bibfield[2]{#2}
\providecommand\bibinfo[2]{#2}
\providecommand\natexlab[1]{#1}
\providecommand\showeprint[2][]{arXiv:#2}

\bibitem[\protect\citeauthoryear{Abiteboul, Hull, and Vianu}{Abiteboul et~al\mbox{.}}{1995}]%
        {abiteboul1995foundations}
\bibfield{author}{\bibinfo{person}{Serge Abiteboul}, \bibinfo{person}{Richard Hull}, {and} \bibinfo{person}{Victor Vianu}.} \bibinfo{year}{1995}\natexlab{}.
\newblock \bibinfo{booktitle}{\emph{Foundations of Databases}}. Vol.~\bibinfo{volume}{8}.
\newblock \bibinfo{publisher}{Addison-Wesley Reading}.
\newblock
\urldef\tempurl%
\url{http://webdam.inria.fr/Alice/}
\showURL{%
\tempurl}


\bibitem[\protect\citeauthoryear{Abramovich, Deutch, Frost, Kara, and Olteanu}{Abramovich et~al\mbox{.}}{2024}]%
        {Sig24:BanzhafValuesForFactsInQueryAnswering}
\bibfield{author}{\bibinfo{person}{Omer Abramovich}, \bibinfo{person}{Daniel Deutch}, \bibinfo{person}{Nave Frost}, \bibinfo{person}{Ahmet Kara}, {and} \bibinfo{person}{Dan Olteanu}.} \bibinfo{year}{2024}\natexlab{}.
\newblock \showarticletitle{Banzhaf Values for Facts in Query Answering}.
\newblock \bibinfo{journal}{\emph{Proc. ACM Manag. Data}} \bibinfo{volume}{2}, \bibinfo{number}{3}, Article \bibinfo{articleno}{123} (\bibinfo{year}{2024}), \bibinfo{numpages}{26}~pages.
\newblock
\urldef\tempurl%
\url{https://doi.org/10.1145/3654926}
\showDOI{\tempurl}


\bibitem[\protect\citeauthoryear{Abramovich, Deutch, Frost, Kara, and Olteanu}{Abramovich et~al\mbox{.}}{2025}]%
        {full-version-git}
\bibfield{author}{\bibinfo{person}{Omer Abramovich}, \bibinfo{person}{Daniel Deutch}, \bibinfo{person}{Nave Frost}, \bibinfo{person}{Ahmet Kara}, {and} \bibinfo{person}{Dan Olteanu}.} \bibinfo{year}{2025}\natexlab{}.
\newblock \bibinfo{booktitle}{\emph{Advancing Fact Attribution: Aggregate Queries and Novel Algorithms (Extended Version)}}.
\newblock \bibinfo{type}{{T}echnical {R}eport}.
\newblock
\newblock
\shownote{Extended version of the paper, available at \url{https://github.com/Omer-Abramovich/LExaBan-LExaShap}.}


\bibitem[\protect\citeauthoryear{Amsterdamer, Deutch, and Tannen}{Amsterdamer et~al\mbox{.}}{2011}]%
        {ProvenanceForAggregateQueries}
\bibfield{author}{\bibinfo{person}{Yael Amsterdamer}, \bibinfo{person}{Daniel Deutch}, {and} \bibinfo{person}{Val Tannen}.} \bibinfo{year}{2011}\natexlab{}.
\newblock \showarticletitle{Provenance for Aggregate Queries}. In \bibinfo{booktitle}{\emph{PODS}}. \bibinfo{pages}{153–164}.
\newblock
\showISBNx{9781450306607}
\urldef\tempurl%
\url{https://doi.org/10.1145/1989284.1989302}
\showDOI{\tempurl}


\bibitem[\protect\citeauthoryear{Arab, Feng, Glavic, Lee, Niu, and Zeng}{Arab et~al\mbox{.}}{2018}]%
        {AF18}
\bibfield{author}{\bibinfo{person}{Bahareh~Sadat Arab}, \bibinfo{person}{Su Feng}, \bibinfo{person}{Boris Glavic}, \bibinfo{person}{Seokki Lee}, \bibinfo{person}{Xing Niu}, {and} \bibinfo{person}{Qitian Zeng}.} \bibinfo{year}{2018}\natexlab{}.
\newblock \showarticletitle{GProM - {A} Swiss Army Knife for Your Provenance Needs}.
\newblock \bibinfo{journal}{\emph{{IEEE} Data Eng. Bull.}} \bibinfo{volume}{41}, \bibinfo{number}{1} (\bibinfo{year}{2018}), \bibinfo{pages}{51--62}.
\newblock
\urldef\tempurl%
\url{http://sites.computer.org/debull/A18mar/p51.pdf}
\showURL{%
\tempurl}


\bibitem[\protect\citeauthoryear{Banzhaf~III}{Banzhaf~III}{1965}]%
        {Banzhaf:1965}
\bibfield{author}{\bibinfo{person}{John~F Banzhaf~III}.} \bibinfo{year}{1965}\natexlab{}.
\newblock \showarticletitle{Weighted Voting Doesn't Work: A Mathematical Analysis}.
\newblock \bibinfo{journal}{\emph{Rutgers Law Review}} \bibinfo{volume}{19}, \bibinfo{number}{2} (\bibinfo{year}{1965}), \bibinfo{pages}{317--343}.
\newblock
\urldef\tempurl%
\url{https://heinonline.org/HOL/LandingPage?handle=hein.journals/rutlr19&div=19&id=&page=}
\showURL{%
\tempurl}


\bibitem[\protect\citeauthoryear{Baur and Strassen}{Baur and Strassen}{1983}]%
        {BAUR1983317}
\bibfield{author}{\bibinfo{person}{Walter Baur} {and} \bibinfo{person}{Volker Strassen}.} \bibinfo{year}{1983}\natexlab{}.
\newblock \showarticletitle{The Complexity of Partial Derivatives}.
\newblock \bibinfo{journal}{\emph{Theoretical Computer Science}} \bibinfo{volume}{22}, \bibinfo{number}{3} (\bibinfo{year}{1983}), \bibinfo{pages}{317--330}.
\newblock
\showISSN{0304-3975}
\urldef\tempurl%
\url{https://doi.org/10.1016/0304-3975(83)90110-X}
\showDOI{\tempurl}


\bibitem[\protect\citeauthoryear{Bertossi, Kimelfeld, Livshits, and Monet}{Bertossi et~al\mbox{.}}{2023}]%
        {ShapleyValueInDataManagement}
\bibfield{author}{\bibinfo{person}{Leopoldo Bertossi}, \bibinfo{person}{Benny Kimelfeld}, \bibinfo{person}{Ester Livshits}, {and} \bibinfo{person}{Mika\"{e}l Monet}.} \bibinfo{year}{2023}\natexlab{}.
\newblock \showarticletitle{The Shapley Value in Database Management}.
\newblock \bibinfo{journal}{\emph{ACM SIGMOD Rec.}} \bibinfo{volume}{52}, \bibinfo{number}{2} (\bibinfo{year}{2023}), \bibinfo{pages}{6–17}.
\newblock
\showISSN{0163-5808}
\urldef\tempurl%
\url{https://doi.org/10.1145/3615952.3615954}
\showDOI{\tempurl}


\bibitem[\protect\citeauthoryear{Bienvenu, Figueira, and Lafourcade}{Bienvenu et~al\mbox{.}}{2024}]%
        {ShapleyvaluecomputationinontologyMediatedqueryanswering}
\bibfield{author}{\bibinfo{person}{Meghyn Bienvenu}, \bibinfo{person}{Diego Figueira}, {and} \bibinfo{person}{Pierre Lafourcade}.} \bibinfo{year}{2024}\natexlab{}.
\newblock \showarticletitle{Shapley Value Computation in Ontology-Mediated Query Answering}. In \bibinfo{booktitle}{\emph{KR}}. Article \bibinfo{articleno}{15}, \bibinfo{numpages}{11}~pages.
\newblock
\showISBNx{978-1-956792-05-8}
\urldef\tempurl%
\url{https://doi.org/10.24963/kr.2024/15}
\showDOI{\tempurl}


\bibitem[\protect\citeauthoryear{Cheney, Chiticariu, and Tan}{Cheney et~al\mbox{.}}{2009}]%
        {ProvenanceInDatabases:WhyHowandWhere}
\bibfield{author}{\bibinfo{person}{James Cheney}, \bibinfo{person}{Laura Chiticariu}, {and} \bibinfo{person}{Wang-Chiew Tan}.} \bibinfo{year}{2009}\natexlab{}.
\newblock \showarticletitle{Provenance in Databases: Why, How, and Where}.
\newblock \bibinfo{journal}{\emph{Found. Trends Datab.}} \bibinfo{volume}{1}, \bibinfo{number}{4} (\bibinfo{year}{2009}), \bibinfo{pages}{379–474}.
\newblock
\showISSN{1931-7883}
\urldef\tempurl%
\url{https://doi.org/10.1561/1900000006}
\showDOI{\tempurl}


\bibitem[\protect\citeauthoryear{Darwiche}{Darwiche}{2003}]%
        {ADiffrentialApproachToInferenceInBayesianNetworks}
\bibfield{author}{\bibinfo{person}{Adnan Darwiche}.} \bibinfo{year}{2003}\natexlab{}.
\newblock \showarticletitle{A Differential Approach to Inference in Bayesian Networks}.
\newblock \bibinfo{journal}{\emph{J. ACM}} \bibinfo{volume}{50}, \bibinfo{number}{3} (\bibinfo{year}{2003}), \bibinfo{pages}{280–305}.
\newblock
\showISSN{0004-5411}
\urldef\tempurl%
\url{https://doi.org/10.1145/765568.765570}
\showDOI{\tempurl}


\bibitem[\protect\citeauthoryear{Deutch, Frost, Kimelfeld, and Monet}{Deutch et~al\mbox{.}}{2022}]%
        {ComputingTheShapleyValueOfFactsInQueryAnswering}
\bibfield{author}{\bibinfo{person}{Daniel Deutch}, \bibinfo{person}{Nave Frost}, \bibinfo{person}{Benny Kimelfeld}, {and} \bibinfo{person}{Mika\"{e}l Monet}.} \bibinfo{year}{2022}\natexlab{}.
\newblock \showarticletitle{Computing the Shapley Value of Facts in Query Answering}. In \bibinfo{booktitle}{\emph{SIGMOD}}. \bibinfo{pages}{1570–1583}.
\newblock
\urldef\tempurl%
\url{https://doi.org/10.1145/3514221.3517912}
\showDOI{\tempurl}


\bibitem[\protect\citeauthoryear{Fink, Han, and Olteanu}{Fink et~al\mbox{.}}{2012}]%
        {AggregationInProbabalisticDatabasesViaKnowlegeCompilation}
\bibfield{author}{\bibinfo{person}{Robert Fink}, \bibinfo{person}{Larisa Han}, {and} \bibinfo{person}{Dan Olteanu}.} \bibinfo{year}{2012}\natexlab{}.
\newblock \showarticletitle{Aggregation in Probabilistic Databases via Knowledge Compilation}.
\newblock \bibinfo{journal}{\emph{PVLDB}} \bibinfo{volume}{5}, \bibinfo{number}{5} (\bibinfo{date}{Jan.} \bibinfo{year}{2012}), \bibinfo{pages}{490–501}.
\newblock
\showISSN{2150-8097}
\urldef\tempurl%
\url{https://doi.org/10.14778/2140436.2140445}
\showDOI{\tempurl}


\bibitem[\protect\citeauthoryear{Fryer, Str{\"{u}}mke, and Nguyen}{Fryer et~al\mbox{.}}{2021}]%
        {fryer2021shapley}
\bibfield{author}{\bibinfo{person}{Daniel~Vidali Fryer}, \bibinfo{person}{Inga Str{\"{u}}mke}, {and} \bibinfo{person}{Hien~D. Nguyen}.} \bibinfo{year}{2021}\natexlab{}.
\newblock \showarticletitle{Shapley Values for Feature Selection: The Good, the Bad, and the Axioms}.
\newblock \bibinfo{journal}{\emph{{IEEE} Access}}  \bibinfo{volume}{9} (\bibinfo{year}{2021}), \bibinfo{pages}{144352--144360}.
\newblock
\urldef\tempurl%
\url{https://doi.org/10.1109/ACCESS.2021.3119110}
\showDOI{\tempurl}


\bibitem[\protect\citeauthoryear{Glavic, Meliou, and Roy}{Glavic et~al\mbox{.}}{2021}]%
        {TrendsInExplanations}
\bibfield{author}{\bibinfo{person}{Boris Glavic}, \bibinfo{person}{Alexandra Meliou}, {and} \bibinfo{person}{Sudeepa Roy}.} \bibinfo{year}{2021}\natexlab{}.
\newblock \showarticletitle{Trends in Explanations: Understanding and Debugging Data-Driven Systems}.
\newblock \bibinfo{journal}{\emph{Proc. VLDB Endow.}} \bibinfo{volume}{11}, \bibinfo{number}{3} (\bibinfo{year}{2021}), \bibinfo{pages}{226–318}.
\newblock
\showISSN{1931-7883}
\urldef\tempurl%
\url{https://doi.org/10.1561/1900000074}
\showDOI{\tempurl}


\bibitem[\protect\citeauthoryear{Hirsch}{Hirsch}{2005}]%
        {hindex_pnas}
\bibfield{author}{\bibinfo{person}{J.~E. Hirsch}.} \bibinfo{year}{2005}\natexlab{}.
\newblock \showarticletitle{An Index to Quantify an Individual's Scientific Research Output}.
\newblock \bibinfo{journal}{\emph{PNAS}} \bibinfo{volume}{102}, \bibinfo{number}{46} (\bibinfo{year}{2005}), \bibinfo{pages}{16569--16572}.
\newblock
\urldef\tempurl%
\url{https://doi.org/10.1073/pnas.0507655102}
\showDOI{\tempurl}


\bibitem[\protect\citeauthoryear{Kara, Olteanu, and Suciu}{Kara et~al\mbox{.}}{2024}]%
        {FromShapleyValueToModelCountingAndBack}
\bibfield{author}{\bibinfo{person}{Ahmet Kara}, \bibinfo{person}{Dan Olteanu}, {and} \bibinfo{person}{Dan Suciu}.} \bibinfo{year}{2024}\natexlab{}.
\newblock \showarticletitle{From Shapley Value to Model Counting and Back}.
\newblock \bibinfo{journal}{\emph{Proc. ACM Manag. Data}} \bibinfo{volume}{2}, \bibinfo{number}{2}, Article \bibinfo{articleno}{79} (\bibinfo{year}{2024}), \bibinfo{numpages}{23}~pages.
\newblock
\urldef\tempurl%
\url{https://doi.org/10.1145/3651142}
\showDOI{\tempurl}


\bibitem[\protect\citeauthoryear{Karmakar, Monet, Senellart, and Bressan}{Karmakar et~al\mbox{.}}{2024}]%
        {Sig24:ExpectedShapleyLikeScores}
\bibfield{author}{\bibinfo{person}{Pratik Karmakar}, \bibinfo{person}{Mika\"{e}l Monet}, \bibinfo{person}{Pierre Senellart}, {and} \bibinfo{person}{Stephane Bressan}.} \bibinfo{year}{2024}\natexlab{}.
\newblock \showarticletitle{Expected Shapley-Like Scores of Boolean Functions: Complexity and Applications to Probabilistic Databases}.
\newblock \bibinfo{journal}{\emph{Proc. ACM Manag. Data}} \bibinfo{volume}{2}, \bibinfo{number}{2}, Article \bibinfo{articleno}{92} (\bibinfo{year}{2024}), \bibinfo{numpages}{26}~pages.
\newblock
\urldef\tempurl%
\url{https://doi.org/10.1145/3651593}
\showDOI{\tempurl}


\bibitem[\protect\citeauthoryear{Khalil and Kimelfeld}{Khalil and Kimelfeld}{2023}]%
        {ThecomplexityoftheShapleyvalueforregularpathqueries}
\bibfield{author}{\bibinfo{person}{Majd Khalil} {and} \bibinfo{person}{Benny Kimelfeld}.} \bibinfo{year}{2023}\natexlab{}.
\newblock \showarticletitle{The Complexity of the Shapley Value for Regular Path Queries}. In \bibinfo{booktitle}{\emph{{ICDT}}}, Vol.~\bibinfo{volume}{255}. \bibinfo{pages}{11:1--11:19}.
\newblock
\urldef\tempurl%
\url{https://doi.org/10.4230/LIPICS.ICDT.2023.11}
\showDOI{\tempurl}


\bibitem[\protect\citeauthoryear{Livshits, Bertossi, Kimelfeld, and Sebag}{Livshits et~al\mbox{.}}{2021}]%
        {TheShapleyValueofTuplesinQueryAnswering}
\bibfield{author}{\bibinfo{person}{Ester Livshits}, \bibinfo{person}{Leopoldo Bertossi}, \bibinfo{person}{Benny Kimelfeld}, {and} \bibinfo{person}{Moshe Sebag}.} \bibinfo{year}{2021}\natexlab{}.
\newblock \showarticletitle{The Shapley Value of Tuples in Query Answering}.
\newblock \bibinfo{journal}{\emph{LMCS}}  \bibinfo{volume}{Volume 17, Issue 3}, Article \bibinfo{articleno}{22} (\bibinfo{date}{Sep} \bibinfo{year}{2021}).
\newblock
\showISSN{1860-5974}
\urldef\tempurl%
\url{https://doi.org/10.46298/lmcs-17(3:22)2021}
\showDOI{\tempurl}


\bibitem[\protect\citeauthoryear{Lundberg and Lee}{Lundberg and Lee}{2017}]%
        {lundberg2017unified}
\bibfield{author}{\bibinfo{person}{Scott~M Lundberg} {and} \bibinfo{person}{Su-In Lee}.} \bibinfo{year}{2017}\natexlab{}.
\newblock \showarticletitle{A Unified Approach to Interpreting Model Predictions}. In \bibinfo{booktitle}{\emph{NeurIPS}}. \bibinfo{pages}{4765--4774}.
\newblock
\urldef\tempurl%
\url{http://papers.nips.cc/paper/7062-a-unified-approach-to-interpreting-model-predictions.pdf}
\showURL{%
\tempurl}


\bibitem[\protect\citeauthoryear{Luo, Pei, Xu, Zhang, and Xu}{Luo et~al\mbox{.}}{2024}]%
        {FastShapleyValueComputationinDataAssemblage}
\bibfield{author}{\bibinfo{person}{Xuan Luo}, \bibinfo{person}{Jian Pei}, \bibinfo{person}{Cheng Xu}, \bibinfo{person}{Wenjie Zhang}, {and} \bibinfo{person}{Jianliang Xu}.} \bibinfo{year}{2024}\natexlab{}.
\newblock \showarticletitle{Fast Shapley Value Computation in Data Assemblage Tasks as Cooperative Simple Games}.
\newblock \bibinfo{journal}{\emph{Proc. ACM Manag. Data}} \bibinfo{volume}{2}, \bibinfo{number}{1}, Article \bibinfo{articleno}{56} (\bibinfo{year}{2024}), \bibinfo{numpages}{28}~pages.
\newblock
\urldef\tempurl%
\url{https://doi.org/10.1145/3639311}
\showDOI{\tempurl}


\bibitem[\protect\citeauthoryear{Meliou, Gatterbauer, Moore, and Suciu}{Meliou et~al\mbox{.}}{2010}]%
        {Thecomplexityofcausalityandresponsibilityforqueryanswersandnon-answers}
\bibfield{author}{\bibinfo{person}{Alexandra Meliou}, \bibinfo{person}{Wolfgang Gatterbauer}, \bibinfo{person}{Katherine~F. Moore}, {and} \bibinfo{person}{Dan Suciu}.} \bibinfo{year}{2010}\natexlab{}.
\newblock \showarticletitle{The Complexity of Causality and Responsibility for Query Answers and Non-Answers}.
\newblock \bibinfo{journal}{\emph{PVLDB}} \bibinfo{volume}{4}, \bibinfo{number}{1} (\bibinfo{year}{2010}), \bibinfo{pages}{34--45}.
\newblock
\urldef\tempurl%
\url{https://www.vldb.org/pvldb/vol4/p34-meliou.pdf}
\showURL{%
\tempurl}


\bibitem[\protect\citeauthoryear{Meliou, Gatterbauer, and Suciu}{Meliou et~al\mbox{.}}{2011}]%
        {BringingProvenancetoItsFullPotentialUsingCausalReasoning}
\bibfield{author}{\bibinfo{person}{Alexandra Meliou}, \bibinfo{person}{Wolfgang Gatterbauer}, {and} \bibinfo{person}{Dan Suciu}.} \bibinfo{year}{2011}\natexlab{}.
\newblock \showarticletitle{Bringing Provenance to Its Full Potential Using Causal Reasoning}. In \bibinfo{booktitle}{\emph{TaPP}}.
\newblock
\urldef\tempurl%
\url{https://www.usenix.org/conference/tapp11/bringing-provenance-its-full-potential-using-causal-reasoning}
\showURL{%
\tempurl}


\bibitem[\protect\citeauthoryear{Meliou, Roy, and Suciu}{Meliou et~al\mbox{.}}{2014}]%
        {VLDB2014:Causality_and_explanations_in_databases}
\bibfield{author}{\bibinfo{person}{Alexandra Meliou}, \bibinfo{person}{Sudeepa Roy}, {and} \bibinfo{person}{Dan Suciu}.} \bibinfo{year}{2014}\natexlab{}.
\newblock \showarticletitle{Causality and Explanations in Databases}.
\newblock \bibinfo{journal}{\emph{PVLDB}} \bibinfo{volume}{7}, \bibinfo{number}{13} (\bibinfo{date}{Aug.} \bibinfo{year}{2014}), \bibinfo{pages}{1715–1716}.
\newblock
\urldef\tempurl%
\url{https://doi.org/10.14778/2733004.2733070}
\showDOI{\tempurl}


\bibitem[\protect\citeauthoryear{Penrose}{Penrose}{1946}]%
        {Penrose:Banzhaf:1946}
\bibfield{author}{\bibinfo{person}{L.~S. Penrose}.} \bibinfo{year}{1946}\natexlab{}.
\newblock \showarticletitle{The Elementary Statistics of Majority Voting}.
\newblock \bibinfo{journal}{\emph{J. Royal Stats. Soc.}} \bibinfo{volume}{109}, \bibinfo{number}{1} (\bibinfo{year}{1946}), \bibinfo{pages}{53--57}.
\newblock
\showISSN{09528385}
\urldef\tempurl%
\url{http://www.jstor.org/stable/2981392}
\showURL{%
\tempurl}


\bibitem[\protect\citeauthoryear{Reshef, Kimelfeld, and Livshits}{Reshef et~al\mbox{.}}{2020}]%
        {TheimpactofnegationonthecomplexityoftheShapleyvalueinconjunctivequeries}
\bibfield{author}{\bibinfo{person}{Alon Reshef}, \bibinfo{person}{Benny Kimelfeld}, {and} \bibinfo{person}{Ester Livshits}.} \bibinfo{year}{2020}\natexlab{}.
\newblock \showarticletitle{The Impact of Negation on the Complexity of the Shapley Value in Conjunctive Queries}. In \bibinfo{booktitle}{\emph{{PODS}}}. \bibinfo{pages}{285--297}.
\newblock
\urldef\tempurl%
\url{https://doi.org/10.1145/3375395.3387664}
\showDOI{\tempurl}


\bibitem[\protect\citeauthoryear{Rozemberczki, Watson, Bayer, Yang, Kiss, Nilsson, and Sarkar}{Rozemberczki et~al\mbox{.}}{2022}]%
        {rozemberczki2022shapley}
\bibfield{author}{\bibinfo{person}{Benedek Rozemberczki}, \bibinfo{person}{Lauren Watson}, \bibinfo{person}{P{\'{e}}ter Bayer}, \bibinfo{person}{Hao{-}Tsung Yang}, \bibinfo{person}{Oliver Kiss}, \bibinfo{person}{Sebastian Nilsson}, {and} \bibinfo{person}{Rik Sarkar}.} \bibinfo{year}{2022}\natexlab{}.
\newblock \showarticletitle{The Shapley Value in Machine Learning}. In \bibinfo{booktitle}{\emph{{IJCAI}}}. \bibinfo{pages}{5572--5579}.
\newblock
\urldef\tempurl%
\url{https://doi.org/10.24963/IJCAI.2022/778}
\showDOI{\tempurl}


\bibitem[\protect\citeauthoryear{Senellart, Jachiet, Maniu, and Ramusat}{Senellart et~al\mbox{.}}{2018}]%
        {senellart2018provsql}
\bibfield{author}{\bibinfo{person}{Pierre Senellart}, \bibinfo{person}{Louis Jachiet}, \bibinfo{person}{Silviu Maniu}, {and} \bibinfo{person}{Yann Ramusat}.} \bibinfo{year}{2018}\natexlab{}.
\newblock \showarticletitle{ProvSQL: Provenance and Probability Management in PostgreSQL}.
\newblock \bibinfo{journal}{\emph{PVLDB}} \bibinfo{volume}{11}, \bibinfo{number}{12} (\bibinfo{year}{2018}), \bibinfo{pages}{2034--2037}.
\newblock
\urldef\tempurl%
\url{https://hal.inria.fr/hal-01851538/file/p976-senellart.pdf}
\showURL{%
\tempurl}


\bibitem[\protect\citeauthoryear{Shapley}{Shapley}{1953}]%
        {shapley1953value}
\bibfield{author}{\bibinfo{person}{Lloyd~S Shapley}.} \bibinfo{year}{1953}\natexlab{}.
\newblock \showarticletitle{A Value for n-Person Games}.
\newblock \bibinfo{journal}{\emph{Contributions to the Theory of Games}} \bibinfo{volume}{2}, \bibinfo{number}{28} (\bibinfo{year}{1953}), \bibinfo{pages}{307--317}.
\newblock
\urldef\tempurl%
\url{http://www.library.fa.ru/files/Roth2.pdf\#page=39}
\showURL{%
\tempurl}


\bibitem[\protect\citeauthoryear{Sun, Zhong, Huang, and Dong}{Sun et~al\mbox{.}}{2018}]%
        {sun2018banzhaf}
\bibfield{author}{\bibinfo{person}{Jianyuan Sun}, \bibinfo{person}{Guoqiang Zhong}, \bibinfo{person}{Kaizhu Huang}, {and} \bibinfo{person}{Junyu Dong}.} \bibinfo{year}{2018}\natexlab{}.
\newblock \showarticletitle{Banzhaf Random Forests: Cooperative Game Theory Based Random Forests with Consistency}.
\newblock \bibinfo{journal}{\emph{Neural Networks}}  \bibinfo{volume}{106} (\bibinfo{year}{2018}), \bibinfo{pages}{20--29}.
\newblock
\urldef\tempurl%
\url{https://doi.org/10.1016/J.NEUNET.2018.06.006}
\showDOI{\tempurl}


\bibitem[\protect\citeauthoryear{Wang and Jia}{Wang and Jia}{2023}]%
        {wang2023dataBanzhaf}
\bibfield{author}{\bibinfo{person}{Jiachen~T. Wang} {and} \bibinfo{person}{Ruoxi Jia}.} \bibinfo{year}{2023}\natexlab{}.
\newblock \showarticletitle{Data Banzhaf: {A} Robust Data Valuation Framework for Machine Learning}. In \bibinfo{booktitle}{\emph{AISTATS}}, Vol.~\bibinfo{volume}{206}. \bibinfo{pages}{6388--6421}.
\newblock
\urldef\tempurl%
\url{https://proceedings.mlr.press/v206/wang23e.html}
\showURL{%
\tempurl}


\end{thebibliography}
